\documentclass{article}

\usepackage[left=3cm, right=3cm, bottom=4cm]{geometry}

\usepackage{amsmath, amsthm, amssymb,amsbsy}
\usepackage{algcompatible}
\usepackage{algorithm}
\usepackage{algorithmicx}
\usepackage{algpseudocode}
\usepackage{graphicx}
\usepackage[colorlinks,citecolor=blue]{hyperref}
\usepackage{tikz}
\usetikzlibrary{calc,positioning, math, tikzmark}
\usepackage{tikzpagenodes}
\usepackage{everypage}

\usepackage{pgfplots}
\usepackage{bbm}
\usepackage{float}
\usepackage[shortlabels]{enumitem}
\usepackage{comment}
\usepackage{caption}
\usepackage{setspace}
\usepackage{tcolorbox}
\usepackage{booktabs}
\usepackage{marginnote}

\usepackage{thmtools,thm-restate}

\usepackage{glossaries}

\newif\ifmaybe
\maybefalse

\newcounter{EditBlock}
\newcounter{CondEditBlock}

\newif\ifeditmark\editmarktrue

\editmarkfalse

\ifeditmark

\newcommand{\editstart}{\renewcommand{\linenumberfont}{\normalfont\sffamily\footnotesize\bf\color{blue}}\stepcounter{EditBlock}\linelabel{L:editstart_\theEditBlock}}
\newcommand{\editfinish}{\renewcommand{\linenumberfont}{\normalfont\sffamily\tiny\color{black}}\linelabel{L:editfinish_\theEditBlock}}

\newcommand{\mathstart}{\renewcommand{\linenumberfont}{\normalfont\sffamily\footnotesize\bf\color{green}}}
\newcommand{\mathfinish}{\renewcommand{\linenumberfont}{\normalfont\sffamily\tiny\color{black}}}

\else 

\newcommand{\editstart}{}
\newcommand{\editfinish}{}

\newcommand{\mathstart}{}
\newcommand{\mathfinish}{}

\fi

\usepackage[caption=false,font=normalsize,labelfon
t=sf,textfont=sf]{subfig}

\algnewcommand\algorithmicinput{\textbf{Input:}}
\algnewcommand\INPUT{\item[\algorithmicinput]}
\algnewcommand\algorithmicoutput{\textbf{Output:}}
\algnewcommand\OUTPUT{\item[\algorithmicoutput]}

\algrenewcommand\algorithmicrequire{\textbf{Input:}}
\algrenewcommand\algorithmicensure{\textbf{Output:}}



\makeatletter
   \def\vhrulefill#1{\leavevmode\leaders\hrule\@height#1\hfill \kern\z@}
\makeatother
   
\newcounter{alg}
{
  \refstepcounter{alg}
  \noindent \rule{\textwidth}{1pt}
  \textbf{Algorithm \thealg.} #1 \\[-0.3em]
  \noindent  \rule{\textwidth}{0.5pt} %
}%
{%
  \rule{\textwidth}{0.5pt} %
}

\usepackage{framed}
\usepackage{tcolorbox}
\tcbuselibrary{breakable,skins}

\newtcolorbox{mybox}{                       
   enhanced,
   colframe=blue!50,
   boxrule=1pt,
   width=1.05\textwidth,
   arc=1mm,
   breakable,                               
}

\title{A Universal Lossless Compression Method applicable to Sparse Graphs and Heavy--Tailed Sparse Graphs}

\newcommand{\ev}[1]{\mathbb{E} \left [ #1 \right ] }
\newcommand{\evwrt}[2]{\mathbb{E}_{#1} \left [ #2 \right ] }
\newcommand{\pr}[1]{\mathbb{P} \left ( #1 \right ) }
\newcommand{\prwrt}[2]{\mathbb{P}_{#1} \left ( #2 \right ) }
\newcommand{\norm}[1]{\left \Vert #1 \right \Vert}
\newcommand{\snorm}[1]{\Vert #1 \Vert}
\newcommand{\one}[1]{\mathbbm{1} \left [ #1 \right ]}

\newtheorem{lem}{Lemma}
\newtheorem*{lem*}{Lemma}
\newtheorem{thm}{Theorem}
\newtheorem{definition}{Definition}
\newtheorem{prop}{Proposition}

\newtheorem{rem}{Remark}

\newcommand{\mG}{\mathcal{G}}

\newcommand{\an}{a^{(n)}}
\newcommand{\tG}{\widetilde{G}}

\newcommand{\mn}{m^{(n)}}
\newcommand{\ms}{m^{*}}
\newcommand{\eln}{\ell^{(n)}}
\newcommand{\tmn}{\widetilde{m}^{(n)}}

\newcommand{\mA}{\mathcal{A}}

\newcommand{\mP}{\mathcal{P}}

\newcommand{\mS}{\mathcal{S}}

\newcommand{\mF}{\mathcal{F}}

\newcommand{\mT}{\mathcal{T}}

\newcommand{\tf}{\tilde{f}}

\newcommand{\Gn}{G^{(n)}}

\newcommand{\tGn}{\widetilde{G}^{(n)}}

\newcommand{\barmn}{\overline{m}_n}
\newcommand{\barm}{\overline{m}}
\newcommand{\barMn}{\overline{M}_n}

\newcommand{\fn}{f^{(n)}}
\newcommand{\Tn}{T^{(n)}}

\newcommand{\fnDn}{f_n^{\Delta_n}}
\newcommand{\gnDn}{g_n^{\Delta_n}}
\newcommand{\fnlwc}{f_n^{\mathsf{lwc}}}
\newcommand{\gnlwc}{g_n^{\mathsf{lwc}}}

\newcommand{\len}{\mathsf{bits}}
\newcommand{\nat}{\mathsf{nats}}

\newcommand{\tX}{\widetilde{X}}

\newcommand{\reals}{\mathbb{R}}
\newcommand{\integers}{\mathbb{Z}}
\newcommand{\nats}{\mathbb{N}}

\newcommand{\ER}{Erd\H{o}s--R\'{e}nyi }
\newcommand{\LP}{L\'{e}vy--Prokhorov }

\newcommand{\dlp}{d_\text{LP}} 
\newcommand{\bch}{ \Sigma} 
\newcommand{\bchover}{\overline{\Sigma}}
\newcommand{\bchunder}{\underbar{$\Sigma$}}

\DeclareMathOperator*{\argmin}{arg\,min}

\let\oldmarginpar\marginpar
\renewcommand{\marginpar}[2][rectangle,draw,rounded corners,text width = 3cm, scale=0.7]{%
  \oldmarginpar{%
    \tikz \node at (0,0) [#1]{#2};}%
}

\newcommand{\muk}{\mu^{(k)}}

\DeclareMathOperator{\Ent}{Ent} 
\newcommand{\pp}[1]{#1'}

\newcommand{\elln}{\ell^{(n)}}

\newcommand{\hpin}{\hat{\pi}_n} 
\newcommand{\tpin}{\tilde{\pi}_n} 

\newcommand{\hW}{\widehat{W}}
\newcommand{\hWn}{\widehat{W}^{(n)}}
\newcommand{\hWns}{\widehat{W}^{(n)}_*}

\newcommand{\gran}{a_n}

\allowdisplaybreaks 

\definecolor{bluenodecolor}{RGB}{62,126,176}
\definecolor{rednodecolor}{RGB}{173,61,58}

\definecolor{blueedgecolor}{RGB}{3,151,255}
\definecolor{orangeedgecolor}{RGB}{255,149,8}

\tikzstyle{nodeB} = [fill=bluenodecolor, circle, inner sep = 1.7pt]
\tikzstyle{nodeR} = [fill=rednodecolor, rectangle, inner sep = 2.3pt]
\tikzstyle{nodeBlack} = [fill=black, circle, inner sep = 1.7pt]

\tikzstyle{edgeB} = [very thick, blueedgecolor]
\tikzstyle{edgeO} = [very thick, orangeedgecolor, decoration = {zigzag,segment length = 0.2cm, amplitude = 0.5mm},decorate]

\newcommand{\nodelabel}[3]{\node at ($(#1)+(#2:5mm)$) {#3};}

\makeatletter
\newcommand{\checkmarkpage}[4]
{\@ifundefined{save@pt@#1}{#2}{%
  \edef\markid{\csname save@pt@#1\endcsname}%
  \edef\markpage{\csname save@pg@\markid\endcsname}%
  \ifnum\thepage<\markpage\relax #2%
  \else%
    \ifnum\thepage=\markpage\relax #3%
    \else #4%
    \fi%
  \fi}%
}
\makeatother

\newcounter{outlineid}
\newcounter{outlinedone}

  {\par\tikzmark{end\theoutlineid}\stepcounter{outlineid}\ignorespacesafterend}%

\newcommand{\drawoutline}{\checkmarkpage{begin\theoutlinedone}{}%
  {\begin{tikzpicture}[remember picture,overlay]
    \path ({pic cs:begin\theoutlinedone}-| current page text area.west)
      ++(0pt,\ht\strutbox) coordinate(A);
    \checkmarkpage{end\theoutlinedone}%
      {\path (current page text area.south west) ++(0pt,-\dp\strutbox)
         coordinate(B);}%
      {\path ({pic cs:end\theoutlinedone}-| current page text area.west)
        ++(0pt,\ht\strutbox) coordinate(B);}%
      {}
      \fill[yellow] ($(A) + (-.333em,0pt)$) rectangle ($(B) + (-1cm,0pt)$);
   \end{tikzpicture}}%
  {\begin{tikzpicture}[remember picture,overlay]
    \coordinate (A) at (current page text area.north west);
    \checkmarkpage{end\theoutlinedone}%
      {\path (current page text area.south west) ++(0pt,-\dp\strutbox)
         coordinate(B);}%
      {\path ({pic cs:end\theoutlinedone}-| current page text area.west)
        ++(0pt,\ht\strutbox) coordinate(B);}%
      {}
      \fill[yellow] ($(A) + (-.333em,0pt)$) rectangle ($(B) + (-1cm,0pt)$);
   \end{tikzpicture}}%
  \checkmarkpage{end\theoutlinedone}{}%
    {\stepcounter{outlinedone}\drawoutline}%
    {}
 }
\AddEverypageHook{\drawoutline}

\newif\ifdraft\draftfalse

\newcommand{\orange}{\color{orange}}
\newcommand{\magenta}{\color{magenta}}

\newcommand{\black}{\color{black}}

\definecolor{pcommentcolor}{RGB}{91,137,245}

\definecolor{pedit}{RGB}{0,153,0}
\newcommand{\pcomment}[1]{{\color{pcommentcolor}[Payam: #1]}}

\author{Payam Delgosha\thanks{Department of Computer Science, University of
    Illinois Urbana-Champaign, \texttt{delgosha@illinois.edu}}
   \, and   Venkat Anantharam\thanks{Department of Electrical Engineering and Computer
    Sciences, University of California, Berkeley, \texttt{ananth@berkeley.edu}}}

\begin{document}

\maketitle

\begin{abstract}
  Graphical data arises naturally in several modern applications, including but
  not limited to internet graphs, social networks, genomics and proteomics. 
  The typically large size of graphical data
  argues  for the importance of designing universal compression methods
  for such data. In most applications, the graphical data is sparse,
  meaning that the number of edges in the graph scales more slowly than $n^2$,
  where
  $n$ denotes the number of vertices. Although in some applications  the number
  of edges scales linearly with $n$, in others 
  the number of edges is much smaller than $n^2$ but appears to scale superlinearly with $n$.
  We call the former \emph{sparse graphs} and the latter  \emph{heavy-tailed sparse graphs}. In
  this paper we introduce a universal lossless compression method which is
  simultaneously applicable to both classes. We do this by employing the local
  weak convergence framework for sparse graphs and the sparse graphon framework for
  heavy-tailed sparse graphs. 
\end{abstract}

\section{Introduction}
\label{sec:intro}

\editstart

The sheer amount of graphical data in modern applications argues for finding efficient and optimal
methods of compressing 
such data for storage and
further data mining tasks.
Graphical data arises in 
social networks,
molecular and systems biology, and web graphs, as well as in several other application areas.
To be concrete, an
instance of graphical data arising in a web graph network would be a snapshot view of the network
at a given time. Each vertex in such a graph represents a web page, and an edge
represents a link between two web pages. 
An instance of graphical data in systems biology would be a protein-protein interaction network. Each vertex corresponds to a protein and an edge to an interaction between proteins.

Largely motivated by such applications,
there 
has 
recently been an increased interest 
in the problem of graphical
data compression. 
In existing works,
typically
assumptions are 
made regarding the 
properties 
of the graphical data of interest. 
One approach is
to design compression schemes for specific data sources such as web graphs or
social networks,
with the model for the
properties of the graphical data derived from a limited set of prior samples.
For instance, 
Boldi and Vigna have proposed the webgraph
framework to address the efficient compression of internet graphs
\cite{boldi2004webgraph},  Boldi et al. have proposed the layer label
propagation (LLP) method to compress social network graphs
\cite{boldi2011layered}, 
and Liakos et al.\ have proposed the BV+
compression method and evaluated its performance on certain datasets such as web
and social graphs \cite{liakos2014pushing}.
In this approach, the compression method is usually
based on some properties of the data which are  extracted based on observing
real-world samples. 
Therefore, such approaches usually do not come with 
information-theoretical guarantees of optimality.

The other approach in the
literature is to assume that the input data is generated through a certain
stochastic model, and the goal is to study the information content and
compression of such models by employing a notion of entropy.
Thus these works are less tied to a specific application.
For instance, Choi and Spankowski have studied the
structural entropy of the \ER model and compressing such graphs \cite{choi2012compression}, Aldous and Ross
have studied the asymptotic behavior of the entropy associated to some models of
sparse random graphs \cite{aldous2014entropy}, and Abbe has studied the
compression of stochastic block models \cite{abbe2016graph}.

In contrast to these prior works, we adopt the perspective of 
\emph{universal compression}. Namely, we 
study the compression of graphical data in a ``pointwise" sense, which is made more precise below.
In particular, we try to make as few assumptions as we can about the properties or statistical characteristics of the graphical data that we are trying to compress.

It is widely believed that
real world graphical data are ``sparse''.
Roughly speaking, 
a graph with $n$ vertices
is said to be sparse 
(in a broad sense)
\black
if its number of edges is \emph{much smaller} than $n^2$. This yields a whole 
spectrum 
of regimes 
under which one can study sparsity. One interesting  sparsity
regime is  where, roughly speaking, the number of edges is a
constant times the number of vertices
(more precisely,
\black
when the number of edges grows linearly
with the number of vertices
in an asymptotic regime).
\black
In recent works
\black
the authors of this paper have studied the problem of universal
lossless compression~\cite{delgosha2020universal} and
distributed compression \cite{delgosha2018distributed} for sparse graphs in  this sparsity regime
(the latter in a model-based framework).
\black%
This was done by
employing the 
notion of ``local weak convergence'', 
an instance of 
\black%
the so-called the
``objective method''
\cite{BenjaminiSchramm01rec,aldous2004objective,aldous2007processes},
which, roughly speaking, allows one to think of the graphical data as a sample from a 
limiting stochastic object derived from the
empirical characteristics of
the given sample (more precisely, this is done in an asymptotic setting, and the limiting stochastic object is a probability distribution on rooted graphs;
details are given in Section \ref{sec:prelim-lwc}).
\black%
Moreover, the authors have built upon the work of Bordenave and Caputo
\cite{bordenave2015large} to introduce a notion of entropy called the
\emph{marked BC entropy} which turns out to be the correct information-theoretic
measure of optimality 
on a per-edge basis
(which is the same as the per-vertex basis in this sparsity regime)
\black
for the purpose of the 
universal compression
of graphical data in this
formulation of 
the compression
problem 
\cite{delgosha2019notion}.
\black
Note that compression to the correct information-theoretic limit on a 
per-edge basis is a significantly deeper guarantee of information-theoretic
optimality than a crude guarantee of matching the growth rate of the overall
entropy of the graphical data, since the leading term in the overall entropy
depends only on the average degree of the graph and is on the scale of 
$n \log n$ where $n$ denotes the number of vertices;
see the details in Section \ref{sec:prelim}.
\black

The idea behind the local weak convergence framework is to study the asymptotic
behavior of the distribution of the neighborhood structure of a typical vertex in
the graph. This allows one to define a limit object associated to a sequence of
sparse graphs, the sparsity regime of interest being
where the number of edges grows linearly with the number of vertices. From the
point of view of the compression problem,
it is desirable to go beyond this sparsity regime and achieve universal
compression for graphs which are still sparse, but with the number of edges 
growing super--linearly with the number of vertices, i.e.\ sparse graphs with
\emph{heavy--tailed} degree distributions. 
Indeed, it is generally believed that 
heavy-tailed degree distributions are more representative of real world networks. 
Achieving universal compression in an information-theoretically optimal sense 
on a per-edge basis
\black
while being able to include heavy-tailed sparse graphical data in the framework 
\black%
is the purpose
\black
of this paper. 
More precisely, we build upon the universal compression scheme of
\cite{delgosha2020universal} to go beyond the local weak
convergence framework, and we design a universal compression scheme which is
capable of encoding graphs which are either consistent with the local weak
convergence framework or come from a specific class of sparse graphs with
 heavy--tailed degree distributions
 (and this has to be done while not knowing which regime the graphical data is from).
 \black%


In order to address graphs
with heavy--tailed degree distributions, we employ a version  of the graphon
theory adapted 
for sparse graphs \cite{borgs2019L,borgs2018L,borgs2015consistent}.
For dense graphs (the graphs where the number of edges scales as $n^2$), the theory of graphons allows one to make sense of a notion of
limit 
and provides 
a comprehensive framework to study the
asymptotic behavior 
(see,
for instance, \cite{lovasz2006limits}, \cite{lovasz2007szemeredi},
\cite{borgs2008convergent}, \cite{borgs2012convergent},
\cite{lovasz2012large}).
\black%
There has
been a recent effort to bridge the gap between the above sparse regime addressed
by local weak convergence, and the dense regime addressed by the
graphon theory (see, for instance, \cite{bollobas2007metrics}, \cite{borgs2019L},
\cite{borgs2018L}).  This 
framework, which we call the \emph{sparse graphon framework},  defines a notion of
convergence for heavy--tailed sparse graphs, similar to
the local weak convergence framework, but in a completely different metric.

Motivated by the above discussion, the local weak convergence framework and the
sparse graphon framework together yield a powerful machinery which is capable of
addressing sparsity in a broad range. In particular, we use this machinery to
address the problem of universal compression of sparse graphical data. More
precisely,  we aim to compress a graph which is either consistent
with the local weak convergence framework or the sparse graphon framework.
However, the universality condition requires that the encoder does not know which of the two frameworks the input graph
is consistent with, neither does it know the limiting object in each of the two
frameworks. However, we want the encoder to be information-theoretically
optimal, in the sense that if we appropriately normalize the codeword length
associated to the input graph, it does not asymptotically exceed the entropy of
the limit object
on a per-edge basis,
\black
with an appropriate notion of the entropy for each of the two frameworks. 
In order to make sense of optimality in the local weak
sense, we employ the notion of BC entropy from \cite{bordenave2015large} which
we discussed above. On the other hand, in
order to make sense of optimality in the sparse graphon sense, we introduce
a notion of entropy for this framework in Section~\ref{sec:graphon-entropy-new}, which can be of independent interest.

The structure of this paper is as follows. In Section~\ref{sec:prelim}, we
review local weak convergence, the BC entropy, sparse graphons,
and the universal lossless compression scheme introduced in
\cite{delgosha2020universal}. Then, in Section~\ref{sec:graphon-entropy-new}, we
introduce our notion of entropy for the sparse graphon framework. We then
rigorously define the problem of finding a universal compression scheme which
addresses both the local weak convergence and the sparse graphon frameworks in
Section~\ref{sec:statement-and-results} and state our main results on the
existence of such schemes. We explain the details of our compression scheme in
Section~\ref{sec:coding-scheme}. Afterwards, we analyze the performance of this
scheme under the local weak convergence and the sparse graphon frameworks in
Sections~\ref{sec:lwc-analysis} and \ref{sec:graphon-analysis} respectively. 

\editfinish

We close this section by introducing some notational conventions.
We write $:= $ and $=: $ for equality by definition.
$\reals$ and $\reals_+$ denote the set of real numbers and nonnegative real
numbers respectively. $\integers$ and $\nats$ denote the set of integers and
the set of positive integers respectively. 
We denote the set of integers $\{1, \dots, n\}$ by $[n]$. 
For $x \in \reals$, $x \geq 1$, we may write $[x]$ as a shorthand for $[\lfloor x \rfloor]$.
All the logarithms are
to the natural base, unless otherwise stated. We write $\{0,1\}^* -
\emptyset$ for the set of finite sequences of  zeros and ones, 
excluding the empty sequence. For
a sequence $x \in \{0,1\}^* - \emptyset$, we denote its length 
in bits
\black%
by
$\len(x)$. Moreover, we denote the length of $x$ in nats by $\nat(x) = \len(x)
\times \log 2$. $\mS^{p \times q}$ denotes the set of $p \times q$ matrices with
values in the set $\mS$.
 For two sequence $(a_n: n \geq
1)$ and $(b_n: n \geq 1)$ of nonnegative real numbers, we write $a_n = O(b_n)$
if there exists a constant $C > 0$ such that $a_n \leq C b_n$ for $n$ large
enough.
Moreover, we write $a_n = o(b_n)$ if $a_n / b_n \rightarrow 0$ as $n
\rightarrow \infty$. Also, we write $a_n = \omega(b_n)$ if $a_n / b_n
\rightarrow \infty$ as $n \rightarrow \infty$.
For a
probability distribution $P$ defined on a finite set, $H(P)$ denotes the Shannon entropy
of $P$. Similarly, for a random variable $X$ with finite support, $H(X)$ denotes
the Shannon entropy associated to $X$.  Moreover, for $\alpha \in [0,1]$, we define $H_b(\alpha) := - \alpha \log \alpha - (1-\alpha)
\log (1-\alpha)$ to be the Shannon entropy 
(to the natural base)
\black%
of a 
Bernoulli 
\black%
random
variable
with parameter $\alpha$.
\black%
We use the abbreviation ``a.s.'' for the phrase ``almost
surely''.

\section{Preliminaries}
\label{sec:prelim}

All graphs 
in this document
are assumed to be 
undirected and 
simple, 
the latter meaning that 
self loops and multiple edges are not
allowed. Hence we may drop the term 
``simple" when referring to graphs. We use the terms ``node'' and ``vertex''
exchangeably. We consider graphs which
may have either a finite or a countably infinite number of vertices. For a graph
$G$, let $V(G)$ denote the set of vertices in $G$. Two nodes $v$ and $w$ in a graph $G$ are said to be adjacent if
they are connected by an edge, and we show this by writing $v \sim_G w$.
We denote the degree of a vertex $v$ in a graph $G$ 
by $\deg_G(v)$.
$\mG_n$ denotes the set of simple graphs on the vertex set $[n]$. 
 A graph $G$ is called locally
finite if the degree of every vertex in the graph is finite.
Given a
graph $G \in \mG_n$, we denote its adjacency matrix by $A(G)$,
which
we recall is the
$n \times n$
matrix whose entry $(i,j)$ is one if nodes $i$ and $j$ are adjacent in $G$, and
zero otherwise. The density of a graph $G$, which is denoted by $\rho(G)$,  is defined to be the density of ones
in its adjacency matrix. More precisely,
\begin{equation}
  \label{eq:rho-G-def}
  \rho(G) := \frac{1}{n^2} \sum_{1 \leq i,j \leq n} (A(G))_{i,j} = \frac{2m}{n^2}.
\end{equation}
Here, $n$ and $m$ denote the number of vertices and edges in $G$ respectively.
For $p \geq 1$, the $L^p$ norm of an $n \times n$ matrix $A$ is defined as
\begin{equation}
  \label{eq:matrix-Lp-norm-def}
 \snorm{A}_p^p := \frac{1}{n^2} \sum_{1 \leq i , j \leq n} |A_{i,j}|^p. 
\end{equation}
Note the normalization. Thus 
$\rho(G) = \snorm{A}_1$.

A path between two vertices $v$ and $w$ in a graph $G$ is a sequence of nodes $v
= v_0, v_1, \dots, v_k = w$ where $v_i \sim_G v_{i+1}$ for $0 \leq i < k$. The
length of such a path is defined to be $k$.  The distance between two nodes $v$
and $w$ in a graph $G$ 
is defined to be the minimum length among the paths connecting them, and is
defined to be $\infty$ if no such path exists.

Two  graphs $G$ and $G'$ 
are said to be isomorphic, and we write $G \equiv G'$, if there is a bijection
$\phi: V(G) \rightarrow V(G')$ such that for all pair of vertices $v, w \in
V(G)$, we have $v \sim_G w$ iff $\phi(v) \sim_{G'} \phi(w)$.

To better understand this notion, let $\mathcal{S}_n$ denote the permutation
group on the set $[n]$. For a permutation $\pi \in \mathcal{S}_n$
and a graph $G$ on the vertex set $[n]$, let $\pi G$ be the
graph on the same vertex set after the permutation $\pi$ is applied on
the vertices. Namely, for each edge $(v,w)$ in $G$, we place an edge between the vertices $\pi(v)$ and $\pi(w)$ 
in $\pi G$.
Then each $\pi G$ is isomorphic to $G$ and every graph
that is isomorphic to $G$ is of the form $\pi G$ for some
$\pi \in \mathcal{S}_n$. 

Given a graph $G$, 
and a subset $S$ of its vertices, the subgraph induced by $S$ is 
the graph comprised of the vertices in $S$ 
and those edges in $G$
that have both their endpoints in $S$.
The {\em connected component} of a vertex $v \in V(G)$ is the subgraph of $G$ induced by 
the vertices that are at a finite distance from $v$. 
We write $G_v$ for the connected component of $v \in V(G)$. Note that $G_v$ is
a connected graph.

The focus on how a graph looks from the point of view of each of its
vertices is the key conceptual ingredient in the theory of local weak 
convergence.
For this, we introduce the notion of a
{\em rooted} 
graph and the notion of isomorphism of rooted  graphs. 
Roughly speaking, a rooted  graph should be thought of as
a graph as seen from a specific vertex in it and the
notion of two rooted 
graphs being isomorphic as capturing the idea that the respective 
graphs as seen from the respective distinguished vertices look the same. 
Notice that it is natural to restrict attention to the connected component
containing the root when making such a definition, because,  
roughly speaking, a vertex of the  graph should only be able to see the component to which it belongs.


For a precise definition, consider a graph $G$ and a distinguished vertex $o \in V(G)$.
The pair $(G,o)$ is called a rooted  graph. 
We call two rooted  graphs $(G, o)$ and $(G',o')$ isomorphic and write $(G, o) \equiv (G',o')$ if $G_o \equiv G'_{o'}$ through a bijection $\phi: V(G_o) \rightarrow V(G'_{o'})$ preserving the root, i.e.\ $\phi(o) = o'$.
This notion of isomorphism defines an equivalence relation 
on rooted graphs. Note that in order to determine
if two rooted graphs are isomorphic (as rooted  graphs)
it is only necessary to 
examine the connected component of the root in each of the graphs.
Let $[G, o]$ denote the equivalence class corresponding to $(G_o, o)$. 
In the sequel, we will only use this notion for locally finite graphs.

For a rooted graph $(G, o)$ and integer $h\geq 1$, let $(G, o)_h$ denote
the subgraph of  $G$ rooted at $o$ induced by  vertices with distance no more
than $h$ from $o$. 
Note that if $h=0$ then $(G, o)_h$ is 
the isolated root $o$.
Moreover, let $[G, o]_h$ be the  equivalence class corresponding to $(G, o)_h$, i.e.\ $[G, o]_h := [(G, o)_h]$.
Note that $[G, o]_h$ depends only on  $[G, o]$.

\subsection{The framework of Local Weak Convergence}
\label{sec:prelim-lwc}

In this section we review the framework of local weak convergence
of graphs, 
which is an instance of the 
so-called objective method.
See \cite{benjamini2011recurrence, aldous2004objective, aldous2007processes} for more details. 
This framework can also take into account \emph{marked} graphs, i.e.\ graphs
where each vertex carries a label
from a set called the vertex mark set and each edge carries a label from a set called the edge mark set.
However, for the purpose of
this work, we only focus on simple graphs without marks.

Let $\mG_*$ be the space of equivalence classes 
$[G, o]$
arising from locally finite rooted  graphs $(G,o)$.
We 
emphasize
again that in defining $[G, o]$ all that matters about
$(G,o)$ is the connected component of the root.
We define the metric $d_*$ on $\mG_*$ as follows: given $[G, o]$ and $[G',o']$, let $\hat{h}$ be the supremum over all integers $h\geq 0$ such that $(G, o)_h \equiv (G', o')_h$, where 
$(G, o)$ and $(G',o')$ are arbitrary members in equivalence classes $[G, o]$ and $[G', o']$ respectively\footnote{As all elements in an equivalence class are isomorphic, the definition is invariant under the choice of the representatives.
}.
With this, $d_*([G,o], [G',o'])$ is defined to be $1/(1+\hat{h})$. One can check that $d_*$ is a metric; in particular, it satisfies the triangle inequality.
Moreover, 
$\mG_*$ together with this metric is a Polish space, i.e. a 
complete separable metric space \cite{aldous2007processes}.
Let $\mT_*$ denote the subset of $\mG_*$ comprised of the 
equivalence classes $[G, o]$ arising from some $(G,o)$ where
the graph underlying $G$ is a tree. 
In the sequel we will 
think of $\mG_*$ as a Polish space with the metric $d_*$
defined above, rather than just a set. Note that $\mT_*$ is a closed subset of $\mG_*$.

For a Polish space 
$\Omega$, let $\mP(\Omega)$ denote the set of Borel probability measures on
$\Omega$. We say that a sequence of measures $\mu_n$ on $\Omega$ converges
weakly to $\mu \in \mP(\Omega)$, and write $\mu_n \Rightarrow \mu$, if for any
bounded continuous function on $\Omega$, we have $\int f d \mu_n \rightarrow
\int f d \mu$.
It can be shown that it suffices to verify this condition only for
uniformly continuous and bounded functions \cite{billingsley2013convergence}.
For a Borel set $B \subset \Omega$, the $\epsilon$--extension of
$B$, denoted by $B^\epsilon$, is defined as the union of the open balls with
radius $\epsilon$ centered around the points in $B$. For two probability
measures $\mu$ and $\nu$ in $\mP(\Omega)$, the \LP distance $\dlp(\mu, \nu)$ is
defined to be the infimum of all $\epsilon>0$ such that for all Borel sets $B
\subset \Omega$ we have $\mu(B) \leq \nu(B^\epsilon) + \epsilon$ and  $\nu(B) \leq \mu(B^\epsilon) + \epsilon$.
It is known that the \LP distance metrizes the topology of weak convergence 
on the space of probability distributions on a Polish space
(see, for instance, \cite{billingsley2013convergence}).
For $x \in \Omega$, let $\delta_x$ be the Dirac measure at $x$.

For a finite graph $G$, define $U(G) \in \mP(\mG_*)$ as 
\begin{equation}
  \label{eq:UG}
  U(G) := \frac{1}{|V(G)|} \sum_{o \in V(G)} \delta_{[G, o]}.
\end{equation}
Note that $U(G) \in \mP(\mG_*)$. In creating $U(G)$
from $G$, we have created a probability distribution on 
rooted  graphs from the given  graph $G$ 
by rooting the graph at a vertex chosen uniformly at random. 
Furthermore, for an integer $h \geq 1$, let 
\begin{equation}
\label{eq:UkG}
  U_h(G) := \frac{1}{|V(G)|} \sum_{o \in V(G)} \delta_{[G, o]_h}.
\end{equation}
We then have $U_h(G) \in \mP(\mG_*)$.
See Figure~\ref{fig:UG} for an example.

We say that a probability distribution $\mu$ on $\mG_*$ is the {\em local weak limit} of a sequence of finite  graphs $\{G_n\}_{n=1}^\infty$ when $U(G_n)$ converges weakly to $\mu$ 
(with respect to the topology 
on $\mP(\mG_*)$ induced by the metric $d_*$ on $\mG_*$).
This turns out to be equivalent to the condition that, for any finite depth $h \ge 0$, the structure of $G_n$ from the point of view of a root chosen uniformly at random and then looking around it only to depth $h$ converges in distribution to $\mu$ truncated up to depth $h$. This description of what is being captured by the definition 
justifies the term ``local'' in local weak convergence.

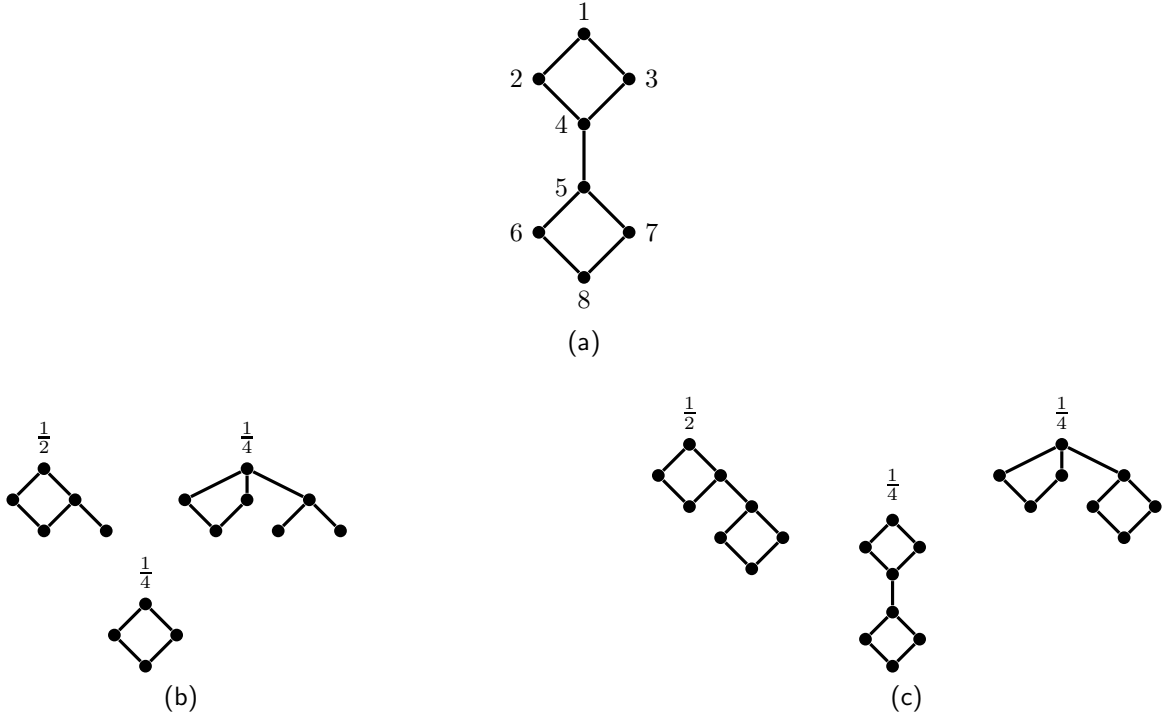
\begin{figure}
  \centering
  \subfloat[]{
      \begin{tikzpicture}[scale=0.6]
    \begin{scope}[yshift=0.7cm]
    \node[nodeBlack] (n1) at (0,2) {};
    \node[nodeBlack] (n2) at (-1,1) {};
    \node[nodeBlack] (n3) at (1,1) {};
    \node[nodeBlack] (n4) at (0,0) {};
  \end{scope}
  \begin{scope}[yshift=-0.7cm]
    \node[nodeBlack] (n5) at (0,0) {};
    \node[nodeBlack] (n6) at (-1,-1) {};
    \node[nodeBlack] (n7) at (1,-1) {};
    \node[nodeBlack] (n8) at (0,-2) {};
  \end{scope}
  \nodelabel{n1}{90}{1};
  \nodelabel{n2}{180}{2};
  \nodelabel{n3}{0}{3};
  \nodelabel{n4}{180}{4};
  \nodelabel{n5}{180}{5};
  \nodelabel{n6}{180}{6};
  \nodelabel{n7}{0}{7};
  \nodelabel{n8}{270}{8};

  \draw[very thick](n1)--(n2);
  \draw[very thick](n1)--(n3);
  \draw[very thick](n8)--(n6);
  \draw[very thick](n8)--(n7);

  \draw[very thick](n2)--(n4);
  \draw[very thick](n3)--(n4);
  \draw[very thick](n5)--(n6);
  \draw[very thick](n5)--(n7);
  \draw[very thick] (n4) -- (n5);
  \end{tikzpicture}
}

\subfloat[]{
\centering
\begin{tikzpicture}[scale=0.9]
  \begin{scope}[yshift=2cm,xshift=1.5cm, scale=0.23]
    \node[nodeBlack] (n1) at (0,0) {};
    \node[nodeBlack] (n2) at (-4,-2) {};
    \node[nodeBlack] (n3) at (0,-2) {};
    \node[nodeBlack] (n4) at (4,-2) {};
    \node[nodeBlack] (n5) at (-2,-4) {};
    \node[nodeBlack] (n6) at (2,-4) {};
    \node[nodeBlack] (n7) at (6,-4) {};

    \draw[very thick] (n1) -- (n2);
    \draw[very thick] (n1) -- (n3);
    \draw[very thick] (n1) -- (n4);
    \draw[very thick] (n2) -- (n5);
    \draw[very thick] (n3) -- (n5);
    \draw[very thick] (n4) -- (n6);
    \draw[very thick] (n4) -- (n7);

    \node at (0,2) {$\frac{1}{4}$};

  \end{scope}

  \begin{scope}[yshift=2cm,xshift=-1.5cm,scale=0.23]
    \node[nodeBlack] (n1) at (0,0) {};
    \node[nodeBlack] (n2) at (-2,-2) {};
    \node[nodeBlack] (n3) at (2,-2) {};
    \node[nodeBlack] (n4) at (0,-4) {};
    \node[nodeBlack] (n5) at (4,-4) {};
    \draw[very thick] (n1) -- (n2);
    \draw[very thick] (n1) -- (n3);
    \draw[very thick] (n2) -- (n4);
    \draw[very thick] (n3) -- (n4);
    \draw[very thick] (n3) -- (n5);
    
    \node at (0,2) {$\frac{1}{2}$};

  \end{scope}

  \begin{scope}[scale=0.23]
    \node[nodeBlack] (n1) at (0,0) {};
    \node[nodeBlack] (n2) at (-2,-2) {};
    \node[nodeBlack] (n3) at (2,-2) {};
    \node[nodeBlack] (n4) at (0,-4) {};
    
    \draw[very thick] (n1)--(n2);
    \draw[very thick] (n1)--(n3);
    \draw[very thick] (n2)--(n4);
    \draw[very thick] (n3)--(n4);
    
    \node at (0,2) {$\frac{1}{4}$};
  \end{scope}
\end{tikzpicture}
\label{fig:U-d2}
}%
\hfill\hfill\hfill
\subfloat[]{
\centering
\begin{tikzpicture}[scale=0.9]
  \begin{scope}[yshift=1cm,xshift=2.5cm,scale=0.23]
        \node[nodeBlack] (n1) at (0,0) {};
    \node[nodeBlack] (n2) at (-4,-2) {};
    \node[nodeBlack] (n3) at (0,-2) {};
    \node[nodeBlack] (n4) at (4,-2) {};
    \node[nodeBlack] (n5) at (-2,-4) {};
    \node[nodeBlack] (n6) at (2,-4) {};
    \node[nodeBlack] (n7) at (6,-4) {};
    \node[nodeBlack] (n8) at (4,-6) {};
    
    \draw[very thick] (n1) -- (n2);
    \draw[very thick] (n1)--(n3);
    \draw[very thick] (n1) -- (n4);
    \draw[very thick] (n2)--(n5);
    \draw[very thick] (n3)--(n5);
    \draw[very thick] (n4)--(n6);
    \draw[very thick] (n4)--(n7);
    \draw[very thick] (n6)--(n8);
    \draw[very thick] (n7)--(n8);
    
    \node at (0,2) {$\frac{1}{4}$};

  \end{scope}

  \begin{scope}[yshift=1cm,xshift=-3cm,scale=0.23]

    \node[nodeBlack] (n1) at (0,0) {};
    \node[nodeBlack] (n2) at (-2,-2) {};
    \node[nodeBlack] (n3) at (2,-2) {};
    \node[nodeBlack] (n4) at (0,-4) {};
    \node[nodeBlack] (n5) at (4,-4) {};
    \node[nodeBlack] (n6) at (2,-6) {};
    \node[nodeBlack] (n7) at (6,-6) {};
    \node[nodeBlack] (n8) at (4,-8) {};
    
    \draw[very thick](n1)--(n2);
    \draw[very thick](n1)--(n3);
    \draw[very thick](n2)--(n4);
    \draw[very thick](n3)--(n4);
    \draw[very thick](n5)--(n6);
    \draw[very thick](n5)--(n7);
    \draw[very thick](n6)--(n8);
    \draw[very thick](n7)--(n8);
    \draw[very thick] (n3) -- (n5);

    \node at (0,2) {$\frac{1}{2}$};

  \end{scope}

  \begin{scope}[scale=0.4,yshift=-3cm]
        \begin{scope}[yshift=0.7cm]
    \node[nodeBlack] (n1) at (0,2) {};
    \node[nodeBlack] (n2) at (-1,1) {};
    \node[nodeBlack] (n3) at (1,1) {};
    \node[nodeBlack] (n4) at (0,0) {};
  \end{scope}
  \begin{scope}[yshift=-0.7cm]
    \node[nodeBlack] (n5) at (0,0) {};
    \node[nodeBlack] (n6) at (-1,-1) {};
    \node[nodeBlack] (n7) at (1,-1) {};
    \node[nodeBlack] (n8) at (0,-2) {};
  \end{scope}

  \draw[very thick](n1)--(n2);
  \draw[very thick](n1)--(n3);
  \draw[very thick](n8)--(n6);
  \draw[very thick](n8)--(n7);

  \draw[very thick](n2)--(n4);
  \draw[very thick](n3)--(n4);
  \draw[very thick](n5)--(n6);
  \draw[very thick](n5)--(n7);
  \draw[very thick] (n4) -- (n5);

    \node at (0,4) {$\frac{1}{4}$};
  \end{scope}
\end{tikzpicture}
\label{fig:U-U}
}
\caption{\label{fig:UG} 
With $G$ being the graph in (a),
(b) illustrates $U_2(G)$, which is a probability distribution on rooted graphs of depth at most 2 and
(c) depicts $U(G)$, which is a probability distribution on $\mG_*$. 
In each of the figures in (b) and (c) the root is the vertex at the top. 
\black
}
\end{figure}

In fact, $U_h(G)$ could be thought of as the ``depth $h$ empirical
distribution'' of the  graph $G$. On the other hand, a probability distribution
$\mu \in \mP(\mG_*)$ that arises as a local weak limit plays the role of a
stochastic process on graphical data, and a sequence of graphs
$\{G_n\}_{n=1}^\infty$ could be  thought of as being 
asymptotically distributed like
this process when $\mu$ is the local weak limit of the sequence.

The degree of a probability measure $\mu \in \mP(\mG_*)$,
denoted by $\deg(\mu)$, is defined as
\begin{equation*}
  \deg(\mu) := \evwrt{\mu}{ \deg_{G}(o)} = \int \deg_{G}(o) d \mu([G, o]),
\end{equation*}
which is the expected degree of the root.



We next present some examples to illustrate the concepts defined so far.
\begin{enumerate}
\item Let $G_n$ be the finite lattice $\{-n, \dots n \} \times \{-n, \dots, n
  \}$ in $\mathbb{Z}^2$. As $n$ goes to infinity,  the local weak limit of this sequence is the distribution that gives probability one to the lattice $\mathbb{Z}^2$ rooted at the origin. The reason is that if we fix a depth $h \ge 0$ then for $n$ large almost all of the vertices in $G_n$ cannot see the borders of the lattice when they look at the graph around them up to depth $h$, so
these vertices cannot locally distinguish the graph on which they live from the infinite lattice $\mathbb{Z}^2$.
\item Suppose $G_n$ is a cycle of length $n$. The local weak limit 
of this sequence of graphs gives probability one to an infinite $2$--regular tree
rooted at one of its vertices. The intuitive explanation for this is 
essentially identical to that for the preceding example.

\item Let $G_n$ be a realization of the sparse \ER graph $\mG(n, \alpha / n)$
where $\alpha > 0$, i.e.\  $G_n$ has $n$ vertices and each edge is independently present with
  probability $\alpha / n$
 (here $n$ is assumed to be sufficiently large).
 One can show that if all the $G_n$ are defined on
  a common probability space then, almost surely, the local weak limit of the
  sequence is the Poisson Galton--Watson tree with mean  $\alpha$, 
rooted at the initial vertex. To justify why this should be true without going
through the details, note that the degree of a vertex in $G_n$ is the sum of
$n-1$ independent Bernoulli random variables, each with parameter $\alpha /n$.
For $n$ large, this  approximately has a Poisson distribution with mean
$\alpha$. This argument could be repeated for any of the vertices to which the
chosen vertex is connected, which play the role of the offspring of the initial
vertex in the limit. The essential point is that the probability of having loops
in the neighborhood of a typical vertex up to a depth $h$ is  negligible
whenever $h$ is fixed and $n$ goes to infinity. 

\end{enumerate}

\subsection{Unimodularity}
\label{sec:unimodularity}

In order to get a better understanding of the nature of the 
results proved in this paper, it is 
helpful
to understand what is
meant by a
{\em unimodular} probability distribution $\mu \in \mP(\mG_*)$.
We give the relevant definitions and context in this section.

Since each vertex in $G_n$ has the same chance of being chosen as 
the root in the definition of $U(G_n)$, this should manifest itself as some
kind of stationarity property of the limit $\mu$ with respect
to changes of the root.
A probability distribution $\mu \in \mP(\mG_*)$ is called {\em sofic} 
if there exists a sequence of finite graphs $G_n$ with local weak limit $\mu$. 
The definition of unimodularity is made in an attempt to understand
what it means for a Borel probability distribution on $\mG_*$
to be sofic.

To define unimodularity, let $\mG_{**}$ be the set of isomorphism classes
$[G,o,v]$ where $G$ is a  connected graph with two distinguished vertices
$o$ and $v$ in $V(G)$ (ordered, but not necessarily distinct). Here, isomorphism
is defined by an adjacency-preserving vertex bijection which also maps the two distinguished vertices of one
object to the respective ones of the other. 
$\mG_{**}$ can be metrized as a Polish space in a manner similar to that used to metrize $\mG_{*}$.
A measure $\mu \in \mP(\mG_*)$ is
said to be unimodular if, for all measurable functions $f:
\mG_{**} \rightarrow \reals_+$, we have
\begin{equation}
  \label{eq:unim-integral}
  \int \sum_{v \in V(G)} f([G,o,v]) d\mu([G,o]) = \int \sum_{v \in V(G)} f([G,v,o]) d\mu([G,o]).
\end{equation}
Here, the summation is taken over all vertices $v$ which are in the same
connected component of $G$ as $o$. 
(Note that the integrand on the left hand side is $f([G,o,v])$,
  while the integrand on the right hand side is $f([G,v,o])$.) Roughly speaking, this condition ensures that the distribution
  $\mu$ is invariant under switching the root, and it can be considered as a
   stationarity condition.
It can be seen that it suffices to check the
above condition for a function $f$ such that $f([G,o,v]) = 0$ unless $v \sim_G
o$. This is called \emph{involution invariance} \cite{aldous2007processes}.
Let $\mP_u(\mG_*)$ denote the set of unimodular probability measures on
$\mG_*$. Also, since $\mT_* \subset \mG_*$, we can define the set of
unimodular probability measures on $\mT_*$ and denote it by $\mP_u(\mT_*)$. 
A sofic probability measure is unimodular. Whether the other direction also
holds is unknown.

\subsection{The BC Entropy}
\label{sec:bc-ent}

In this section we review the notion of entropy introduced by Bordenave and
Caputo  for probability distributions on the space
$\mG_*$ \cite{bordenave2015large}. We call this notion  the \emph{BC
  entropy}. The authors of this paper have generalized this entropy to
the regime where the vertices and edges in the graph 
also carry marks,
but we omit that discussion here since
we focus on unmarked graphs throughout this work \cite{delgosha2019notion}.

For integers $n, m \in \nats$, let $\mG_{n, m}$ denote the set of graphs on the
vertex set $[n]$ with precisely $m$ edges. 
An application of Stirling's formula implies that if $d > 0$ and the
sequence $m_n$ is such that $m_n / n \rightarrow d/2$, then we have
\begin{equation*}
  \log |\mG_{n, m_n}| = m_n \log n + s(d) n + o(n),
\end{equation*}
where $s(d) := \frac{d}{2} - \frac{d}{2} \log d$.

The key idea to define the  BC entropy is to count the number of ``typical''
graphs.
More precisely, given $\mu \in
\mP(\mG_*)$ and $\epsilon > 0$, let $\mG_{n, m}(\mu, \epsilon)$ denote the set
of graphs $G \in \mG_{n, m}$ such that $\dlp(U(G), \mu) < \epsilon$, where
$\dlp$ refers to the \LP metric on $\mP(\mG_*)$
\cite{billingsley2013convergence}. In fact, one can interpret $\mG_{n, m}(\mu,
\epsilon)$ as the set of $\epsilon$--typical graphs with respect to $\mu$.
It turns out that, roughly speaking,  the number of
$\epsilon$--typical graphs scales as follows:
\begin{equation*}
  |\mG_{n,m}(\mu, \epsilon)| = \exp(m \log n + n \bch(\mu) + o(n)), 
\end{equation*}
where $\bch(\mu)$ is the  BC entropy of $\mu$ which will be defined below.
In order to make this precise, we make the following definition.

\begin{definition}
  \label{def:unmarked-bc-upper-lower}
  Assume $\mu \in \mP(\mG_*)$ is given, with $0 < \deg(\mu) < \infty$. Assume
  that $d > 0$ is fixed and a sequence $m_n$ of integers is given such that $m_n
  / n \rightarrow d / 2$ as $n \rightarrow \infty$. With these, for $\epsilon >
  0$, we define
  \begin{equation*}
    \bchover_d(\mu, \epsilon)|_{(m_n)} := \limsup_{n \rightarrow \infty} \frac{\log |\mG_{n, m_n}(\mu, \epsilon)| - m_n \log n}{n},
  \end{equation*}
  which we call the $\epsilon$--upper BC entropy. Since this is increasing in $\epsilon$, we can define the upper BC
  entropy as
  \begin{equation*}
    \bchover_d(\mu)|_{(m_n)} := \lim_{\epsilon \downarrow 0} \bchover_d(\mu, \epsilon)|_{(m_n)}.
  \end{equation*}
  We may similarly define the $\epsilon$--lower BC entropy
  $\bchunder_d(\mu,\epsilon)|_{(m_n)}$ as
  \begin{equation*}
    \bchunder_d(\mu, \epsilon)|_{(m_n)} := \liminf_{n \rightarrow \infty} \frac{\log |\mG_{n, m_n}(\mu, \epsilon)| - m_n \log n}{n}.
  \end{equation*}
  Since this is increasing in $\epsilon$, we can define the lower BC
  entropy as
  \begin{equation*}
    \bchunder_d(\mu)|_{(m_n)} := \lim_{\epsilon \downarrow 0} \bchunder_d(\mu, \epsilon)|_{(m_n)}.
  \end{equation*}
\end{definition}

Theorem 1.2 in \cite{bordenave2015large} summarizes some of the main properties
of the BC entropy. For better readability, we split that theorem as Theorems~\ref{thm:unmarked-bc-entr-infinity}
and \ref{thm:unmarked-bc-ent-well-defined} below. The following Theorem~\ref{thm:unmarked-bc-entr-infinity} shows that certain conditions must be met
for the  BC entropy to be of interest.

\begin{thm}
  \label{thm:unmarked-bc-entr-infinity}
  Fix $d>0$ and assume that 
  $\mu \in \mP(\mG_*)$ with
  $0< \deg(\mu) < \infty$ satisfies any of the following conditions:
  \begin{enumerate}
  \item $\mu$ is not unimodular;
  \item $\mu$ is not supported on $\mT_*$;
  \item $d \neq \deg(\mu)$.
  \end{enumerate}
  Then, for any choice of the sequence $m_n$ such that $m_n / n \rightarrow d/2$
  as $n \rightarrow \infty$, we have $\bchover_{d}(\mu)|_{(m_n)} = -\infty$.
\end{thm}

  A consequence of Theorem~\ref{thm:unmarked-bc-entr-infinity}
is that the only case of interest in the discussion of 
BC entropy is when $\mu \in \mP_u(\mT_*)$,
$d = \deg(\mu)$, 
and the
sequence $m_n$  is such that $m_n / n \rightarrow \deg(\mu) / 2$.
Namely, the only upper and lower  BC entropies of interest are 
$\bchover_{\deg(\mu)}(\mu)|_{(m_n)}$ and $\bchunder_{\deg(\mu)}(\mu)|_{(m_n)}$ respectively.

The following Theorem~\ref{thm:unmarked-bc-ent-well-defined} establishes that the upper and lower
 BC entropies do not depend on the 
choice of the defining  sequence
$m_n$. Further, 
this theorem establishes that
the upper  BC entropy 
is always equal to the lower  BC entropy.

\begin{thm}
  \label{thm:unmarked-bc-ent-well-defined}
  Assume that $d>0$ is given. For any $\mu \in \mP(\mG_*)$ such that $0 <
  \deg(\mu) < \infty$, we have
  \begin{enumerate}
  \item The values of $\bchover_{d}(\mu)|_{(m_n)}$ and
    $\bchunder_d(\mu)|_{(m_n)}$ are invariant under the specific choice of the
    sequence $m_n$ such that $m_n / n \rightarrow d/2$. With this, we may
    simplify the notation and unambiguously write $\bchover_d(\mu)$ and
    $\bchunder_d(\mu)$.
  \item $\bchover_d(\mu) = \bchunder_d(\mu)$. We may therefore unambiguously
    write $\bch_d(\mu)$ for this common value and call it the BC entropy of $\mu
    \in \mP(\mG_*)$ with respect to  $d$. Moreover, $\bch_d(\mu) \in [-\infty, s(d)]$. Here,
    $s(d) := \frac{d}{2} - \frac{d}{2} \log d$.
  \end{enumerate}
\end{thm}

From Theorem~\ref{thm:unmarked-bc-entr-infinity} we conclude that unless 
$d = \deg(\mu)$, and  $\mu$
  is a unimodular measure on $\mT_*$, we have 
  $\bch_{d}(\mu) = -\infty$. 
  In view of this, for $\mu \in \mP(\mG_*)$
  with $0<\deg(\mu) < \infty$, we write $\bch(\mu)$
  for  $\bch_{\deg(\mu)}(\mu)$. Likewise, we may write
  $\bchunder(\mu)$ and $\bchover(\mu)$ for $\bchunder_{\deg(\mu)}(\mu)$ and $\bchover_{\deg(\mu)}(\mu)$, respectively. 
    Note that, unless $\mu \in \mP_u(\mT_*)$, 
    we have $\bchover(\mu) = \bchunder(\mu) = \bch(\mu) = -\infty$.
    
We are now in a position to define the  BC entropy.

\begin{definition}
  \label{def:BC-entropy-new}
  For $\mu \in \mP(\mG_*)$
  with $0 < \deg(\mu) < \infty$, the  BC entropy of $\mu$ is defined to be $\bch(\mu)$.
\end{definition}

The reader is referred to \cite{bordenave2015large} for a detailed discussion of
the BC entropy and 
some of 
its additional properties. For instance, it can be shown that
the BC entropy of a probability distribution $\mu \in \mP_u(\mT_*)$ can be
approximated in terms of the finite depth truncation of $\mu$ \cite[Theorem 1.3]{bordenave2015large}.
The reader is also referred  to
\cite{delgosha2019notion} for the generalization of this notion to the marked
regime. 

An important property of the BC entropy which we will need in our analysis is
\emph{upper semi--continuity}.


orange
\begin{lem}[Lemma 5.3 in \cite{bordenave2015large}]
\label{lem:bc-ent-simple-upp-sem-cont}
  Assume that a sequence $\mu_k \in \mP(\mG_*)$ together with $\mu \in
  \mP(\mG_*)$  are given such
that $\deg(\mu) \in (0,\infty)$, $\deg(\mu_k) \in (0,\infty)$ for $k$
sufficiently large, and $\mu_k \Rightarrow \mu$.
Then
\begin{equation*}
  \bch(\mu) \geq \limsup_{k \rightarrow \infty} \bch(\mu_k).
\end{equation*}
\end{lem}

Let $\mu \in \mP_u(\mT_*)$ be given such that $\deg(\mu) \in (0,\infty)$.
Given $[T,o] \in \mT_*$ and $k > 0$, we define $[T^k,o] \in \mT_*$ to be obtained from $[T,o]$
by removing all the edges in $T$ where the degree of at least one of their
endpoints is strictly bigger than $k$, followed by taking the connected component of
the root. Now, let $\muk \in \mP(\mT_*)$ be the law of $[T^k, o]$ when $[T,o]$ has law $\mu$.
It is easy to see that $\muk$ is unimodular. 

The following proposition is then an immediate consequence of 
Lemma~\ref{lem:bc-ent-simple-upp-sem-cont}.

\begin{prop}
  \label{prop:bc-ent-trunc-upper-sem-cont}
  Assume that $\mu \in \mP_u(\mT_*)$ is given such that $\deg(\mu) \in
  (0,\infty)$. Then, we have
  \begin{equation*}
    \limsup_{k \rightarrow \infty} \bch(\muk) \leq \bch(\mu). 
  \end{equation*}
\end{prop}

\black

\subsection{Graphons}
\label{sec:prelim-graphon}

\editstart

The theory of graphons provides a comprehensive framework to study the
asymptotics of dense graphs by introducing a limit theory for such graphs (see,
for instance, \cite{lovasz2006limits}, \cite{lovasz2007szemeredi},
\cite{borgs2008convergent}, \cite{borgs2012convergent}, \cite{lovasz2012large}). There has been some effort in adapting this theory for sparse
graphs (see, for instance, \cite{bollobas2007metrics}, \cite{borgs2019L}, \cite{borgs2018L}). Also, the problem of graphon estimation given
random graph samples has been extensively studied both in the dense regime and
in the sparse regime (see, for instance, \cite{bickel2009nonparametric},
\cite{wolfe2013nonparametric}, \cite{borgs2015private},
\cite{chatterjee2015matrix}, \cite{gao2015rate}). In this section, we review the notion of graphons in the
sparse regime. Furthermore, we review the result from \cite{borgs2015consistent} on graphon
estimation in this regime.  Here, we mainly stick to the setup and notation introduced in
\cite{borgs2015consistent}.

Assume that a probability space $(\Omega, \mF, \pi)$ is given. A graphon on this
probability space is
defined to be a
measurable function $W: \Omega \times \Omega \rightarrow \reals_+$ which is
symmetric, i.e.\ $W(x,y) = W(y,x)$ for all $x, y \in \Omega$, and is $L^1$,
i.e.\ $\snorm{W}_1 < \infty$. Here, the $L^p$ norm of a function $f : \Omega
\times \Omega \rightarrow \reals$ for $p \geq 1$ is defined as $\snorm{f}_p^p = \int_{\Omega
  \times \Omega} |f(x,y)|^p d\pi(x) d\pi(y)$. A graphon $W$ is said to be $L^p$
if $\snorm{W}_p < \infty$.
Moreover, $\snorm{W}_\infty$ is defined to be the essential supremum of $W$
with respect to the product measure $\pi \times \pi$.
We may simply say that $W$ is a
graphon on $\Omega$ when the $\sigma$-algebra $\mF$ and the probability measure
$\pi$ are clear from the context. In particular, when we refer to  a graphon $W$ as
being defined over 
$[0,1]$, it refers to a 
graphon over the probability space $[0,1]$ equipped with the standard Borel
$\sigma$-algebra and the uniform distribution, unless otherwise stated.
A simple graph $G$ on a finite vertex set $V$ naturally defines a
graphon $W$ over the probability space $V$ equipped with the uniform distribution,
defined as $W(v,w) = (A(G))_{v,w}$ for $v,w \in V$. 
Note that for each $p \ge 1$ the 
$L^p$ norm of this graphon is the same
as that of the adjacency matrix of the underlying graph $G$, as defined 
in~\eqref{eq:matrix-Lp-norm-def}.

Assume that a symmetric $n \times n$ matrix $B$ with nonnegative entries is given together with a probability  vector $\vec{p} = (p_1, \dots, p_n)$  such that $p_i \geq 0$,
$1 \leq i \leq n$, and $\sum_{i=1}^n p_i = 1$. We define the \emph{block
  graphon} $(\vec{p}, B)$ to be a graphon $W$ over the finite probability space
$[n]$ equipped with the probability distribution $\vec{p}$ such that for
$1 \leq i , j\leq n$, we have $W(i,j) = B_{i,j}$.
This generalizes the notion in the preceding paragraph of a graphon associated to a simple graph.

Now, we state the notion of equivalence for graphons (see Definition~2.5 in  \cite{borgs2015consistent}).
Given two graphons $W$ and $W'$ on probability spaces $(\Omega, \mF, \pi)$ and
$(\Omega', \mF', \pi')$, respectively, we say that $W$ and $W'$ are equivalent if there
exists a third probability space $(\Omega'', \mF'', \pi'')$ and two measure
preserving maps $\phi: \Omega \rightarrow \Omega''$ and $\phi': \Omega'
\rightarrow \Omega''$ together with a graphon $U$ on $(\Omega'', \mF'', \pi'')$,
such that for almost all $(x, y) \in \Omega \times \Omega$, with respect to the product
measure $\pi \times \pi$, we have $W(x,y) = U(\phi(x), \phi(y))$,
and similarly for almost all $(x', y') \in \Omega'$, with respect to the product
measure $\pi' \times \pi'$, we have $W'(x',y') = U(\phi'(x'), \phi'(y'))$.

For two $L^2$ graphons $W$ and $W'$, defined on  probability spaces $(\Omega,
  \mF, \pi)$ and $(\Omega', \mF', \pi')$ respectively, we define 
  \begin{equation}
    \label{eq:delta-2-def}
    \delta_2(W, W') := \inf_\nu \left( \int |W(x,y) - W'(x',y')|^2 d \nu(x,x') d\nu(y,y') \right)^{1/2},
  \end{equation}
  where the infimum is taken over all couplings $\nu$ of $\pi$ and $\pi'$, i.e.\
  $\nu$ is a probability measure over $\Omega \times \Omega'$ with marginals
  $\pi$ and $\pi'$, respectively.  In fact, $\delta_2$
    yields  a metric on the space of equivalence classes of $L^2$ graphons, with
    reference to the notion of equivalence described above (see \cite[Theorem 2.11
    and Appendix A]{borgs2015consistent} and \cite{janson2010graphons}).
Moreover, for two graphons $W$ and $W'$ on two probability spaces $(\Omega, \mF,
\pi)$ and $(\Omega', \mF', \pi')$ respectively, we define the cut norm as
\begin{equation}
  \label{eq:graphon-cut-norm-def}
  \delta_\square(W, W') := \inf_{\nu} \sup_{S, T \subseteq \Omega \times \pp{\Omega}} \left| \int_{(x,\pp{x}) \in S, (y,\pp{y})\in T} \left( W(x,y) - W'(x',y') \right) d\nu(x,x') d \nu(y,y')\right|,
\end{equation}
where the infimum is taken over all couplings $\nu$  of $\pi$ and $\pi'$, similar
to the above, and the supremum is over measurable subsets $S$ and $T$ of
$\Omega \times \pp{\Omega}$.
Note that every graphon is by definition an $L^1$ function
with respect to its underlying product measure,
hence the
cut norm is well defined. In fact,  $\delta_\square$
    yields a metric on the space of equivalence classes of graphons with
    reference to the notion of equivalence described above (see \cite[Theorem 2.11
    and Appendix A]{borgs2015consistent} and \cite{janson2010graphons}).

A graphon $W$ is said to be normalized if $\snorm{W}_1 = 1$.
Given a normalized graphon $W$ on a probability space $(\Omega, \mF, \pi)$ and
a  sequence of \emph{target densities} $\rho_n$,
i.e. strictly positive real numbers,
we define the sequence of $W$--random graphs with target density $\rho_n$ as a
sequence of random graphs $\Gn$, where $\Gn$ is defined
on the vertex set $[n]$, 
as follows.
We first generate random variables
$(X_i: i \geq 1)$ i.i.d.\ from $(\Omega, \mF, \pi)$. Then, for each $n$ and
each  pair of vertices
$1 \leq v, w \leq n$, we independently place an edge between $v$ and $w$ in $\Gn$ with
probability $\min\{1, \rho_n W(X_v, X_w)\}$. Note that the distribution of $\Gn$
is dependent on the random variables $X_1, \dots, X_n$, and the sequence $(X_i:
i \geq 1)$ is generated prior to generating $\Gn$. Consequently, the random graphs $\Gn$
defined in this procedure are 
dependent.
We denote the law of $\Gn$ in this
procedure by $\mG(n; \rho_n W)$. 
The following theorem from 
\cite{borgs2015consistent} shows that
equivalent graphons
generate identical distributions, given some conditions on the sequence
$\rho_n$.

\begin{thm}[Theorem~2.6 in \cite{borgs2015consistent}]
\label{thm:borgs-equivalent-graphons-same-graph-dist}
Let $W$ and $W'$ be graphons over $(\Omega, \mF, \pi)$ and $(\Omega', \mF',
\pi')$, respectively. Assume that $n \rho_n \rightarrow \infty$ and either
$\rho_n \max\{\snorm{W}_\infty, \snorm{W'}_\infty\} \leq 1$ or $\rho_n
\rightarrow 0$. Then
the sequences $(\mG(n; \rho_n W))_{n \geq 0}$ and $(\mG(n; \rho_n
  W))_{n \geq 0}$ are identically distributed
if and only if $W$ and $W'$ are equivalent. 
\end{thm}

In the dense regime, graphons are defined to take values bounded by 1
(see for instance Section 7.1.\ in \cite{lovasz2012large}).
However,
in the sparse regime discussed above, this condition is relaxed, and the
graphons are allowed to be unbounded.
Instead, the sequence of target densities $\rho_n$ is introduced which
scales the graphon in order to get the desired probability.
In fact, under some conditions on the sequence $\rho_n$, if $W$ is a normalized graphon, then $W$-random graphs
have a density close to $\rho_n$, justifying the term \emph{target density}.  Moreover, under some conditions on the sequence $\rho_n$, the
sequence of $W$--random graphs converges to  $W$ with respect
to the cut
metric defined in~\eqref{eq:graphon-cut-norm-def}. These statements are made precise in the following theorem.

\begin{thm}[Theorem~2.14 in~\cite{borgs2015consistent}]
  \label{thm:ganguli-W-random-graphon-convergence}
  Let $\Gn \sim \mG(n; \rho_n W)$ be a sequence of $W$-random graphons with target density
  $\rho_n$, where $W$ is a normalized graphon over an arbitrary probability
  space, and $\rho_n$ is such that $n \rho_n \rightarrow \infty$ and either
  $\limsup_{n \rightarrow \infty} \rho_n \snorm{W}_\infty \leq 1$  or $\rho_n
  \rightarrow 0$. Then, almost surely, $\rho(\Gn) / \rho_n \rightarrow 1$ and
  \begin{equation*}
    \delta_\square \left( \frac{1}{\rho(\Gn)} \Gn, W \right) \rightarrow 0 \qquad \text{a.s.}.
  \end{equation*}
\end{thm}

Note that, as we discussed above, $\Gn$ naturally defines a graphon, and $\Gn /
\rho(\Gn)$ refers to the scaled graphon corresponding to $\Gn$. Recall
from~\eqref{eq:rho-G-def} that
$\rho(\Gn)$ is the density of the graph $\Gn$, and is defined to be $2\mn /
n^2$, where $\mn$ denotes the number of edges in $\Gn$.
Theorem~\ref{thm:ganguli-W-random-graphon-convergence} above
  implies that for $\Gn \sim \mG(n; \rho_n W)$, with probability one  $\rho(\Gn) / \rho_n \rightarrow 1$ or equivalently
  $2\mn / (\rho_n n^2) \rightarrow 1$. Recall that since we want to study sparse graphs, we
  want $\mn$ to scale much slower than $n^2$. For this to happen, since $2\mn /
  (\rho_n n^2) \rightarrow 1$ almost surely, from this
  point forward, we assume that $\rho_n \rightarrow 0$. Moreover, motivated by
  Theorem~\ref{thm:ganguli-W-random-graphon-convergence} above, from this
  point forward we assume that $n \rho_n \rightarrow \infty$. 
Roughly speaking, the condition $n \rho_n \rightarrow \infty$ ensures
that the 
graphs in the 
\black
sequence of $W$--random graphs  are not \emph{too sparse}. 
More precisely, 
\black
since $n \rho_n \approx 2\mn / n$, 
the condition $n \rho_n \rightarrow \infty$ roughly means that the average
degree in $\Gn$ goes to infinity.
Therefore, this sparse graphon framework allows us to study 
\emph{heavy-tailed}
sparse graphs, 
as opposed to 
the local weak convergence theory, 
which requires the existence of a well-defined 
\black
limit degree
distribution at the root.


Borgs et al.\ have addressed the problem of estimating the graphon $W$ upon
observing 
a sample, for each vertex size $n$, of 
\black
a sequence of $W$-random graphs \cite{borgs2015consistent}. 
They study three methods for doing
so, namely least squares estimation, cut norm estimation, and degree sorting.
Here, we only review the least square estimation
method, 
\black
and refer the reader to
\cite{borgs2015consistent} for further reading. We will later employ this
estimation method in our universal compression scheme.

\subsubsection{Least Squares Algorithm}
\label{sec:prelim-least-square-algo}

In this section, we explain the least squares algorithm for graphon estimation from
\cite{borgs2015consistent} and state its properties. First, we need to introduce
some notation. Given integers $n$ and $k$, a function $\pi: [n] \rightarrow
[k]$, and a $k \times k$ matrix $B$, we define $B^\pi$ as an $n \times n$ matrix
such that $(B^\pi)_{i,j} = B_{\pi(i), \pi(j)}$ for $1 \leq i, j \leq n$.

\emph{Least Squares Algorithm:} Given a graph $G$ on $n$ vertices, and a parameter $\beta$ such that $1 \leq
\beta \leq n$, let
\begin{equation}
  \label{eq:least-sq-alg-min-1}
  (\hat{\pi}, \hat{B}) \in \argmin_{k, \pi, B \in \reals_+^{k \times k}} \snorm{A(G) - B^\pi}_2,
\end{equation}
where the minimization is over 
natural numbers
$k$, $k \times k$ matrices $B$ with
nonnegative entries, and functions $\pi : [n] \rightarrow
[k]$ such that for all $1 \leq i \leq k$, either $\pi^{-1}(\{i\}) =
\emptyset$ or $|\pi^{-1}(\{i\})| \geq \lceil  n / \beta \rceil$. Therefore, 
it suffices to restrict $k$ to be at most
$\lfloor  n / \lceil n / \beta \rceil \rfloor \leq
\lfloor \beta \rfloor$. In
other words, we may rewrite~\eqref{eq:least-sq-alg-min-1} equivalently as
follows
\begin{equation}
  \label{eq:least-sq-alg-better}
  (\hat{\pi}, \hat{B}) \in \argmin_{\pi: [n] \rightarrow [ \beta ], B \in \reals_+^{[ \beta ] \times [\beta ]}} \snorm{A(G) - B^\pi}_2,
\end{equation}
where the minimization is taken over $[\beta] \times [\beta]$ matrices $B$ and
$\pi: [n] \rightarrow [\beta]$ such that for $1 \leq i \leq \lfloor \beta
\rfloor$, either $\pi^{-1}(\{i\}) = \emptyset$ or $|\pi^{-1}(\{i\})| \geq
\lceil  n / \beta \rceil$. (Recall that $[\beta]$ is a shorthand for $[\lfloor\beta \rfloor]$.) Assume that we have solved the optimization
in~\eqref{eq:least-sq-alg-better}, and $\hat{\pi}$ and $\hat{B}$ 
are 
arbitrary
optimizers.
Then, we define the output of the least squares estimation algorithm
to be the block graphon $(\vec{\hat{p}}, \hat{B})$ where the probability vector $\vec{\hat{p}} = (\hat{p}_1, \dots,
\hat{p}_{\lfloor \beta \rfloor})$ is such that $\hat{p}_i = |{\hat{\pi}}^{-1}(\{i\})| / n$
for $1 \leq i \leq \lfloor \beta \rfloor$.

Note that the discrete optimization problem in~\eqref{eq:least-sq-alg-better}
requires searching over all mappings $\pi : [n] \rightarrow [\beta]$. However,
since the objective is an $L^2$ norm, by fixing $\pi$ the objective is
minimized by choosing $B$ such that for $1 \leq i,j \leq \lfloor \beta \rfloor$
such that $\pi^{-1}(\{i\})$ and $\pi^{-1}(\{j\})$ are not empty we have 
  \begin{equation}
    \label{eq:Bij-optimal-block-average}
    B_{i,j} = \frac{1}{|\pi^{-1}(i)| |\pi^{-1}(j)|} \sum_{u \in \pi^{-1}(i), v \in \pi^{-1}(j)} (A(G))_{u,v}.
  \end{equation}
Note that the choice of $B_{i,j}$ for $i$ and $j$ such that either
$\pi^{-1}(\{i\}) = \emptyset$ or $\pi^{-1}(\{j\}) = \emptyset$ has no effect on
the objective. In other words, 
for fixed $\pi$
the optimizer  $B$ must take the
average of the adjacency matrix over the blocks defined by $\pi$.

It can be shown that with an appropriate choice of the parameter $\beta$, the
above algorithm yields a consistent graphon estimation scheme,  in the following
sense.

\begin{thm}[Theorem~3.1 in \cite{borgs2015consistent}]
  \label{thm:ganguli-consistent-result}
  Let $W$ be an $L^2$ graphon, normalized so that $\norm{W}_1 = 1$, and let
  $\Gn$ be a sequence of $W$-random graphs with target densities $(\rho_n: n
  \geq 1)$.
  Furthermore,  let $\hW
  = (\vec{\hat{p}}, \hat{B})$ be the output of the above least squares
  algorithm for $\Gn$ with parameter $\beta_n$.
  Then, if $\rho_n$ and $\beta_n$ are such that as $n \rightarrow \infty$, we
  have $\rho_n \rightarrow 0$, $n \rho_n \rightarrow \infty$, $\beta_n
  \rightarrow \infty$, and $\beta_n^2 \log \beta_n = o(n \rho_n)$, then, with
  probability one, we have
  \begin{equation*}
    \lim_{n \rightarrow \infty} \delta_2\left( \frac{1}{\rho_n} \hW, W \right) = 0.
  \end{equation*}
\end{thm}


\begin{rem}
  In order to simplify the expressions, in the above we have stated a reparametrization of the
  least squares algorithm presented in~\cite{borgs2015consistent}. 
  We can show that Theorem~\ref{thm:ganguli-consistent-result} is a
  consequence of Theorem~3.1 in \cite{borgs2015consistent}.
  \black
  In~\cite{borgs2015consistent}, the least squares algorithm is explained based
  on the optimization problem~\eqref{eq:least-sq-alg-min-1}, with the only
  difference  that the constraint $|\pi^{-1}(\{i\})| \geq \lceil n / \beta
  \rceil$ for nonempty classes 
  is replaced by the constraint $|\pi^{-1}(\{i\})|
  \geq \lfloor n \kappa \rfloor$,
  \black
  where $\kappa$
  is a parameter satisfying $\kappa \in [1/n, 1]$. 
  Given $\beta$ and
  the optimization problem in
  \black
  the above form~\eqref{eq:least-sq-alg-min-1}, 
  one can choose $\kappa$ to be $\lceil n / \beta \rceil / n$ 
  to obtain an optimization problem in the form  presented in
  \cite{borgs2015consistent}. Also, given $\kappa$ and the optimization
  in the form presented in \cite{borgs2015consistent}, one can choose $\beta = n
  / \lfloor  n \kappa \rfloor$ to obtain an optimization of the
  form~\eqref{eq:least-sq-alg-min-1}. Furthermore, in the setup of
  Theorem~\ref{thm:ganguli-consistent-result} above, given the sequence
  $\beta_n$ 
  such that $\beta_n \rightarrow \infty$ 
  \black
  and $\beta_n^2 \log \beta_n = o(n
  \rho_n)$, 
  if we choose
  \black
  $\kappa_n = \lceil n / \beta_n \rceil / n$, 
  since $1 / \beta_n \le \kappa_n \le 1 / \beta_n + 1/n$
  \black
  we have $\kappa_n
  \rightarrow 0$ and $\kappa_n^{-2} \log \kappa_n^{-1} = o(n \rho_n)$, which is
  precisely the assumption required by Theorem~3.1 in
  \cite{borgs2015consistent}.
\end{rem}

\subsection{A Universal Lossless Compression scheme adapted to the Local Weak
  Convergence Framework}
\label{sec:prelim-baby-compression}

In this section we review the compression scheme introduced by the authors
in~\cite{delgosha2020universal}.
This scheme yields a universal lossless compression for a sequence of
sparse graphs converging to a limit in the local weak sense, without knowing
a priori what that limit is. 
The compression scheme in~\cite{delgosha2020universal}
allows for the graphs to be marked, i.e.\ vertices and edges in the graph can
carry additional marks on top of the connectivity structure of the graph.
However, since we do not include marks in our discussion here, we present the
results of \cite{delgosha2020universal} reduced to our unmarked setting.

More precisely, we introduce a compression map $\fnlwc: \mG_n \rightarrow
\{0,1\}^* - \emptyset$ which assigns 
a codeword to each graph on
the vertex set $[n]$ in a prefix-free way.
Here, the superscript $\mathsf{lwc}$ stands for ``local
weak convergence'', and is assigned to distinguish it from the compression map
we will introduce later in Section~\ref{sec:coding-scheme}. This compression
scheme is
lossless, i.e.\ there exists a decompression map $\gnlwc$ such that $\gnlwc
\circ \fnlwc$ is the identity map. Moreover, the compression scheme is
universal in the sense that  given a sequence of graphs $\Gn$ converging to a
limit  $\mu \in \mP_u(\mG_*)$ in the local weak sense where $\deg(\mu) \in (0,\infty)$, without a priori knowledge of $\mu$, we  have
\begin{equation*}
  \limsup_{n \rightarrow \infty} \frac{\nat(\fnlwc(\Gn)) - \mn \log n}{n} \leq \bch(\mu).
\end{equation*}
Here, $\mn$ is the number of the edges in $\Gn$, and normalization is done in a
way consistent with the definition of the BC entropy in  Section~\ref{sec:bc-ent}.

It can be shown that the  compression scheme described below satisfies the above properties.

\begin{enumerate}
\item Given the input graph $\Gn$, define the set
  \begin{equation*}
    Y_n := \{v \in [n]: \deg_{\Gn}(v) > \log \log n \text{ or } \deg_{\Gn}(w) > \log \log n \text{ for some } w \sim_{\Gn} v\}.
  \end{equation*}
  \item Let $\tGn$ be the subgraph of $\Gn$ obtained by removing all the edges
    $(v,w)$ in $\Gn$ such that either $\deg_{\Gn}(v) > \log \log n$, or
    $\deg_{\Gn}(w) > \log \log n$.
  \item We first compress $\tGn$ as follows.
    \begin{enumerate}
    \item Let $h_n := \sqrt{\log \log n}$ and define $\mA_n \subset \mG_*$ to be
      the set of $[G,o] \in \mG_*$ with
      depth at most 
      $h_n$
      where the degree of 
      each vertex in $G$
      is at most
      $\log \log n$. 
      Note that for all $v \in [n]$, we have
      $[\tGn,v]_{h_n} \in \mA_n$.
      \item For each $[G,o] \in \mA_n$, encode the number of vertices $v$ in
        $\tGn$ 
        such that $[\tGn,v]_{h_n} = [G,o]$ 
        using $1 + \lfloor \log_2 n
        \rfloor$ bits. In other words, we encode the appearance frequency   of each
        possible local neighborhood in $\tGn$.
        \item Let $W_n$ denote the set of graphs $G' \in \mG_n$ with degrees
          bounded by $\log \log n$ such that for all
          $[G,o] \in \mA_n$, we have
          \begin{equation*}
            |\{v \in [n]: [G',v]_{h_n} = [G,o]  \}| = |\{v \in [n]: [\tGn,v]_{h_n} = [G,o]  \}|.
          \end{equation*}
          In other words, $W_n$ is the set of graphs with the same 
          appearance frequency 
          of local
          structures as in $\tGn$. Note that $\tGn \in W_n$, and we can encode
          $\tGn$ by specifying it among the elements of $W_n$ using $1 + \lfloor
          \log_2 |W_n| \rfloor$ bits.
        \end{enumerate}
      \item Now, it remains to encode those edges present in $\Gn$ but not in
        $\tGn$, i.e.\ those edges which were removed during the truncation step
        2 above. Let $Z_n$ denote the set of such edges.  Note that, by
        definition, for every edge $(v,w) \in Z_n$ we have $v \in Y_n$ and $w
        \in Y_n$. 
        We first encode the set $Y_n$ 
        by encoding $|Y_n|$  using $1 + \lfloor \log_2 n
        \rfloor$ bits, and then 
        encoding $Y_n$ among the set of all subsets of
        $[n]$ with the same size using $1 + \lfloor \log_2
        \binom{n}{|Y_n|} \rfloor$ bits.
      \item Let $\mn$ and $\tmn$ denote the number of edges in $\Gn$ and $\tGn$
        respectively. Therefore, the set $Z_n$ consists of $\mn - \tmn$ many
        edges, and both of the endpoints of each such edge are in $Y_n$.
        Thereby, having encoded $Y_n$ in the previous steps, we can encode the
        $\mn - \tmn$ remaining edges in $Z_n$ using $1+ \lfloor \log_2
        \binom{\binom{|Y_n|}{2}}{\mn - \tmn} \rfloor$ bits.
\end{enumerate}

It can be shown that the above compression scheme is indeed universal in the
following sense. 

\begin{thm}[Theorem~3 in \cite{delgosha2020universal}]
  \label{thm:prelim-baby-compression}
  Given any unimodular $\mu \in \mP_u(\mT_*)$ such that $\deg(\mu) \in
  (0,\infty)$, and a sequence of graphs $\Gn$ converging to $\mu$ in the local
  weak sense, we have 
  \begin{equation*}
    \limsup_{n \rightarrow \infty} \frac{\nat(\fnlwc(\Gn)) - \mn \log n}{n} \leq \bch(\mu),
  \end{equation*}
where $\mn$ denotes the number of edges in $\Gn$.
\end{thm}


Together with the following converse result, this implies that the BC entropy is
indeed the correct information-theoretic limit for compression on a per-edge basis
in the local weak convergence framework.
\black

\begin{thm}[Theorem 4 in\cite{delgosha2020universal}]
  \label{thm:baby-converse}
  Assume that a lossless  compression scheme $((f_n, g_n): n \geq 1)$ is given. Fix
  some unimodular $\mu \in \mP_u(\mT_*)$ such that $\deg(\mu) \in (0,\infty)$
  and $\bch(\mu) > -\infty$. Then there exists a sequence of random graphs
  $(\Gn: n \geq 1)$ defined on a joint probability space such that $\Gn$
 converges a.s.\ to $\mu$ in the local weak sense and 
  \begin{equation*}
    \liminf_{n \rightarrow \infty} \frac{\nat(f_n(\Gn)) - \mn \log n}{n} \geq \bch(\mu) \qquad \text{a.s.},
  \end{equation*}
  where $\mn$ denotes the number of edges in $\Gn$. 
\end{thm}

The following bound will also be useful for our future analysis.

\begin{lem}
  \label{lem:prelim-baby-bounded-graph-codeword-bound}
  Assume that  
  the sequence $(\Gn)_{n \geq 1}$ is given where,
    for all $n \geq 1$,
    $\Gn \in \mG_n$ 
    and
    $\deg_{\Gn}(v) \leq \log
  \log n$ for all $v \in [n]$. Then we have the following bound on the codeword
  length associated to $\Gn$:
  \begin{equation*}
    \nat(\fnlwc(\Gn)) \leq  \log \binom{\binom{n}{2}}{\mn} + o(n),
  \end{equation*}
  where $\mn$ denotes the number of edges in $\Gn$, and the $o(n)$ term does not
  depend on
    $(\Gn)_{n \geq 1}$.
\end{lem}

\begin{proof}
Following the compression scheme that we discussed above, since all degrees in $\Gn$ are bounded by $\log \log n$, we have $Y_n =
\emptyset$ and $\tGn = \Gn$. Therefore, the number of bits required to encode
the set $Y_n$ in step 4 is bounded by $2 + \lfloor \log_2 n \rfloor \leq 2 +
\log_2 n$. Also,  since $Z_n = \emptyset$, step 5 does not contribute any bits to
the output codeword. Now, we find a bound on the number of bits required to
encode $\tGn = \Gn$ in step 3. Note that we use $|\mA_n| (1 + \lfloor \log_2 n
\rfloor)$ bits in part 3b. But from Lemma~7
in \cite{delgosha2020universal}, we have $|\mA_n| = o(n / \log n)$. Thereby, $|\mA_n| (1 + \lfloor \log_2 n
\rfloor) = o(n)$.
Observe that since the sequence of sets $(\mA_n)_{n \ge 1}$ does not depend on $(\Gn)_{n \ge 1}$,
this $o(n)$ term also does not depend on 
$(\Gn)_{n \ge 1}$.
Now, we find a bound on the size
of the set $W_n$ defined in step 3c. 
We first claim that all the graphs in $W_n$ have                                                                           %
precisely $\mn$ edges. To see this, take $\pp{G} \in W_n$ and note that since all the                                %
degrees in $\pp{G}$ are bounded by $\log \log n$, we have $[\pp{G},v]_{h_n} \in \mA_n$                               %
for all $v \in [n]$. Consequently, we have                                                                           %
\begin{align*}                                                                                                       %
\sum_{v=1}^n \deg_{G'}(v) &= \sum_{[G,o] \in \mA_n} \sum_{v: [\pp{G},v]_{h_n} = [G,o]} \deg_{\pp{G}}(v) \\         %
                           &= \sum_{[G,o] \in \mA_n} \sum_{v: [\pp{G},v]_{h_n} = [G,o]} \deg_{G}(o) \\              %
                            &\stackrel{(a)}{=} \sum_{[G,o] \in \mA_n} \sum_{v: [\Gn,v]_{h_n} = [G,o]} \deg_{G}(o) \\ %
                            &=  \sum_{[G,o] \in \mA_n} \sum_{v: [\Gn,v]_{h_n} = [G,o]} \deg_{\Gn}(o) \\              %
&= \sum_{v=1}^n \deg_{\Gn}(v).                                                                                       %
\end{align*}                                                                                                         %
where $(a)$ uses the fact that by definition, $|\{v: [\Gn,v]_{h_n} = [G,o]\}| =                                      %
|\{v: [G',v]_{h_n} = [G,o]\}|$ for all $[G,o] \in \mA_n$. This means that the                                        %
number of edges in $\pp{G}$ is precisely $\mn$.     
\black
As a result, we have  $|W_n| \leq
\binom{\binom{n}{2}}{\mn}$. Hence, the number of bits we use in step 3c in order
to specify $\tGn$ among the graphs in $W_n$ is at most $1 + \log_2 \binom{\binom{n}{2}}{\mn}$.
Putting all the above together and multiplying by $\log 2$ to convert bits to
nats, we get
\begin{equation*}
  \nat(\fnlwc(\Gn)) \leq 2 + \log n + 1 + \log\binom{\binom{n}{2}}{\mn} + o(n) \leq \log\binom{\binom{n}{2}}{\mn} + o(n).
\end{equation*}
This completes the proof.
\end{proof}

\editfinish

\section{A Notion of Entropy for the Sparse Graphon Framework}
\label{sec:graphon-entropy-new}

\editstart

In the dense regime, the asymptotic behavior of the entropy of dense
graph ensembles generated by a graphon has been extensively studied in the
literature (see, for instance, \cite{janson2010graphons}). We first briefly review this notion before
focusing on the sparse regime.

To remain consistent with the notation we defined
in Section~\ref{sec:prelim-graphon}, let $W:[0,1] \times [0,1] \rightarrow
[0,1]$ be a graphon with values bounded by 1, but not necessarily normalized.
Also, consider the sequence of $W$-random graphs $G_n \sim \mG(n; \rho_n W)$
where the target density $\rho_n$ is set to 1 for all $n$. This yields a
sequence of dense graphs. It can be shown that the entropy of this  sequence of random graphs
has the following asymptotic behavior:
\begin{equation}
  \label{eq:dense-graphon-entropy}
  \lim_{n \rightarrow \infty}\frac{H(G_n)}{\binom{n}{2}} = \int_{[0,1] \times [0,1]} H_b(W(x,y)) dx dy,
\end{equation}
where $H_b(x) := -x \log x - (1-x) \log(1-x)$ for $x \in [0,1]$ (see, for
instance, \cite[Theorem D.5]{janson2010graphons}). In the following, we study the
analog of this question in the sparse regime, i.e.\ we study the asymptotic behavior
of the entropy of 
a sequence of 
\black
$W$-random graphs when $W$ is not necessarily bounded, but is
an $L^2$ graphon, and the sequence of target densities $\rho_n $ is such that $\rho_n
\rightarrow 0$ and $n \rho_n \rightarrow \infty$.
To the best of our knowledge, this is the first instance of such 
an
\black
analysis.

As we will see, in the sparse regime the asymptotic behavior of the entropy is
quite different from~\eqref{eq:dense-graphon-entropy}. For one thing, since $W$
no longer necessarily takes values in 
$[0,1]$,
\black
the right hand side of~\eqref{eq:dense-graphon-entropy} is no
longer meaningful. In fact, the following definition turns out to be useful for
our analysis in the sparse regime.

\begin{definition}
  \label{def:ent-graphon}
  For  an $L^2$ graphon $W$ over a probability space $(\Omega, \mF, \pi)$, we
  define $\Ent(W)$ as follows
  \begin{equation}
    \label{eq:ent-w-def}
    \Ent(W) := \int W(x,y) \log W(x,y) d\pi(x) d\pi(y) - \left( \int W(x,y) d\pi(x) d\pi(y) \right) \log \left( \int W(x,y) d\pi(x) d\pi(y) \right).
  \end{equation}
Here, as usual, we have $0 \log 0 = 0$.
\end{definition}

Viewing $W$ as a nonnegative random variable on the space $\Omega \times \Omega$
equipped with the product measure $\pi \times \pi$, we may write
\begin{equation*}
  \Ent(W) = \ev{W \log W} - \ev{W} \log \ev{W}.
\end{equation*}
This is in fact the so called \emph{entropy functional} associated to $W$
(see, for instance,  \cite[page 96]{boucheron2013concentration}).
Note that when $W$ is a normalized graphon, we have $\ev{W} = 1$, and  $\Ent(W) =
\ev{W \log W}$. In fact, every normalized graphon $W$ corresponds to a
probability measure $\nu$  on $\Omega \times \Omega$ which is defined through
the relation $\frac{d\nu}{d (\pi
  \times \pi)}= W$.
With this, for such a normalized graphon we may write
\begin{equation}
  \label{eq:normalized-graphon-ent-KL-div}
  \Ent(W) = D(\nu \Vert \pi \times \pi).
\end{equation}

Before studying the asymptotics of the entropy of 
a sequence of
\black
$W$-random graphs, 
we state  some properties for this entropy in the following
Theorem~\ref{thm:graphon-entropy-props}, whose proof  is given in Appendix~\ref{sec:app-graphon-ent-properties}.

\begin{thm}
  \label{thm:graphon-entropy-props}
  Assume that $W$ is an $L^2$ graphon on a probability space $(\Omega, \mF,
  \pi)$. Then the following hold for the 
  notion of entropy above:
  \black
  \begin{enumerate}
  \item 
  $\Ent(W)$ is well defined and $\Ent(W) <
    \infty$. 
    \black \label{thm-ge-well}
  \item $\Ent(W) \geq 0$. \label{thm-ge-Ent-pos}
  \item For $\alpha>0 $ we have $\Ent(\alpha W) = \alpha \Ent(W)$. \label{thm-ge-scale}
  \item Assume that a sequence $W_n$ of $L^2$ graphons over  
  $(\Omega_n, \mF_n, \pi_n)$ 
  is given
  \black
  such that $\delta_2(W_n, W) \rightarrow
  0$ as $n \rightarrow \infty$. Then we have $\Ent(W_n) \rightarrow \Ent(W)$ as
  $n \rightarrow \infty$. \label{thm-ge-L2-conv}
  \end{enumerate}
\end{thm}

In the following Proposition~\ref{prop:graphon-entropy-asymptotics}, we discuss the relation between the entropy of a normalized
graphon $W$ and the asymptotic behavior of the entropy of 
a sequence of
\black
$W$-random graphs.
The proof of Proposition~\ref{prop:graphon-entropy-asymptotics} will be given in Appendix~\ref{sec:app-graphon-random-graph-asymp-entropy}.

\begin{prop}
  \label{prop:graphon-entropy-asymptotics}
  Assume that $W$ is a normalized $L^2$ graphon over $(\Omega, \mF, \pi)$. Also, assume that $\Gn \sim \mG(n; \rho_n W)$
  is a sequence of $W$-random graphs with target density $\rho_n$ such that 
  $\rho_n \rightarrow 0$ and
  $n \rho_n \rightarrow \infty$.
  \black
  Then,  with $\barmn :=
  \binom{n}{2} \rho_n$,  we have 
  \begin{equation*}
    \lim_{n \rightarrow \infty} \frac{H(\Gn) - \barmn \log \frac{1}{\rho_n}}{\barmn} = 1 - \Ent(W).
  \end{equation*}
\end{prop}

Recall from 
Theorem~\ref{thm:ganguli-W-random-graphon-convergence} that we have
\black
$\rho(\Gn) / \rho_n \rightarrow 1$ a.s.. Therefore, if $\mn$ denotes the number
of edges in $\Gn$, this implies that $\mn / \barmn \rightarrow 1$ a.s.. 
We also have $\ev{\mn} / \barmn \rightarrow 1$.
To see this, note that
\begin{equation}
  \label{eq:ev-mn-mbarn}
\begin{aligned}
  \ev{\mn} &= \ev{\ev{\mn|X_{[1:n]}}} \\
           &= \ev{\sum_{1 \leq i < j \leq n} (\rho_nW(X_i,X_j)\wedge 1)}\\
           &= \rho_n \binom{n}{2} \ev{W \wedge \frac{1}{\rho_n}} \\
           &= \barmn \ev{W \wedge \frac{1}{\rho_n}},
         \end{aligned}
       \end{equation}
and, since $\rho_n \rightarrow 0$ as $n \rightarrow \infty$, we have  $\ev{W \wedge
  1/\rho_n} \rightarrow \ev{W} = 1$, which implies that $\ev{\mn} / \barmn
\rightarrow 1$ as $n \rightarrow \infty$.
 Motivated by this, we
can think of $\barmn$ as, roughly speaking, the ``nominal'' number of edges in
$\Gn$.


\begin{rem}
  Note that unlike the asymptotic in the dense regime as 
in~\eqref{eq:dense-graphon-entropy}, in the sparse regime of 
Proposition~\ref{prop:graphon-entropy-asymptotics} above the entropy of $\Gn$
has a leading term which is $\barmn \log 1 /\rho_n$, and $\Ent(W)$  
appears at the second order term. 
\end{rem}

\begin{rem}
  From part~\ref{thm-ge-Ent-pos} in Theorem~\ref{thm:graphon-entropy-props}, we have $1 - \Ent(W) \leq 1$. Also, $1 - \Ent(W) = 1$ when $W =1$
  almost everywhere. This means that, by fixing the sequence of target densities
  $\rho_n$, among the normalized $L^2$ graphons the constant graphon $W=1$, which corresponds to
  the measure $\pi \times \pi$ on $\Omega \times \Omega$, has the maximum
  asymptotic entropy rate. Moreover, comparing 
  to~\eqref{eq:normalized-graphon-ent-KL-div}, for a normalized $L^2$ graphon
  $W$ the amount by which the  asymptotic entropy rate
  deviates from this maximum value is precisely the  divergence
  between the measure corresponding to $W$ and the product measure $\pi \times \pi$ on
  $\Omega \times \Omega$, which corresponds to the constant graphon with value 1.
\end{rem}

\editfinish

\section{Problem statement and main results}
\label{sec:statement-and-results}

\editstart

In this section we formalize the problem of finding 
an optimal
\black
universal compression
scheme which is capable of compressing 
a sequence of 
\black
sparse graphs
which is
\black
either convergent in the local weak sense 
as we discussed in Section~\ref{sec:prelim-lwc},
or is
\black
generated as a sequence of $W$-random graphs as we discussed in Section~\ref{sec:prelim-graphon},
the compression being information-theoretically optimal on a per-edge basis.
\black

More precisely, for each integer $n$, we want to design a compression map $f_n :
\mG_n \rightarrow \{0,1\}^* - \emptyset$ which assigns a prefix--free codeword
to every graph on the vertex set $[n]$, as well as a decompression map $g_n$, such
that $g_n \circ f_n$ is the identity map, i.e.\ we have lossless compression. In
addition to this, we want this compression scheme to be 
information-theoretically
\black
\emph{optimal}.
More precisely, assume that
we have a sequence of random graphs $\Gn$ where either $\Gn$ converges in the
local weak sense to a unimodular measure $\mu \in \mP_u(\mT_*)$ with probability
one, or $\Gn$  is a sequence of $W$-random graphs with target densities
$\rho_n$ for a normalized $L^2$ graphon $W$. However, the encoder does not know which of the two cases holds, nor
does it know the limit objects $\mu$ or $W$ in each case, or even the sequence of
target densities $\rho_n$
in the latter case. 
\black
Nonetheless, we want the compression scheme to be
universally optimal. 
Motivated by our discussion of the BC entropy in
Section~\ref{sec:bc-ent}, this means that if $\Gn$ converges in the local weak sense to $\mu
\in \mP_u(\mT_*)$ with probability one, then we want 
\begin{equation}
  \label{eq:goal-lwc-optimal}
  \limsup_{n \rightarrow \infty} \frac{\nat(f_n(\Gn)) - \mn \log n}{n} \leq \bch(\mu) \qquad \qquad \text{a.s.}.
\end{equation}
    Here, $\mn$ denotes the number of edges in $\Gn$, and the normalization of
    the codeword length is performed in a way consistent with the definition of
    the BC entropy in Section~\ref{sec:bc-ent}.
    Moreover,  motivated by the discussion of the notion of entropy for the
    sparse graphon framework  in
    Section~\ref{sec:graphon-entropy-new}, and in particular
    Proposition~\ref{prop:graphon-entropy-asymptotics} therein, 
if $\Gn \sim \mG(n; \rho_n W)$ for a normalized $L^2$ graphon $W$ and a sequence
of target densities $\rho_n$ with $\rho_n \rightarrow 0$ and $n \rho_n
\rightarrow \infty$, then we want 
  \begin{equation}
    \label{eq:goal-graphon-optimal}
    \limsup_{n \rightarrow \infty} \frac{\nat(f_n(\Gn)) - \barmn \log \frac{1}{\rho_n}}{\barmn} \leq 1 - \Ent(W) \qquad \qquad \text{a.s.},
  \end{equation}
where $\barmn := \binom{n}{2} \rho_n$. 

  Note that in this setup the encoder  only observes the graph
  realization $\Gn$ and not the whole sequence $(\Gn: n \geq 1)$. Moreover, as
  we discussed above, the
  encoder does not a priori know from which of the two ensemble types the
  realization  $\Gn$ comes from, nor does it know the limit objects for each of
  the two sequences of ensembles.

  We address this problem by introducing a compression scheme, and will further
  discuss a converse result. Our compression scheme employs an \emph{splitting}
  method. More precisely, given a graph realization $\Gn \in \mG_n$ as an input to the encoder, we
  choose a splitting parameter $\Delta_n$, and split $\Gn$ into two graphs denoted by
  $\Gn_{\Delta_n}$ and $\Gn_*$. These two graphs are both on the vertex set $[n]$, and
  each edge in $\Gn$ appears in precisely one of them. More precisely,
  $\Gn_{\Delta_n}$ consists of  those edges $(v,w)$ from $\Gn$ such that 
  each of their endpoints has degree 
  at  most $\Delta_n$. 
  \black
  We then define $\Gn_*$ to
  include 
  those edges in $\Gn$ which do not appear in $\Gn_{\Delta_n}$, i.e.\ those
  edges where the degree of at least one of their endpoints is bigger than
  $\Delta_n$. We then encode each of these two graphs separately, where the
  details are provided in  Section~\ref{sec:coding-scheme}. Roughly speaking, the splitting
  parameter is chosen so that when $\Gn$ is
  coming from a local weak convergence  ensemble $\Gn_{\Delta_n}$ contains most
  of the edges in $\Gn$,
  while when  $\Gn$ is coming from a sparse graphon
  ensemble, $\Gn_*$ contains most of the edges in $\Gn$. In order to emphasize the
  dependence of the compression and the decompression maps on the parameter
  $\Delta_n$ we denote these mappings by $\fnDn$ and $\gnDn$, respectively.

We will explain the details of this compression scheme in
Section~\ref{sec:coding-scheme}. However, in the following, we state how the
choice of the parameter $\Delta_n$ affects the asymptotic normalized codeword length
associated to our compression method in each of the 
two (local weak convergence and sparse graphon)
\black
regimes.
We do this in Propositions~\ref{prop:lwc-part} and \ref{prop:graphon-optimality}
below.
Although in the sequel we will mainly fix the splitting parameter $\Delta_n$
prior to observing the realization $\Gn$, the analysis in
Propositions~\ref{prop:lwc-part} and \ref{prop:graphon-optimality} is general in
the sense that $\Delta_n$ can be chosen after observing $\Gn$.
First, in Proposition~\ref{prop:lwc-part} below, we study the local weak
convergence scenario. Note that although in this setting we assume that the
sequence of random graphs $\Gn$ is convergent in the local weak sense, it is
not necessarily the case  that the sequence of truncated graphs $\Gn_{\Delta_n}$ also
converges to the same limit. In fact, the following proposition states that the
asymptotic behavior of the codeword length depends on the local weak limit of
this truncated sequence $\Gn_{\Delta_n}$, if such a limit exists. We define
$R_n$ to be the set of vertices $1 \leq v \leq n$ in $\Gn$ such that either
$\deg_{\Gn}(v) > \Delta_n$ or $\deg_{\Gn}(w) > \Delta_n$ for some $w \sim_{\Gn}
v $. The proof of Proposition~\ref{prop:lwc-part} below is given in Section~\ref{sec:lwc-analysis}.

\begin{prop}
  \label{prop:lwc-part}
  Assume that $\mu \in \mP_u(\mT_*)$ is given such that $\deg(\mu) \in
  (0,\infty)$. Furthermore, assume that  $\Gn$  is a sequence of random graphs converging to $\mu$ in the local weak
  sense, i.e.\ $U(\Gn) \Rightarrow \mu$ a.s..
  If
  \black
  the parameter $\Delta_n$ is
  chosen so that $U(\Gn_{\Delta_n}) \Rightarrow \nu$ a.s.\  for some $\nu \in
  \mP_u(\mT_*)$ with $\deg(\nu) \in (0,\infty)$, and $|R_n| / n \rightarrow
  \eta$ a.s.\ for some $\eta \geq 0$, then, with probability one, we have 
  \begin{equation}
    \label{eq:prop-lwc-part}
    \limsup_{n \rightarrow \infty} \frac{\nat(\fnDn(\Gn)) - \mn \log n}{n} \leq \bch(\nu) +  (25 + e/2) \eta + H_b(\eta),
  \end{equation}
  where $\mn$ denotes the number of edges in $\Gn$, 
  and $\fnDn$
    refers to the compression scheme of Section~\ref{sec:coding-scheme}.
\end{prop}

As we will discuss later, 
  if $\Delta_n \rightarrow
  \infty$ a.s.
  \black 
  then with probability one we have $U(\Gn_{\Delta_n}) \Rightarrow
  \mu$ and $|R_n| / n \rightarrow 0$. Therefore, the right hand side
  of~\eqref{eq:prop-lwc-part} in Proposition~\ref{prop:lwc-part} above becomes
  $\bch(\mu)$. Comparing this with~\eqref{eq:goal-lwc-optimal}, we realize that
  choosing $\Delta_n$ 
  deterministically (i.e. in a way that depends only on $n$)
\black
so that $\Delta_n \rightarrow \infty$ as $n
  \rightarrow \infty$ is  reasonable  from the local weak convergence perspective.
However, motivated by  Proposition~\ref{prop:graphon-optimality} below, roughly
speaking, if
$\Delta_n$ goes to infinity ``too fast'', then we may lose the optimality
condition~\eqref{eq:goal-graphon-optimal} in the sparse graphon scenario. In other
words, there is a trade-off between the two regimes in terms of choosing the
parameter $\Delta_n$. Next, we state the asymptotic behavior of the codeword
length in the sparse graphon regime. 
We denote the number of edges in $\Gn_{\Delta_n}$ and $\Gn_*$ by
$\mn_{\Delta_n}$ and $\mn_*$, respectively.
The proof of
Proposition~\ref{prop:graphon-optimality} below is given in
Section~\ref{sec:graphon-analysis}. 

\begin{prop}
  \label{prop:graphon-optimality}
  Assume 
  that
  \black
  $W$ is a normalized $L^2$  graphon on a probability space $(\Omega, \mF, \pi)$. Let
  $\Gn \sim \mG(n; \rho_n W)$ be a sequence of $W$--random graphs with target
  density $\rho_n$, such that 
  $\rho_n \rightarrow 0$ and 
  $n \rho_n \rightarrow \infty$.
  \black
Furthermore, assume that the sequence $\Delta_n$ is chosen so that we
have $\mn_{\Delta_n} / \barmn \rightarrow 0$ a.s., $\Delta_n \leq \log \log n$
for $n$ large enough a.s., and
$\Delta_n / \sqrt{n \rho_n} \rightarrow 0$ a.s.. 
  Then, with probability one,  we have
  \begin{equation*}
  \limsup_{n \rightarrow \infty} \frac{\nat(\fnDn(\Gn)) -\barmn \log \frac{1}{\rho_n}}{\barmn} \leq 1 - \Ent(W),
\end{equation*}
where $\fnDn$
    refers to the compression scheme of Section~\ref{sec:coding-scheme}.
\end{prop}

This proposition requires that 
$\Delta_n$ must not grow faster
than $\sqrt{n \rho_n}$. 
\black
Note that the encoder does not have any a priori
knowledge of the sequence $\rho_n$. In fact, it turns out that it is impossible
to simultaneously satisfy the above conditions imposed by both
Propositions~\ref{prop:lwc-part} and \ref{prop:graphon-optimality}. In
particular, we
show that any general splitting mechanism, which does not even necessarily truncate
the graph using the parameter $\Delta_n$ above, cannot satisfy even a subset of
the conditions imposed by Propositions~\ref{prop:lwc-part} and
\ref{prop:graphon-optimality}. This is the purpose of
Proposition~\ref{prop:no-good-splitting}  below. Before that, we need to define
what we mean by a general splitting mechanism.

Given an integer $n$, we define a splitting mechanism for graphs in $\mG_n$
  as a pair of functions $\Tn_1 : \mG_n \rightarrow \mG_n$ and $\Tn_2 :
\mG_n \rightarrow \mG_n$, such that for all $G \in \mG_n$, the superposition of
$\Tn_1(G)$ and $\Tn_2(G)$ is $G$, and each edge in $G$ appears in precisely one
of the two graphs $\Tn_1(G)$ or $\Tn_2(G)$. A special case  of such a splitting
mechanism is $\Tn_1(\Gn) = \Gn_{\Delta_n}$ and $\Tn_2(\Gn) = \Gn_*$ given a
splitting parameter $\Delta_n$, as we discussed above.

\begin{definition}
  \label{def:good-splitting}
We say that a sequence of splitting mechanisms
$((\Tn_1, \Tn_2): n \geq 1)$ is good if the following two conditions hold:
\begin{enumerate}
\item For any unimodular $\mu \in \mP_u(\mT_*)$ with $\deg(\mu) \in (0,\infty)$,
  and any sequence of random graphs $\Gn$ converging with probability one to
  $\mu$ in the local weak sense, $\Tn_1(\Gn)$ also converges with probability
  one to $\mu$.
\item For any normalized $L^2$ graphon $W$, and any sequence $\rho_n$ such that
  $\rho_n \rightarrow 0$ and 
  $n \rho_n \rightarrow \infty$,
  \black
  if $\Gn \sim \mG(n;
  \rho_n W)$ is the sequence of $W$-random graphs with target density $\rho_n$, with $\mn_1$ being
  the number of edges in $\Tn_1(\Gn)$, then we have $\mn_1
  / \barmn \rightarrow 0$ a.s.,
  where $\barmn := {n \choose 2} \rho_n$.
  \black
\end{enumerate}
\end{definition}

The first condition above is motivated by Proposition~\ref{prop:lwc-part} above
so that the limit $\nu$ in that proposition is the same as $\mu$, and the second condition is
motivated by Proposition~\ref{prop:graphon-optimality}.

\begin{prop}
  \label{prop:no-good-splitting}
  There does not exists a sequence of good splitting mechanisms. 
\end{prop}

We give the proof of Proposition~\ref{prop:no-good-splitting} in Appendix~\ref{app:no-good-splitting-proof}.\footnote{
In fact, one can show that the proof in
Proposition~\ref{prop:no-good-splitting} still holds  even if we allow for random
splitting mechanisms. We  have restricted $\Tn_1$ and $\Tn_2$  to be
deterministic mainly to simplify the presentation and because this
suffices for our purpose, which is to motivate the introduction of a sequence 
$(a_n, n \ge 1)$ for which we require that $n \rho_n \geq \gran$ for all $n \ge 1$.
\black
}
Roughly
speaking, the reason why there does not exists a sequence of good splitting
mechanisms is that the sequence $\rho_n$ can be chosen such that $n\rho_n$ goes
to infinity arbitrarily slowly, and this confuses the splitting mechanism and
prevents it from being able  to
distinguish between the local weak convergence and the sparse graphon regimes. 

Motivated by the above discussion, we restrict the sequence $\rho_n$ such that
$n \rho_n$ does not  go to infinity arbitrarily slowly. More precisely, we assume
that $n \rho_n \geq \gran$ where $(a_n: n \geq 1)$ is a sequence known a priori to
both the encoder and the decoder such that $a_n \rightarrow \infty$ as $n
\rightarrow \infty$. In this case, we still assume that the encoder
does not know whether the input graph $\Gn$ is coming from an ensemble
consistent with the local weak convergence convergence framework or the sparse
graphon framework, nor does it know the limit objects in each case. However,
both the encoder and the decoder know that \emph{if} $\Gn$ is a realization of a
sparse graphon ensemble, the unknown target density $\rho_n$ is such that $n \rho_n \geq \gran$.
We show that under this assumption
information-theoretically
\black
optimal 
\black
universal compression 
on a per-edge basis
\black
can be achieved by
appropriately choosing the sequence of  splitting parameters $\Delta_n$.

\begin{thm}
  \label{thm:dual-main-an}
  Let $(\gran : n \geq 1)$ be a sequence known
  to both the encoder and the decoder such that $a_n \rightarrow \infty$ as $n
  \rightarrow \infty$. 
  Assume that $n \rho_n \geq \gran$.
  \black
  Then, there exists an appropriate choice for the
  sequence of splitting parameters $(\Delta_n: n \geq 1)$,
  with $\Delta_n$ depending only on $n$,
  \black
  such that our sequence of compression schemes $((\fnDn, \gnDn): n \geq 1)$,
  which is introduced in Section~\ref{sec:coding-scheme},
achieves optimal universal compression in the
sense discussed above. 
More precisely,  we have
\begin{enumerate}
  \item If $\Gn$ is a sequence
  of random graphs such that, almost surely, $U(\Gn) \Rightarrow \mu$ for some
  $\mu \in \mP_u(\mT_*)$ with $\deg(\mu) \in (0,\infty)$,  we have
  \begin{equation}
    \label{eq:dual-main-lwc}
    \limsup_{n \rightarrow \infty} \frac{\nat(\fnDn(\Gn)) - \mn \log n}{n} \leq \bch(\mu) \qquad \qquad \text{a.s.},
  \end{equation}
  where $\mn$ denotes the number of edges in $\Gn$.
  \item On the other hand, if $\Gn \sim
  \mG(n; \rho_n W)$ is a
  sequence of $W$-random graphs with target densities $\rho_n$, where $W$ is  a normalized  $L^2$ graphon, assuming
  that $\rho_n \rightarrow 0$ as $n \rightarrow \infty$, $n \rho_n \rightarrow
  \infty$ as $n \rightarrow \infty$, and
  $n \rho_n \geq \gran$, with $\barmn := \binom{n}{2} \rho_n$, we have
  \begin{equation}
    \label{eq:dual-main-graphon}
    \limsup_{n \rightarrow \infty} \frac{\nat(\fnDn(\Gn)) - \barmn \log \frac{1}{\rho_n}}{\barmn} \leq 1 - \Ent(W) \qquad \qquad \text{a.s.}.
  \end{equation}
  \end{enumerate}
  In the above setting, the encoder and the decoder only know the sequence $a_n$,
  and do not know from which of the two settings the  input graph $\Gn$ is
  generated,  neither do they know the
  limit objects $\mu$ or $W$ in each setting.
\end{thm}

\begin{proof}[Proof of Theorem~\ref{thm:dual-main-an}]
  We choose $\Delta_n = \min\{\log \gran, \log \log n\}$. Since $a_n \rightarrow
  \infty$ as $n \rightarrow \infty$, we have $\Delta_n \rightarrow \infty$.
  Assume that  $\Gn$ is such that almost surely $U(\Gn) \Rightarrow \mu$ for some $\mu \in
  \mP_u(\mT_*)$ with $\deg(\mu) \in (0,\infty)$. Then, since $\Delta_n
  \rightarrow \infty$, Lemma~6 in \cite{delgosha2020universal} implies that
  $U(\Gn_{\Delta_n}) \Rightarrow \mu$ a.s.. Moreover, Lemma~8 in
  \cite{delgosha2020universal} implies that $|R_n| / n \rightarrow 0$ a.s..
  Consequently, \eqref{eq:dual-main-lwc} follows from
  Proposition~\ref{prop:lwc-part}. Now, assume that $\Gn \sim \mG(n; \rho_n W)$
  for a normalized $L^2$ graphon $W$, and $n \rho_n \geq \gran$. We verify that the
  assumptions in Proposition~\ref{prop:graphon-optimality} hold. Note that
  clearly 
  $\Delta_n \leq \log \log n$. Also, since $\Delta_n \leq \log \gran$, $\gran
  \rightarrow \infty$, and $n \rho_n \geq \gran$, we have $\Delta_n / \sqrt{n
    \rho_n} \rightarrow 0$. On the other hand, since all the degrees in
  $\Gn_{\Delta_n}$ are bounded by $\Delta_n$, we have $\mn_{\Delta_n} \leq n
  \Delta_n / 2$. Thereby
  \begin{equation*}
    \frac{\mn_{\Delta_n}}{\barmn} \leq \frac{\Delta_n}{(n-1)\rho_n} \leq \frac{n}{n-1} \frac{\Delta_n}{a_n} \leq \frac{n}{n-1} \frac{\log \gran}{\gran} \rightarrow 0.
  \end{equation*}
  Hence, all the assumptions in Proposition~\ref{prop:graphon-optimality}
  hold, and \eqref{eq:dual-main-graphon} follows from Proposition~\ref{prop:graphon-optimality}.
\end{proof}

When a sequence of lower bounds $\gran$, as we discussed above, is not known, we
may choose the sequence $\Delta_n$ to be a constant, i.e.\ $\Delta_n = \Delta$
for some fixed $\Delta > 0$. Theorem~\ref{thm:main-dual-fixed} below suggests that if we do so, we
still have the universal optimality condition \eqref{eq:goal-graphon-optimal} in
the sparse graphon regime. However, the optimality
condition~\eqref{eq:goal-lwc-optimal} in the local weak convergence
framework holds in a weaker sense, i.e.\ it only holds after we send $\Delta$ to
infinity.

\begin{thm}
  \label{thm:main-dual-fixed}
  If  $\Delta_n = \Delta$ for $n \geq 1$ where $\Delta > 0$ is fixed, our
  sequence of compression schemes $((f_n^\Delta, g_n^\Delta): n \geq 1)$ has the
  following properties:
  \begin{enumerate}
  \item If $\Gn$ is a sequence
  of random graphs such that, almost surely, $U(\Gn) \Rightarrow \mu$ for some
  $\mu \in \mP_u(\mT_*)$ with $\deg(\mu) \in (0,\infty)$, we have
  \begin{equation}
    \label{eq:main-fixed-lwc}
    \limsup_{\Delta \rightarrow \infty} \limsup_{n \rightarrow \infty} \frac{\nat(f_n^\Delta(\Gn)) - \mn \log n}{n} \leq \bch(\mu) \qquad \qquad \text{a.s.},    
  \end{equation}
  where $\mn$ denotes the number of edges in $\Gn$.
  \item On the other hand, if $\Gn \sim
  \mG(n; \rho_n W)$ is a
  sequence of $W$-random graphs with target densities $\rho_n$, where $W$ is a
  normalized  $L^2$ graphon, and 
  $\rho_n \rightarrow 0$ and $n \rho_n \rightarrow
  \infty$ as $n \rightarrow \infty$, with $\barmn := \binom{n}{2} \rho_n$, 
  then
  \black
  for
  all $\Delta > 0$ we have 
  \begin{equation}
    \label{eq:dual-main-fixed-graphon}
    \limsup_{n \rightarrow \infty} \frac{\nat(f_n^\Delta(\Gn)) - \barmn \log \frac{1}{\rho_n}}{\barmn} \leq 1 - \Ent(W) \qquad \qquad \text{a.s.}.
  \end{equation}
  \end{enumerate}
\end{thm}

\begin{proof}[Proof of Theorem~\ref{thm:main-dual-fixed}]
  First assume that $\Gn$ is a sequence
  of random graphs such that, almost surely, $U(\Gn) \Rightarrow \mu$ for some
  $\mu \in \mP_u(\mT_*)$ with $\deg(\mu) \in (0,\infty)$.
  For $\Delta > 0$, define $\mu_\Delta \in \mP(\mT_*)$  to be the law of $[T^\Delta,o]$ when $[T,o] \sim \mu$.
  Here, $T^\Delta$ is the tree obtained from $T$ by removing all the edges where
  the degree of at least one of their endpoints is strictly bigger than
  $\Delta$, followed by taking the connected component of the root  $o$. It
  is easy to see that $U(\Gn) \Rightarrow \mu$
  a.s.\ implies that $U(\Gn_{\Delta}) \Rightarrow \mu_\Delta$ a.s., and
  $\mu_\Delta$ is unimodular.
  Furthermore, $\deg(\mu_\Delta) \leq \deg(\mu) < \infty$. Also, since
  $\deg(\mu) > 0$, if $\Delta$ is
  sufficiently large we have $\deg(\mu_\Delta) > 0$.
  Now, define the map $F_\Delta : \mG_* \rightarrow \{0,1\}$ such
  that $F([G,o]) = 1$ if $\deg_G(o) > \Delta$ or $\deg_G(v) > \Delta$ for some
  $v \sim_G o$, and $F([G,o]) = 0$ otherwise. With this, we have $|R_n| / n =
  \int F_\Delta d U(\Gn)$. Clearly, $F_\Delta$ is bounded.
  It is also continuous, because its value is determined by the depth-$2$ 
  neighborhood of the root.
  \black
  Therefore, the assumption $U(\Gn) \Rightarrow \mu$
  a.s.\ implies that with probability one, $|R_n| / n \rightarrow \int F_\Delta
  d \mu =:\eta_\Delta$. Consequently, if $\Delta$ is sufficiently large so that
  $\deg(\mu_\Delta) > 0$, using Proposition~\ref{prop:lwc-part}, we get
  \begin{equation*}
    \limsup_{n \rightarrow \infty} \frac{\nat(f_n^\Delta(\Gn)) - \mn \log n}{n} \leq \bch(\mu_\Delta) + (16+e/2) \eta_\Delta + H_b(\eta_\Delta).
  \end{equation*}
  The dominated convergence theorem implies that $\eta_\Delta \rightarrow 0$ as
  $\Delta \rightarrow \infty$. Moreover, from
  Proposition~\ref{prop:bc-ent-trunc-upper-sem-cont} in Section~\ref{sec:bc-ent}, we
  know that $\limsup _{\Delta \rightarrow \infty} \bch(\mu_\Delta) \leq
  \bch(\mu)$. Therefore, we arrive at~\eqref{eq:main-fixed-lwc} by sending
  $\Delta$ to infinity in the above inequality.

  Now, assume that $\Gn \sim
  \mG(n; \rho_n W)$ is a
  sequence of $W$-random graphs with target densities $\rho_n$, where $W$ is a  normalized  $L^2$ graphon, and
  $\rho_n \rightarrow 0$ and $n \rho_n \rightarrow
  \infty$ as $n \rightarrow \infty$. We claim that with $\Delta_n = \Delta$
  fixed, all the assumptions in Proposition~\ref{prop:graphon-optimality} are
  satisfied for $n$ large enough. Indeed, $\Delta \leq \log \log n$ for $n$
  large, and $\Delta / \sqrt{n \rho_n} \rightarrow 0$. Furthermore, since
  $\mn_{\Delta_n} = \mn_\Delta \leq n \Delta / n$,  we have $\mn_{\Delta_n} / \barmn
  \rightarrow 0$ as $n\rightarrow \infty$. Therefore,~\eqref{eq:dual-main-fixed-graphon} follows from
  Proposition~\ref{prop:graphon-optimality}. This completes the proof.
\end{proof}

\editfinish

So far, in Theorems~\ref{thm:dual-main-an} and \ref{thm:main-dual-fixed}, we have discussed the existence of a sequence of  compression and decompression
schemes $((f_n, g_n): n \geq 1)$ 
that
\black
almost surely achieve the asymptotic compression
limits $\bch(\mu)$ and $1 - \Ent(W)$ in the local weak convergence and the
sparse graph regimes respectively. In the following converse result, we argue
that these are indeed the smallest possible
thresholds that can be achieved almost surely. The proof of the following
Theorem~\ref{thm:lwc-graphon-convesee} is given in Section~\ref{sec:converse-proof}.

\begin{thm}
  \label{thm:lwc-graphon-convesee}
  Assume that $((f_n, g_n): n \geq 1)$ is a sequence of lossless
  compression/decompression maps (i.e.\ $g_n \circ f_n$ is identity).
  Then,
  \begin{enumerate}
  \item For any $\mu \in \mP_u(\mT_*)$ with $\deg(\mu) \in (0,\infty)$, there
    exists a 
    sequence of random graphs $\Gn$ defined on a joint probably space such that
    $U(\Gn) \Rightarrow \mu$ a.s.\ and we have 
    \begin{equation}
      \label{eq:converse-lwc-part-statement}
      \pr{\limsup_{n \rightarrow \infty} \frac{\nat(f_n(\Gn)) - \mn \log n}{n} \leq t} < 1 \qquad \forall t < \bch(\mu),
    \end{equation}
    where $\mn$ denotes the number of edges in $\Gn$.
  \item For any normalized $L^2$ graphon $W$ and any sequence of target
    densities $\rho_n$ such that
 $\rho_n
    \rightarrow 0$ and
    $n\rho_n \rightarrow \infty$,
    \black
    if $\Gn \sim \mG(n; \rho_n W)$ is the sequence of $W$-random
    graphs with target densities $\rho_n$, then we have
    \begin{equation}
      \label{eq:converse-graphon-part-statement}
      \pr{\limsup_{n \rightarrow \infty} \frac{\nat(f_n(\Gn)) - \barmn \log \frac{1}{\rho_n}}{\barmn} \leq t} < 1 \qquad \forall t < 1 - \Ent(W).
    \end{equation}
  \end{enumerate}
\end{thm}

\section{Coding Scheme}
\label{sec:coding-scheme}

\editstart

In this section, we provide the details of our compression scheme by
introducing the compression and decompression maps $\fnDn$ and $\gnDn$. Recall
from Section~\ref{sec:statement-and-results} that $\Delta_n$ is the splitting
parameter
which governs
\black
how we obtain $\Gn_{\Delta_n}$ and
$\Gn_*$ from the input graph  $\Gn$. More precisely, both $\Gn_{\Delta_n}$ and $\Gn_*$ are simple
graphs
on the vertex set $[n]$, and for each edge $(v,w)$ in $\Gn$, if $\deg_{\Gn}(v)
\leq \Delta_n$ and $\deg_{\Gn}(w) \leq \Delta_n$, we put an edge in
$\Gn_{\Delta_n}$ between the nodes $v$ and $w$. Otherwise, if either
$\deg_{\Gn}(v) > \Delta_n$ or $\deg_{\Gn}(w) > \Delta_n$, we place an edge in
$\Gn_*$ between the nodes $v$ and $w$. Let $R_n$ denote the set of vertices $v
\in [n]$ such that either $\deg_{\Gn}(v) > \Delta_n$ or $\deg_{\Gn}(w) >
\Delta_n$ for some $w \sim_{\Gn} v$. Note that for every edge $(v,w)$ in
$\Gn_*$, we have $v \in R_n$ and $w \in R_n$.

We first encode $\Gn_{\Delta_n}$ using the compression method from
\cite{delgosha2020universal} which we reviewed in
Section~\ref{sec:prelim-baby-compression}
\footnote{
Please note that
the splitting parameter $\Delta_n$ here should not be confused with the
truncation threshold $\log \log n$ in Section~\ref{sec:prelim-baby-compression}. The
parameter $\Delta_n$ determines how to split edges between $\Gn_{\Delta_n}$ and
$\Gn_*$, while the $\log \log n$ threshold in
Section~\ref{sec:prelim-baby-compression} determines which edges in
$\Gn_{\Delta_n}$ are separated out and placed in the set $Z_n$. However, as we
 discussed in Section~\ref{sec:statement-and-results}, we have two methods for
 choosing $\Delta_n$: we either set $\Delta_n = \min \{\log a_n, \log \log n\}$ as in
 Theorem~\ref{thm:dual-main-an}, or $\Delta_n = \Delta$ is fixed as in
 Theorem~\ref{thm:main-dual-fixed}. In either case, we have
$\Delta_n \leq \log
\log n$ for $n$ large enough, which means that all the degrees in
 $\Gn_{\Delta_n}$ are automatically bounded by $\log \log n$ are hence no edges
 will be truncated during the compression of $\Gn_{\Delta_n}$ using the method of
 Section~\ref{sec:prelim-baby-compression}.
 }.
Next, we discuss how to encode
$\Gn_*$. 
Overall, $\fn(\Gn)$ will be comprised of  
\black
the compressed form of
$\Gn_{\Delta_n}$ concatenated with the compressed form of $\Gn_*$. In order to
encode $\Gn_*$, we first encode the set $R_n$. For this, we first encode $|R_n|$
using at most  $1 + \log n$ nats, and then we encode the set $R_n$ using at most
$1 + \log
\binom{n}{|R_n|}$ nats by specifying $R_n$ among all the subsets of $[n]$ with
the same size. Recall that all the edges in $\Gn_*$ have both of their endpoints
in $R_n$. Thereby, if $|R_n| =0$,
$\Gn_*$ has no
edges and nothing remains to be done. Hence we assume that $|R_n| \geq  2$ from this point forward. Let
$\mn_*$ denote the number of edges in $\Gn_*$. Since $\mn_* \leq \binom{n}{2}$,
we may encode $\mn_*$ using at most $1 + 2 \log n$ nats.
 Define 
\begin{equation}
  \label{eq:alphan-def}
  \alpha_n := \exp \left( \left \lfloor  \log \frac{\mn_*}{n} \right \rfloor \right).
\end{equation}
Moreover, define the function $\phi:[0,\infty) \rightarrow [1,\infty)$ as
follows
\begin{equation}
  \label{eq:phi-function-def}
  \phi(x) :=
  \begin{cases}
    \frac{\sqrt{x}}{\log x} & x > e^2 \\
    1 & \text{otherwise}.
  \end{cases}
\end{equation}
Additionally, we define
\begin{equation}
  \label{eq:beta-n-def}
  \beta_n := \phi(\alpha_n).
\end{equation}

Note that the decoder knows $\mn_*$ at this point, and can
compute the value of $\beta_n$. Next, we run the least squares algorithm
from \cite{borgs2015consistent} which we discussed in
Section~\ref{sec:prelim-least-square-algo} on the input graph $\Gn$, with
parameter $\beta_n$ defined above.  Let $\hpin$
and $\hat{B}_n$ be the outputs of this algorithm, i.e.\ the optimizers in~\eqref{eq:least-sq-alg-better}.
Recall from~\eqref{eq:Bij-optimal-block-average} that for $1 \leq i ,j \leq
\beta_n$ such that $\hpin^{-1}(\{i\})$ and $\hpin^{-1}(\{j\})$ are not empty, we have
\begin{equation}
\label{eq:hatBn-entries-explicit}
  (\hat{B}_n)_{i,j} = \frac{1}{|\hpin^{-1}(\{i\})| |\hpin^{-1}(\{j\})|} \sum_{u \in \hpin^{-1}(\{i\}), v \in \hpin^{-1}(\{j\})} (A(\Gn))_{u,v}.
\end{equation}
Note that when $\hpin^{-1}(\{i\}) = \emptyset$ or $\hpin^{-1}(\{j\}) =
\emptyset$, the value of $(\hat{B}_n)_{i,j}$ does not affect the objective function
in~\eqref{eq:least-sq-alg-better}. Therefore, without loss of generality, we may
assume that $(\hat{B}_n)_{i,j} = 0$ for such $i,j$.
We emphasize that we run this algorithm  on the input graph $\Gn$, and not on
$\Gn_*$. However, we will use $\hpin$ and $\hat{B}_n$ to compress $\Gn_*$, as we
discuss below.

We first need to define some notation.
Let $\beta'_n$ be the number of $1 \leq i \leq \beta_n$ such that
  $\hpin^{-1}(\{i\}) \neq \emptyset$. Note that
  in~\eqref{eq:least-sq-alg-better}, we may reorder the vertex class labels
  governed by $\pi$ and modify $B$ accordingly without changing the objective.  
  Therefore, without loss of generality, we may assume
  that $\hpin^{-1}(\{i\}) \neq \emptyset$ for $1 \leq i \leq \beta'_n$ and
  $\hpin^{-1}(\{i\}) = \emptyset $ for $i > \beta'_n$.
Let $\tpin$ be the restriction of $\hpin$ on $R_n$. More precisely, with $R_n =
\{r_1, \dots, r_{|R_n|}\}$ such that $r_1 < r_2 < \dots < r_{|R_n|}$, we define $\tpin
  : [|R_n|] \rightarrow [ \beta_n ]$ such that $\tpin(i) =
  \hpin(r_i)$.
For $1 \leq i \leq \beta_n$, let $n_i := |\hpin^{-1}(\{i\})|$. Moreover,
  let $n^*_i := |\tpin^{-1}(\{i\})| = |\hpin^{-1}(\{i\}) \cap R_n|$.
  Let $\beta^*_n$ be the number of $1 \leq i \leq \beta_n$ such that $n^*_i
  \neq 0$. Note that $\beta^*_n \leq \beta'_n$.
  By a reordering argument similar to the above in~\eqref{eq:least-sq-alg-better}, 
 without
  loss of generality, we may assume that $n^*_i > 0$ for $1 \leq i \leq
  \beta^*_n$ and $n^*_i = 0$ for $i > \beta^*_n$.
  For $1 \leq i ,  j \leq \beta^*_n$, let
  \begin{equation*}
    \ms_{i,j} :=
    \begin{cases}
      \sum_{u < v \in \tpin^{-1}(\{i\})} (A(\Gn_*))_{u,v} & \mbox{if } i = j, \\
      \sum_{u \in \tpin^{-1}(\{i\}), v \in \tpin^{-1}(\{j\})} (A(\Gn_*))_{u,v} & \mbox{if } i \neq j.
    \end{cases}
  \end{equation*}
  Also, we define $\ms_{i,j}$ to be zero if $i > \beta^*_n$ or $j > \beta^*_n$.
  In other words, $\ms_{i,j}$ is the number of edges in the $i,j$ block formed
  by $\tpin$ in the adjacency matrix of $\Gn_*$.
  Likewise, for $1 \leq i , j \leq \beta'_n$, we define
  \begin{equation}
    \label{eq:mij-def}
    m_{i,j} :=
    \begin{cases}
      \sum_{u < v \in \hpin^{-1}(\{i\})} (A(\Gn))_{u,v} & \mbox{if } i = j, \\
      \sum_{u \in \hpin^{-1}(\{i\}), v \in \hpin^{-1}(\{j\})} (A(\Gn))_{u,v} & \mbox{if } i \neq j,
    \end{cases}
  \end{equation}
  and we define $m_{i,j}$ to be zero for $i > \beta'_n$ or $j > \beta'_n$.
    In other words, $m_{i,j}$ is the number of edges in the $i,j$ block formed
    by $\hpin$ in the adjacency matrix of  $\Gn$.
Note that $\ms_{i,j} \leq m_{i,j}$
for $1 \leq i, j \leq \beta_n$.
\black

Having defined the above notation, we continue with encoding $\tpin$. Since
$\tpin(i) \leq \beta_n$ 
for $1 \leq i \leq |R_n|$, we may encode $\tpin$ using
at most $|R_n| (1+\log \beta_n)$ nats. 
Next, 
\black
we encode $A(\Gn_*)$. We do this by
separately encoding
each  block in $A(\Gn_*)$ formed by $\hpin$. More precisely,
for $1 \leq i \leq \beta^*_n$, we encode the block
  $\tpin^{-1}(\{i\}) \times \tpin^{-1}(\{i\})$ of $A(\Gn_*)$ as
  follows. We first encode $\ms_{i,i}$ using at most $1 + 2 \log |R_n|$ nats. This is possible
  since $\ms_{i,i} \leq \mn_* \leq |R_n|^2$. Then, we encode the positions
  of the $\ms_{i,i}$ ones in the upper triangular part of the
  $\tpin^{-1}(\{i\}) \times \tpin^{-1}(\{i\})$ block 
  \black
  by at most $1 + \log
  \binom{\binom{n^*_i}{2}}{\ms_{i,i}}$ nats. 
Similarly, for $1 \leq i < j \leq \beta^*_n$, we first encode $\ms_{i,j}$ using
at most $1 + 2 \log |R_n| $ nats, and then we  encode the positions of the $\ms_{i,j}$
  ones in the block $\tpin^{-1}(\{i\}) \times \tpin^{-1}(\{j\})$ of
  $A(\Gn_*)$ using at most  $1 + \log \binom{n^*_i n^*_j}{\ms_{i,j}}$ nats.

At the decoder, we first reconstruct $\Gn_{\Delta_n}$. Next, we decode for the
set $R_n$ and $\mn_*$. We then find $\beta_n$ from~\eqref{eq:beta-n-def}. Then
we decode for $\tpin$, 
and decode each of the blocks of $A(\Gn_*)$ separately.
We then reconstruct $\Gn_*$ by putting the blocks of $A(\Gn_*)$ together.
Finally, we put $\Gn_{\Delta_n}$ and $\Gn_*$ together to reconstruct $\Gn$.

Let $\elln_{\Delta_n}$ and $\elln_{*}$ denote the number of nats we use to
encode $\Gn_{\Delta_n}$ and $\Gn_*$, respectively,
so that $\nat(f_n(\Gn)) =
\elln_{\Delta_n} + \elln_*$. Going over the description of the compression
procedure discussed above, we obtain the following result.
The following Lemma~\ref{lem:elln-*-nat-count} will be used in the proof of
Propositions~\ref{prop:lwc-part} and \ref{prop:graphon-optimality}.

\begin{lem}
  \label{lem:elln-*-nat-count}
  If $R_n \neq \emptyset$, we have $\elln_* \leq \elln_{*,1} + \elln_{*,2}$  where
    \begin{align*}
      \elln_{*,1} &:= 3 + |R_n| + 3\log n + \log \binom{n}{|R_n|}  + |R_n| \log \beta_n, \\
      \elln_{*,2} &:= 2 {\beta_n^*}^2 + \sum_{i=1}^{\beta^*_n} \left( 2 \log |R_n| + \log \binom{\binom{n^*_i}{2}}{\ms_{i,i}} \right) + \sum_{1 \leq i < j \leq \beta^*_n} \left( 2 \log |R_n| + \log \binom{n^*_i n^*_j}{\ms_{i,j}} \right).
    \end{align*}
    Otherwise, if $R_n = \emptyset$, we have $\elln_* \leq 1 + \log n$.
\end{lem}

\begin{proof}
  Following the encoding procedure, we begin with encoding the set $R_n$ using
  at most 
  $2+ \log n + \log \binom{n}{|R_n|}$
  nats.
  \black
  If $R_n = \emptyset$, we encode $R_n$ using
  at most $1 + \log n$ nats, and the encoding procedure stops at this point. Therefore, if
  $R_n = \emptyset$, we have $\elln_* \leq 1 +  \log n$. Now, assume that $R_n \neq
  \emptyset$. In this case, after encoding $R_n$, we encode $\mn_*$ using at
  most $ 1 + 2
  \log n$ nats. Then, we encode $\tpin$ using $|R_n| (1+\log \beta_n)$ nats.
  Moreover, we encode  each diagonal block $1 \leq i \leq \beta^*_n$ in $A(\Gn_*)$
  using at most $2 + 2 \log |R_n| + \log \binom{\binom{n^*_i}{2}}{m^*_{i,j}}$ nats. Also, we
  encode 
  each non--diagonal block $1 \leq i < j \leq \beta^*_n$ using at most $2 + 2 \log |R_n| +
  \log \binom{n^*_i n^*_j}{m^*_{i,j}}$ nats. This completes the proof.
\end{proof}


\editfinish

\editstart

\section{Proof of Proposition~\ref{prop:lwc-part}: Local weak Convergence Analysis}
\label{sec:lwc-analysis}



In this section, we prove Proposition~\ref{prop:lwc-part}.
  Since
  $U(\Gn_{\Delta_n}) \Rightarrow \nu$ a.s.,
  using Theorem~\ref{thm:prelim-baby-compression} in Section~\ref{sec:prelim-baby-compression}, if $\elln_{\Delta_n}$ denotes the number
  of nats we use to encode $\Gn_{\Delta_n}$, we have
  \begin{equation}
    \label{eq:lwc-Gn-Deltan-limsup-BC}
    \limsup_{n \rightarrow \infty} \frac{\elln_{\Delta_n} - \mn_{\Delta_n} \log n}{n} \leq \bch(\nu) \qquad \qquad \text{a.s.,}
  \end{equation}
  where $\mn_{\Delta_n}$ denotes the number of edges in $\Gn_{\Delta_n}$. Now, let
  $\elln_*$ denote the number of nats we use to encode $\Gn_*$, so that
  $\nat(f_n(\Gn)) = \elln_{\Delta_n} + \elln_*$. We claim that
  \begin{equation}
    \label{eq:lwc-claim-elln-star}
    \limsup_{n \rightarrow \infty} \frac{\elln_* - \mn_* \log n}{n} \leq  (25 + e/2) \eta + H_b(\eta) \qquad \qquad \text{a.s.},
  \end{equation}
  where $\mn_*$ denotes the number of edges in $\Gn_*$. Note that this together
  with~\eqref{eq:lwc-Gn-Deltan-limsup-BC} finishes the proof. 
  From Lemma~\ref{lem:elln-*-nat-count}, if $R_n = \emptyset$, we have $\elln_*
  \leq 1 + \log n $ and~\eqref{eq:lwc-claim-elln-star} holds. Therefore, we assume that
  $R_n \neq \emptyset$ for the rest of the proof. 
Note that when $R_n \neq \emptyset$, there must be at least two vertices in
$R_n$ and hence $|R_n| \geq 2$.
In this case, again from
  Lemma~\ref{lem:elln-*-nat-count}, we have  $\elln_* \leq  \elln_{*,1} +
  \elln_{*,2}$ where
    \begin{align*}
      \elln_{*,1} &:= 3 + |R_n| + 3\log n + \log \binom{n}{|R_n|}  + |R_n| \log \beta_n, \\
      \elln_{*,2} &:= 2 {\beta_n^*}^2 + \sum_{i=1}^{\beta^*_n} \left( 2 \log |R_n| + \log \binom{\binom{n^*_i}{2}}{\ms_{i,i}} \right) + \sum_{1 \leq i < j \leq \beta^*_n} \left( 2 \log |R_n| + \log \binom{n^*_i n^*_j}{\ms_{i,j}} \right).
    \end{align*}
Using the bound $\binom{r}{s} \leq (re/s)^s$, 
    we can write
    \begin{equation}
      \label{eq:elln-1-simplification}
      \begin{aligned}
        \elln_{*,2} &\leq 2{\beta_n^*}^2 + \sum_{i=1}^{\beta^*_n} \left( 2 \log |R_n| + \ms_{i,i} \log \frac{{n^*_i}^2 e}{2 \ms_{i,i}} \right) + \sum_{1 \leq i < j \leq \beta^*_n} \left( 2 \log |R_n| + \ms_{i,j} \log \frac{n^*_i n^*_j e}{\ms_{i,j}} \right) \\
        &\leq 2 {\beta^*_n}^2 (1+\log |R_n|) + \sum_{i=1}^{\beta^*_n} \ms_{i,i} \log n^*_i + n^*_i \left( \frac{\ms_{i,i}}{n^*_i} \log \frac{n^*_i e}{2 \ms_{i,i}} \right) \\
        &\qquad + \sum_{1 \leq i < j \leq \beta^*_n} \ms_{i,j} \log \sqrt{n^*_i n^*_j} + \sqrt{n^*_i n^*_j} \left( \frac{\ms_{i,j}}{\sqrt{n^*_i n^*_j}} \log \frac{\sqrt{n^*_i n^*_j}e}{\ms_{i,j}} \right)\\
        &\stackrel{(a)}{\leq} 6{\beta^*_n}^2 \log |R_n| + \sum_{i=1}^{\beta^*_n} \ms_{i,i} \log n + n^*_i \left( \frac{\ms_{i,i}}{n^*_i} - \frac{\ms_{i,i}}{n^*_i} \log \frac{2\ms_{i,i}}{n^*_i}\right) \\
        &\qquad + \sum_{1 \leq i < j \leq \beta^*_n} \ms_{i,j} \log n + \sqrt{n^*_i n^*_j} \left( \frac{\ms_{i,j}}{\sqrt{n^*_i n^*_j}} - \frac{\ms_{i,j}}{\sqrt{n^*_i n^*_j}} \log \frac{\ms_{i,j}}{\sqrt{n^*_i n^*_j}} \right) \\
        &\stackrel{(b)}{=} 6 {\beta^*_n}^2 \log |R_n| + \mn_* \log n + \sum_{i=1}^{\beta^*_n} n^*_i s\left( \frac{2\ms_{i,i}}{n^*_i} \right) + \sum_{1 \leq i < j \leq \beta^*_n} 2 \sqrt{n^*_i n^*_j} s\left( \frac{\ms_{i,j}}{\sqrt{n^*_i n^*_j}} \right) \\
        &\stackrel{(c)}{\leq} 6 {\beta^*_n}^2 \log |R_n| + \mn_* \log n + \left( \sum_{i=1}^{\beta^*_n} \sqrt{n^*_i} \right)^2 s\left( \frac{2\mn_*}{\left( \sum_{i=1}^{\beta^*_n} \sqrt{n^*_i} \right)^2} \right),
      \end{aligned}
    \end{equation}
    where in $(a)$, we have used the fact that since $|R_n| \geq 2$, $1 + \log
    |R_n| \leq 3 \log |R_n|$, 
    in $(b)$, we use $s(x)=\frac{x}{2} - \frac{x}{2} \log x$ for $x > 0$ and $s(0) = 0$,  and we have simplified the expression using $\sum_{i=1}^{\beta_n^*}
    \ms_{i,i} + \sum_{1 \leq i < j \leq \beta_n^*} \ms_{i,j} = \mn_*$.
    Furthermore, $(c)$ uses
    the concavity of the function $s(.)$. 
    Now, we consider two cases to bound this expression.

\underline{Case 1:} Assume that  $\mn_* / n < e^3$, in which case
using~\eqref{eq:alphan-def}, we have $\alpha_n \leq e^2$, and
from~\eqref{eq:beta-n-def}, we have  $\beta_n = 1$.
Since $|R_n| > 0$, this implies that $\beta^*_n
= 1$ 
\black
and $n^*_1 = |R_n|$. Furthermore,
using~\eqref{eq:elln-1-simplification}, we have
\begin{align*}
  \elln_{*,2} &\leq 6 \log |R_n| + \mn_* \log n + |R_n| s\left( \frac{2 \mn_*}{|R_n|} \right) \\
              &\stackrel{(a)}{\leq} 6 \log n + \mn_* \log n + \frac{1}{2} |R_n| \\
\end{align*}
where in $(a)$  we have used the facts that $|R_n| \leq n $ and $s(x) \leq 1/2$ for $x \geq 0$.
Using this together with $\elln_* \leq \elln_{*,1} + \elln_{*,2}$, recalling
that $\elln_{*,1} = 3 + |R_n| + 3\log n + \log \binom{n}{|R_n|}  + |R_n| \log \beta_n$,
and simplifying using
$\beta_n  =1$, we get
\begin{equation}
  \label{eq:eln-star-bound-case-1}
  \elln_* \leq   \mn_* \log n + \log \binom{n}{|R_n|} + \frac{3}{2} |R_n|  + 9 \log n +3.
\end{equation}


\underline{Case 2:} Assume that $\mn_* / n \geq e^3$.
Note that in this case, $\alpha_n
\geq e^3 >e^2$ and $\beta_n = \sqrt{\alpha_n} / \log \alpha_n$. Since $\phi(x)$ is
increasing for $x > e^2$, and $\alpha_n \leq \mn_* / n$,
we have
\begin{equation}
  \label{eq:lwc-proof-case-2-betan-bound}
\beta_n \leq \frac{\sqrt{\frac{\mn_*}{n}} }{\log(\frac{\mn_*}{n})}.
\end{equation}
Thereby, 
using the inequality
$(\sum_{i=1}^{\beta^*_n} \sqrt{n_i^*})^2 \leq \beta^*_n |R_n| \leq n \beta_n $, 
we get
\begin{equation}
  \label{eq:2mn-sum-bigger-than-e}
  \begin{aligned}
    \frac{2\mn_*}{(\sum_{i=1}^{\beta^*_n} \sqrt{n^*_i})^2} &\geq \frac{2\mn_*}{\beta^*_n |R_n|} \\
    &\geq \frac{2\mn_*}{n \beta_n} \\
    &\geq 2 \sqrt{\frac{\mn_*}{n}} \log \frac{\mn_*}{n} \\
    &\geq e,
  \end{aligned}
\end{equation}
where the last inequality uses $\mn_* / n > e^2$.  Note that for $x \geq e$, we have $s(x) \leq 0$.
Moreover, $s(x)$ is decreasing for $x \geq e$. On the other hand, we have 
$(\sum_{i=1}^{\beta^*_n} \sqrt{n^*_i})^2 \geq
\sum_{i=1}^{\beta^*_n} n^*_i = |R_n|$. 
Simple calculus shows 
that   $s(x) = \frac{x}{2} - \frac{x}{2} \log x \leq (e-x)/2$
for $x \geq 0$.
This discussion together with 
inequality~\eqref{eq:2mn-sum-bigger-than-e}  implies that 
\begin{equation}
  \label{eq:s-upper-bound}
  \begin{aligned}
    \left( \sum_{i=1}^{\beta^*_n} \sqrt{n^*_i} \right)^2 s\left( \frac{2\mn_*}{\left( \sum_{i=1}^{\beta^*_n} \sqrt{n^*_i} \right)^2} \right) &\leq |R_n| s\left( \frac{2\mn_*}{\left( \sum_{i=1}^{\beta^*_n} \sqrt{n^*_i} \right)^2} \right) \\
    &\stackrel{(*)}{\leq} |R_n| \frac{e - 2 \sqrt{\frac{\mn_*}{n}} \log \frac{\mn_*}{n}}{2} \\
    &\leq \frac{e}{2} |R_n| - |R_n| \log \frac{\mn_*}{n},
  \end{aligned}
\end{equation}
where $(*)$ uses the third line in~\eqref{eq:2mn-sum-bigger-than-e} and 
 the last inequality uses the fact $\mn_* / n \geq e^3 > 1$.
Substituting this into~\eqref{eq:elln-1-simplification}, we get
\begin{equation}
  \label{eq:case-2-elln-1-upperbound}
  \elln_{*,2} \leq 6 {\beta_n^*}^2 \log |R_n| + \mn_* \log n + \frac{e}{2} |R_n| - |R_n| \log \frac{\mn_*}{n}.
\end{equation}
On the other hand, using~\eqref{eq:lwc-proof-case-2-betan-bound}, we have
\begin{equation}
  \label{eq:2betanstarlogR-n-bound-1}
    6 {\beta_n^*}^2 \log |R_n| \leq 6 \beta_n^2 \log |R_n| \leq 6\frac{\mn_*}{n} \frac{\log |R_n|}{\log^2 \frac{\mn_*}{n}}.
\end{equation}
We find an upper bound for this term in two cases. First, assume that $\mn_*
\leq |R_n|^{3/2}$. Note that $|R_n| \leq n$. Furthermore,  since $\mn_* / n \geq e^3 > e$,
we have 
$\log (\mn_*/ n) > 1$. Using these, we get 
\begin{equation*}
  6\frac{\mn_*}{n} \frac{\log |R_n|}{\log^2 \frac{\mn_*}{n}} \leq 6 \frac{\mn_*}{|R_n|} \log |R_n| \leq 6\sqrt{|R_n|} \log |R_n| \leq 12|R_n|,
\end{equation*}
where the last inequality uses the bound $\log |R_n| = 2 \log \sqrt{|R_n|} \leq 2\sqrt{|R_n|}$.
Now, assume that  $\mn_*  > |R_n|^{3/2}$. Since the function $x \mapsto x /
\log^2 x$ is increasing for $x > e^2$ and $\mn_* / n  > e^2$, we may write 
\begin{equation*}
  6\frac{\mn_*}{n} \frac{\log |R_n|}{\log^2 \frac{\mn_*}{n}} \leq 6\frac{\mn_*}{|R_n|} \frac{\log |R_n|}{\log^2 \frac{\mn_*}{|R_n|}} \leq 6\frac{\mn_*}{|R_n|} \frac{\log |R_n|}{\log^2 \sqrt{|R_n|}} \stackrel{(*)}{\leq} 3|R_n| \frac{\log |R_n|}{\log^2 \sqrt{|R_n|}} = 12 \frac{|R_n|}{\log |R_n|} \leq 24 |R_n|,
\end{equation*}
where in $(*)$, we have used the fact that $\mn_* \leq \binom{|R_n|}{2} \leq
|R_n|^2/2$. Also, in the last step, we have used $\log |R_n| \geq \log 2 \geq 1/2$.
Combining the two cases and substituting
into~\eqref{eq:2betanstarlogR-n-bound-1},
we realize that
\begin{equation*}
  6 {\beta^*_n}^2 \log |R_n| \leq 24 |R_n|. 
 \end{equation*}
  Substituting into~\eqref{eq:case-2-elln-1-upperbound},
  we get
\begin{equation*}
  \elln_{*,2} \leq \mn_* \log n +  (24+e/2) |R_n| - |R_n| \log \frac{\mn_*}{n}.
\end{equation*}
Using this together with $\elln_* \leq \elln_{*,1} + \elln_{*,2}$ and recalling
that $\elln_{*,1} = 3 + |R_n| + 3\log n + \log \binom{n}{|R_n|}  + |R_n| \log \beta_n$,
we get
\begin{equation*}
  \elln_* \leq \mn_* \log n  + \log \binom{n}{|R_n|} + (25+e/2) |R_n| + 3 \log n + |R_n| \log \beta_n - |R_n| \log \frac{\mn_*}{n} + 3.
\end{equation*}
Note that since $\mn_* / n \geq e^3> e$,
using~\eqref{eq:lwc-proof-case-2-betan-bound}, we have
\begin{equation*}
  |R_n| \log \beta_n - |R_n| \log \frac{\mn_*}{n} \leq \frac{1}{2} |R_n| \log \frac{\mn_*}{n} - |R_n| \log \log \frac{\mn_*}{n} - |R_n| \log \frac{\mn_*}{n} < 0.
\end{equation*}
Thereby, we have
\begin{equation}
\label{eq:eln-star-bound-case-2}
  \eln_* \leq \mn_* \log n  + \log \binom{n}{|R_n|} + (25+e/2) |R_n| + 3 \log n+3.
\end{equation}

Combining \eqref{eq:eln-star-bound-case-1} and \eqref{eq:eln-star-bound-case-2},
we get the following bound which holds for both cases:
\begin{equation}
  \label{eq:eln-star-lwc-bound-both-cases}
  \eln_* \leq \mn_* \log n + \log \binom{n}{|R_n|} + (25 + e/2) |R_n| + 9 \log n +3.
\end{equation}
Note that $|R_n| / n \rightarrow \eta$ a.s., which implies that $\frac{1}{n}\log
\binom{n}{|R_n|} \rightarrow H_b(\eta)$ a.s., where $H_b(\eta)$ denotes the binary
entropy of $\eta$. Using this in~\eqref{eq:eln-star-lwc-bound-both-cases}, we
realize that  
\begin{equation*}
  \limsup_{n \rightarrow \infty} \frac{\eln_* - \mn_* \log n}{n} \leq (25 + e/2) \eta + H_b(\eta) \qquad \qquad \text{a.s.},
\end{equation*}
which is precisely~\eqref{eq:lwc-claim-elln-star}. This together
with~\eqref{eq:lwc-Gn-Deltan-limsup-BC} completes the proof.

\section{Proof of Proposition~\ref{prop:graphon-optimality}: Graphon Analysis}
\label{sec:graphon-analysis}

In this section, we prove Proposition~\ref{prop:graphon-optimality}.
Before giving the proof, we introduce some  notation. 
Recall from Section~\ref{sec:coding-scheme} that $\hpin$ and $\hat{B}_n$ are the
optimizers in~\eqref{eq:least-sq-alg-better} associated to $A(\Gn)$  with
parameter $\beta_n$, and $n_i =
|\hpin^{-1}(\{i\})|$ for $1 \leq i \leq \lfloor \beta_n \rfloor$.
Let $\hWn$ be the
block graphon $(\vec{p}, \hat{B}_n)$ where $\vec{p} =(p_1, \dots, p_{\lfloor \beta_n \rfloor})$ with $p_i = n_i
/ n$ for $1 \leq i \leq \lfloor  \beta_n \rfloor$.  More precisely,
using~\eqref{eq:hatBn-entries-explicit} and \eqref{eq:mij-def} in
Section~\ref{sec:coding-scheme}, $\hWn$ is defined
on the finite probability space $\{1, \dots, \lfloor \beta_n \rfloor\}$ equipped
with probabilities $(\frac{n_1}{n}, \dots, \frac{n_{\lfloor \beta_n
    \rfloor}}{n})$ such that $\hWn(i,j) := \lambda_{i,j}$ where
\begin{equation}
  \label{eq:pij-def}
  \lambda_{i,j} :=
  \begin{cases}
    \frac{2m_{i,j}}{n_i^2} & \text{ if } i = j \text{ and } i, j \leq  \beta'_n,  \\
    \frac{m_{i,j}}{n_i n_j} & \text{ if }i \neq j \text{ and } i, j \leq \beta'_n, \\
    0 & \text{otherwise,}
  \end{cases}
\end{equation}
where $\beta'_n$  and $m_{i,j}$ were defined in Section~\ref{sec:coding-scheme}.
Moreover, 
we define the graphon $\hWns$ 
on the same
probability space  
\black
$\{1, \dots, \lfloor \beta_n \rfloor\}$ equipped
with probabilities $(\frac{n_1}{n}, \dots, \frac{n_{\lfloor \beta_n
    \rfloor}}{n})$  such that $\hWns(i,j) = \lambda^*_{i,j}$ where 
\begin{equation}
  \label{eq:pstarij-def}
  \lambda^*_{i,j} :=
  \begin{cases}
    \frac{2m^*_{i,j}}{n_i^2} & \text{ if } i = j \text{ and } i, j \leq  \beta^*_n,  \\
    \frac{m^*_{i,j}}{n_i n_j} & \text{ if } i \neq j \text{ and } i, j \leq \beta^*_n, \\
    0 & \text{otherwise,}
  \end{cases}
\end{equation}
where $\beta^*_n$ and $m^*_{i,j}$ were defined in Section~\ref{sec:coding-scheme}.


The following Proposition~\ref{prop:graphon-mn-asymptotics} discusses some useful facts regarding the asymptotic
behavior of the number of edges in
a sequence of sparse
$W$-random graphs, and will be useful in the proof of
Proposition~\ref{prop:graphon-optimality}. The proof of
Proposition~\ref{prop:graphon-mn-asymptotics} is given in Appendix~\ref{app:proof-prop-graphon-mn-asymptotics}.

\begin{prop}
  \label{prop:graphon-mn-asymptotics}
  Assume that  $W$ is a normalized $L^2$ graphon and $\Gn \sim \mG(n; \rho_n W)$ is a sequence
  of $W$--random graphs with target density $\rho_n$ such that 
  $\rho_n \rightarrow 0$
  and $n \rho_n
  \rightarrow \infty$.
  \black
  Then, with $\mn$ being the
  number of edges in $\Gn$ and $\barmn :=
  \binom{n}{2} \rho_n$, the following hold:
  \begin{enumerate}
  \item \label{item:prop-mn-asymp-}
    \begin{equation*}
      \lim_{n \rightarrow \infty} \frac{\mn}{\barmn} = 1 \qquad \text{a.s.};
    \end{equation*}
  \item
    \begin{equation*}
      \lim_{n \rightarrow \infty} \frac{\mn - \barmn}{\barmn} \log \frac{1}{\rho_n} = 0 \qquad \text{a.s.};
    \end{equation*}
  \item \label{item:prop-mn-asymp-log-n2-1-as}
    \begin{equation*}
      \lim_{n \rightarrow \infty} \frac{\log \binom{\binom{n}{2}}{\mn} - \barmn \log \frac{1}{\rho_n}}{\barmn} = 1 \qquad \text{a.s.};
    \end{equation*}
  \item \label{item:prop-mn-asymp-log-n2-1-ev}
    \begin{equation*}
      \limsup_{n \rightarrow \infty} \ev{\frac{\log \binom{\binom{n}{2}}{\mn} - \barmn \log \frac{1}{\rho_n}}{\barmn}} \leq 1.
    \end{equation*}
  \end{enumerate}
\end{prop}

The following lemmas will be useful in the proof of Proposition~\ref{prop:graphon-optimality}. The proofs of
these lemmas are given in Appendix~\ref{app:graphon-analysis-lem-proofs}.

\begin{restatable}{lem}{lemphiproperties}
  \label{lem:phi-properties}
  For the function $\phi(.)$ defined in~\eqref{eq:phi-function-def}, we have
  \begin{subequations}
    \begin{align}
      &\lim_{x \rightarrow \infty} \phi(x) = \infty; \label{eq:phi-prop-lem-1}\\
      &\lim_{x \rightarrow \infty} \frac{\phi^2(x) \log \phi(x)}{x} = 0. \label{eq:phi-prop-lem-2}
    \end{align}
  \end{subequations}
Moreover, the function $\phi(.)$ is nondecreasing on $[0,\infty)$ and strictly
increasing on $(e^2, \infty)$.
\end{restatable}

\begin{restatable}{lem}{lemhWnconvergestoW}
  \label{lem:hWn-converges-to-W}
Under the assumptions of 
Proposition~\ref{prop:graphon-optimality}
\black
we have
\begin{equation*}
  \lim_{n \rightarrow \infty} \frac{\mn}{\barmn} = 1 \qquad \text{a.s.}.
\end{equation*}
and
\begin{equation*}
  \lim_{n \rightarrow \infty} \frac{\mn_*}{\barmn} = 1 \qquad \text{a.s.}.
\end{equation*}
Furthermore, we have 
  \begin{equation}
    \delta_2\left( \frac{1}{\rho_n} \hWn, W \right) \rightarrow 0 \qquad \text{a.s.}.
  \end{equation}
Moreover, with probability one, we have $\beta_n \rightarrow \infty$ and
$\beta_n^2 \log \beta_n = o(n \rho_n)$.
\end{restatable}

\begin{restatable}{lem}{lemdeltahWnshWngoestozero}
  \label{lem:delta-hWns-hWn-goes-to-zero}
With the assumptions of Proposition~\ref{prop:graphon-optimality}, we have
  \begin{equation*}
    \delta_2\left( \frac{1}{\rho_n} \hWns, \frac{1}{\rho_n} \hWn \right) \rightarrow 0 \qquad \text{a.s.}.
  \end{equation*}
\end{restatable}

\editfinish

\editstart
\begin{proof}[Proof of Proposition~\ref{prop:graphon-optimality}]
Using Lemma~\ref{lem:hWn-converges-to-W}, we have 
  \begin{equation}
    \label{eq:graphon-analysis-mn-star-barmn-1}
    \lim_{n \rightarrow \infty} \frac{\mn_*}{\barmn} = 1 \qquad \text{a.s.}.
  \end{equation}
Recall from Section~\ref{sec:coding-scheme} that  we denote by
  $\elln_{\Delta_n}$ and $\elln_*$ the number of nats used to encode
  $\Gn_{\Delta_n}$ and $\Gn_*$ respectively, so that $\nat(f_n(\Gn)) =
  \elln_{\Delta_n} + \elln_*$. Recall that we encode $\Gn_{\Delta_n}$ using the
  compression method discussed in~\cite{delgosha2019universal}, which we
  reviewed in Section~\ref{sec:prelim-baby-compression}.
Therefore, using Lemma~\ref{lem:prelim-baby-bounded-graph-codeword-bound}, we
have
  \begin{equation}
    \label{eq:graphon-optimality-elln-Delta-bound}
    \begin{aligned}
      \elln_{\Delta_n} &\leq \log \binom{\binom{n}{2}}{\mn_{\Delta_n}} +o(n) \\
      &\leq  \mn_{\Delta_n} \log \frac{\binom{n}{2} e}{\mn_{\Delta_n}} +o(n)\\
      &=  \mn_{\Delta_n} + \mn_{\Delta_n} \log \frac{\binom{n}{2}\rho_n}{\mn_{\Delta_n}} + \mn_{\Delta_n} \log \frac{1}{\rho_n}  +o(n) \\
      &= \mn_{\Delta_n} \log \frac{1}{\rho_n} + \mn_{\Delta_n} - \mn_{\Delta_n} \log \frac{\mn_{\Delta_n}}{\barmn}  +  o(n),
    \end{aligned}
  \end{equation}
  where we have used ${r \choose s} \le (\frac{ re}{s})^s$ to get the second inequality
  and have used
  $\barmn = n(n-1)
\rho_n / 2$ in the last step.
\black
Therefore, 
\begin{equation*}
  \limsup_{n \rightarrow \infty} \frac{\elln_{\Delta_n} - \mn_{\Delta_n} \log \frac{1}{\rho_n} }{\barmn} \leq \limsup_{n \rightarrow \infty} \frac{\mn_{\Delta_n}}{\barmn} - \frac{\mn_{\Delta_n}}{\barmn} \log \frac{\mn_{\Delta_n}}{\barmn}. 
\end{equation*}
By assumption, we have $\mn_{\Delta_n} / \barmn \rightarrow 0$ a.s..
Hence
\begin{equation}
  \label{eq:graphon-analysis-elln-Delta-n-zero}
  \limsup_{n \rightarrow \infty} \frac{\elln_{\Delta_n} - \mn_{\Delta_n} \log \frac{1}{\rho_n} }{\barmn} \leq 0 \qquad \text{a.s.}.
\end{equation}

Note that since $n \rho_n \rightarrow \infty$ as $n \rightarrow \infty$ we have
$\barmn \rightarrow \infty$ as $n \rightarrow \infty$, and
from~\eqref{eq:graphon-analysis-mn-star-barmn-1}
we have
$\mn_* \rightarrow \infty$ a.s.
as $n \rightarrow \infty$.
\black
Thereby, with probability one, $R_n \neq
\emptyset$ for $n$ large enough, and from Lemma~\ref{lem:elln-*-nat-count}  we have 
  $\elln_* \leq \elln_{*,1} + \elln_{*,2}$ where
      \begin{align*}
      \elln_{*,1} &:= 3 + |R_n| + 3\log n + \log \binom{n}{|R_n|}  + |R_n| \log \beta_n, \\
      \elln_{*,2} &:= 2 {\beta_n^*}^2 + \sum_{i=1}^{\beta^*_n} \left( 2 \log |R_n| + \log \binom{\binom{n^*_i}{2}}{\ms_{i,i}} \right) + \sum_{1 \leq i < j \leq \beta^*_n} \left( 2 \log |R_n| + \log \binom{n^*_i n^*_j}{\ms_{i,j}} \right).
    \end{align*}
    We claim that
    \begin{equation}
      \label{eq:graphon-analysis-elln-star-1-barmn-zero}
      \limsup_{n \rightarrow \infty} \frac{\elln_{*,1}}{\barmn} = 0 \qquad \text{a.s..}
    \end{equation}
    In order to show this, we consider two cases. If $\alpha_n \leq e^2$,
    then
    \black
    from~\eqref{eq:beta-n-def} we have $\beta_n = 1$. Therefore, using
    $\binom{n}{|R_n|} \leq 2^n$ and $|R_n| \leq n$, we have 
    \begin{equation*}
      \elln_{*,1} \leq 3 + n + 3 \log n + \log 2^n = 3 + n + 3 \log n + n \log 2.
    \end{equation*}
    On the other hand, if $\alpha_n > e^2$, recalling the definition of
    $\alpha_n$ in~\eqref{eq:alphan-def}, since $\alpha_n \leq \mn_* / n$, we
    have
    \begin{equation*}
      \log \beta_n = \log \phi(\alpha_n) = \log \frac{\sqrt{\alpha_n}}{\log \alpha_n} \leq \log \sqrt{\alpha_n} \leq \frac{1}{2} \log \frac{\mn_*}{n}.
    \end{equation*}
    Thereby, in this case, we have
    \begin{align*}
      \elln_{*,1} &\leq 3 + n + 3 \log n + \log 2^n + \frac{|R_n|}{2} \log \frac{\mn_*}{n} \\
                  &\leq 3 + n + 3 \log n + n \log 2 + \frac{n}{2} \log \frac{\mn}{n}.
    \end{align*}
    Combining the two cases, we get the following upper bound for $\elln_{*,1}$,
    which holds in both cases,
    \begin{equation}
      \label{eq:graphon-elln-star-1-bound}
      \elln_{*,1} \leq 3 + n + 3 \log n + n \log 2 + \frac{n}{2} \log \left( \frac{\mn}{n} \vee 1 \right).
    \end{equation}
From~Lemma~\ref{lem:hWn-converges-to-W} we have $\mn / \barmn
\rightarrow 1$ a.s.
as $n \rightarrow \infty$.
\black
Therefore, with probability one, for $n$ large enough, we
have
\begin{equation*}
  1 \vee \frac{\mn}{n} \leq 1 \vee \frac{2 \barmn}{n} \leq 1 \vee n \rho_n \leq n \rho_n,
\end{equation*}
where in the last step we have used $n \rho_n \rightarrow \infty$ as $n
\rightarrow \infty$.
Using this in~\eqref{eq:graphon-elln-star-1-bound}, we realize that  with probability one we have
    \begin{equation*}
      \limsup_{n \rightarrow \infty} \frac{\elln_{*,1}}{\barmn} \leq \limsup_{n \rightarrow \infty} \frac{3 + n + 3 \log n + n \log 2 + \frac{n}{2} \log (n \rho_n)}{(n-1)(n\rho_n) /2} = 0,
    \end{equation*}
    which shows~\eqref{eq:graphon-analysis-elln-star-1-barmn-zero}.

    Now, we
    study $\elln_{*,2}$. 
Using ${r \choose s} \le (\frac{re}{s})^s$, we can write
\black
    \begin{equation}
      \label{eq:graphon-analysis-elln-*-2-bound}
      \begin{aligned}
        \elln_{*,2} &\leq 2 {\beta_n^*}^2 + \sum_{i=1}^{\beta_n^*} 2 \log |R_n| + m^*_{i,i} \log \frac{{(n^*_i)}^2 e}{2 m^*_{i,i}}  \\
        &\qquad \qquad + \sum_{1 \leq i < j \leq \beta^*_n}  2 \log |R_n| + m^*_{i,j} \log \frac{n^*_i n^*_j e}{m^*_{i,j}} \\
        &\stackrel{(a)}{\leq} 2 \beta_n^2 (1+\log |R_n|) + \sum_{i=1}^{\beta_n^*} m^*_{i,i} + m^*_{i,i} \log \frac{n_i^2}{2 m^*_{i,i}} + \sum_{1 \leq i < j \leq \beta^*_n} m^*_{i,j} + m^*_{i,j} \log \frac{n_in_j}{m^*_{i,j}} \\
        &\stackrel{(b)}{\leq} 6 \beta_n^2 \log |R_n| + \sum_{i=1}^{\beta_n^*} m^*_{i,i} + m^*_{i,i} \log \frac{n_i^2}{2 m^*_{i,i}} + \sum_{1 \leq i < j \leq \beta^*_n} m^*_{i,j} + m^*_{i,j} \log \frac{n_in_j}{m^*_{i,j}} \\
        &= 6 \beta_n^2 \log |R_n| + \mn_* + \sum_{i=1}^{\beta_n^*} \frac{n_i^2}{2} \frac{2m^*_{i,i}}{n_i^2} \log \frac{n_i^2}{2 m^*_{i,i}} + \sum_{1 \leq i < j \leq \beta_n^*} n_i n_j \frac{m^*_{i,j}}{n_i n_j} \log \frac{n_i n_j}{m^*_{i,j}} \\
        &= 6 \beta_n^2 \log |R_n| + \mn_* + \frac{n^2}{2} \left( \sum_{i=1}^{\beta_n^*} \left( \frac{n_i}{n} \right)^2 \lambda^*_{i,i} \log \frac{1}{\lambda^*_{i,i}} + 2 \sum_{1 \leq i < j \leq \beta_n^*} \frac{n_i}{n} \frac{n_j}{n} \lambda^*_{i,j} \log \frac{1}{\lambda^*_{i,j}} \right) \\
        &= 6 \beta_n^2 \log |R_n| + \mn_* + \frac{n^2}{2} \sum_{i=1}^{\beta_n^*} \sum_{j=1}^{\beta_n^*} \frac{n_i}{n} \frac{n_j}{n} \lambda^*_{i,j} \log \frac{1}{\lambda^*_{i,j}}, \\
      \end{aligned}
    \end{equation}
where $(a)$ uses $\beta_n^* \leq \beta_n$ and $n_i^* \leq n_i$ and in $(b)$, we have used the fact that since $|R_n| \geq 2$, $1 + \log |R_n|
\leq 3 \log |R_n|$.
    Note that
    \begin{equation}
      \label{eq:Ent-hWn-*-expand}
    \begin{aligned}
      \Ent(\hWn_*) &= \ev{\hWn_* \log \hWn_*} - \ev{\hWn_*} \log \ev{\hWn_*} \\
      &= \sum_{i=1}^{\beta_n^*} \sum_{j=1}^{\beta_n^*} \frac{n_i}{n} \frac{n_j}{n} \lambda^*_{i,j} \log \lambda^*_{i,j} - \ev{\hWn_*} \log \ev{\hWn_*}.
    \end{aligned}
    \end{equation}
    But
    \begin{equation}
      \label{eq:ev-hWn-*}
      \begin{aligned}
        \ev{\hWn_*} = \sum_{i=1}^{\beta_n^*} \sum_{j=1}^{\beta_n^*} \frac{n_i}{n} \frac{n_j}{n} \lambda^*_{i,j} &= \sum_{i=1}^{\beta_n^*} \frac{n_i^2}{n^2} \frac{2m^*_{i,i}}{n_i^2} + 2 \sum_{1 \leq i< j \leq \beta_n^*} \frac{n_i}{n} \frac{n_j}{n} \frac{m^*_{i,j}}{n_i n_j} \\
        &= \frac{2}{n^2} \left( \sum_{i=1}^{\beta_n^*} m^*_{i,i} +  \sum_{1 \leq i < j \leq \beta_n^*} m^*_{i,j}\right) \\
        &= \frac{2 \mn_*}{n^2}.
      \end{aligned}
    \end{equation}
Simplifying the bound in~\eqref{eq:graphon-analysis-elln-*-2-bound} using
\eqref{eq:Ent-hWn-*-expand} and \eqref{eq:ev-hWn-*}, we get
\begin{equation}
\label{eq:graphon-analysis-elln-*-2-bound-2}
  \elln_{*,2} \leq 6 \beta_n^2 \log |R_n| + \mn_* - \frac{n^2}{2} \Ent(\hWn_*) - \mn_* \log \frac{2\mn_*}{n^2}.
\end{equation}
    We claim that
    \begin{equation}
      \label{eq:graphon-analysis-betan-log-Rn-claim}
      \limsup_{n \rightarrow \infty} \frac{\beta_n^2 \log |R_n|}{\barmn} = 0 \qquad \text{a.s.}.
    \end{equation}
To see this, note that if $\alpha_n \leq e^2$, then $\beta_n = \phi(\alpha_n) = 1$ and 
    $\beta_n^2 \log
    |R_n| = \log |R_n| \leq \log n$. 
\black 
On the other hand, if $\alpha_n > e^2$,
    recalling the definition of $\alpha_n$ in~\eqref{eq:alphan-def}, since $\alpha_n \leq \mn_* / n$, we have
    \begin{equation*}
      \beta_n^2 = \phi^2(\alpha_n) = \frac{\alpha_n}{\log^2 \alpha_n} \leq \alpha_n \leq \frac{\mn_*}{n}.
    \end{equation*}
    Thereby, we have $\beta_n^2 \log |R_n| \leq \frac{\mn_*}{n} \log n$.
    Combining the two cases, we get
    \begin{equation*}
      \limsup_{n \rightarrow \infty} \frac{\beta_n^2 \log |R_n|}{\barmn} \leq \limsup_{n \rightarrow \infty} \frac{(1+\frac{\mn_*}{n}) \log n}{\barmn} \leq \limsup_{n \rightarrow \infty} \frac{2\log n}{n(n-1) \rho_n} + \limsup_{n \rightarrow \infty} \frac{\mn_*}{\barmn} \frac{\log n}{n}.
    \end{equation*}
    But $n \rho_n \rightarrow \infty$, and from~\eqref{eq:graphon-analysis-mn-star-barmn-1},
    $\mn_* / \barmn \rightarrow 1$
    a.s..\ Hence, we arrive at~\eqref{eq:graphon-analysis-betan-log-Rn-claim}.
    Using~\eqref{eq:graphon-analysis-betan-log-Rn-claim} back in~\eqref{eq:graphon-analysis-elln-*-2-bound-2}, we get
    \begin{equation}
      \label{eq:graphon-analysis-elln-star-2-bound}
    \begin{aligned}
      \limsup_{n \rightarrow \infty} \frac{\elln_{*,2} - \mn_* \log \frac{1}{\rho_n}}{\barmn} &\leq \limsup_{n \rightarrow \infty} \frac{\mn_* - \frac{n^2}{2} \Ent(\hWn_*) - \mn_* \log \frac{2\mn_*}{n^2 \rho_n}}{\barmn} \\
                                                                                              &\stackrel{(*)}{=} \limsup_{n \rightarrow \infty} \frac{\mn_* - \frac{n^2 \rho_n}{2} \Ent\left( \frac{1}{\rho_n} \hWn_* \right) - \mn_* \log\left( \frac{\mn_*}{\barmn} \frac{n-1}{n} \right)}{\barmn} \\
      &= \limsup_{n \rightarrow \infty} \frac{\mn_*}{\barmn} - \frac{n}{n-1} \Ent\left( \frac{1}{\rho_n} \hWn_* \right) - \frac{\mn_*}{\barmn} \log\left( \frac{\mn_*}{\barmn} \frac{n-1}{n} \right),
    \end{aligned}
    \end{equation}
where in $(*)$, we have used part~\ref{thm-ge-scale} of
Theorem~\ref{thm:graphon-entropy-props} to replace $\Ent(\hWn_*)$ by $\rho_n
\Ent(\hWn_* / \rho_n)$.
Note that
from~\eqref{eq:graphon-analysis-mn-star-barmn-1}, $\mn_* / \barmn \rightarrow 1$
a.s.. On the other hand, Lemma~\ref{lem:delta-hWns-hWn-goes-to-zero} implies
that $\delta_2(\hWn_*/ \rho_n , \hWn / \rho_n) \rightarrow 0$ a.s..
Also,
Lemma~\ref{lem:hWn-converges-to-W} implies that $\delta_2(\hWn / \rho_n, W)
\rightarrow 0$ a.s.. These two together imply that as $n\rightarrow \infty$,  $\delta_2(\hWn_* / \rho_n, W) \rightarrow 0$
a.s.. Therefore, part~\ref{thm-ge-L2-conv} of Theorem~\ref{thm:graphon-entropy-props}
implies that as $n\rightarrow \infty$, $\Ent(\hWn_* / \rho_n) \rightarrow \Ent(W)$ a.s.. Using these in
the bound~\eqref{eq:graphon-analysis-elln-star-2-bound} above, we realize that 
\begin{equation}
  \label{eq:graphon-analysis-elln-*-2-mn-*-log-rho-n-Ent-W}
  \limsup_{n \rightarrow \infty} \frac{\elln_{*,2} - \mn_* \log \frac{1}{\rho_n}}{\barmn} \leq 1 - \Ent(W) \qquad \text{a.s.}.
\end{equation}

Combining~\eqref{eq:graphon-analysis-elln-Delta-n-zero},
\eqref{eq:graphon-analysis-elln-star-1-barmn-zero}, and
\eqref{eq:graphon-analysis-elln-*-2-mn-*-log-rho-n-Ent-W}, we conclude that with
probability one
    \begin{align*}
            \limsup_{n \rightarrow \infty} \frac{\elln - \barmn \log \frac{1}{\rho_n}}{\barmn} &\leq \limsup_{n \rightarrow \infty} \frac{\elln - \mn \log \frac{1}{\rho_n}}{\barmn} + \limsup_{n \rightarrow \infty} \frac{\mn - \barmn}{\barmn} \log \frac{1}{\rho_n} \\
                                                                                               &\stackrel{(*)}{\leq} \limsup_{n \rightarrow \infty} \frac{\elln_{\Delta_n} - \mn_{\Delta_n} \log \frac{1}{\rho_n}}{\barmn} + \limsup_{n \rightarrow \infty} \frac{\elln_{*,1}}{\barmn} + \limsup_{n \rightarrow \infty} \frac{\elln_{*,2} - \mn_* \log \frac{1}{\rho_n}}{\barmn} \\
      &\leq 1 - \Ent(W),
    \end{align*}
    where in $(*)$, we have used Proposition~\ref{prop:graphon-mn-asymptotics}.
This completes the proof.
\end{proof}

\editfinish

\section{Proof of Converse (Theorem~\ref{thm:lwc-graphon-convesee})}
\label{sec:converse-proof}

\editstart

In this section, we give the proof of our converse results, i.e.\
the two parts of
\black
Theorem~\ref{thm:lwc-graphon-convesee}. 

The first part directly follows from the converse result of
Theorem~\ref{thm:baby-converse} in Section~\ref{sec:prelim-baby-compression}.
More precisely, first note that if $\bch(\mu) = -\infty$ there is nothing to be
proved. If $\bch(\mu) > -\infty$, let $\Gn$ be the sequence of random graphs
obtained from Theorem~\ref{thm:baby-converse}. Assume
that~\eqref{eq:converse-lwc-part-statement} is violated. This means that there
exists some $t < \bch(\mu)$ such that
\begin{equation*}
  \limsup_{n \rightarrow \infty} \frac{\nat(f_n(\Gn) - \mn \log n )}{n} \leq t < \bch(\mu) \qquad \text{a.s.}. 
\end{equation*}
But this 
is in
\black
contradiction with the result of
Theorem~\ref{thm:baby-converse}. This completes the proof of the first part
of Theorem~\ref{thm:lwc-graphon-convesee}. 
\black

Now we prove the second part. 
Assume that a sequence of lossless/decompression maps $((f_n, g_n): n \geq 1)$
is given.
  Consider the lossless compression map $\pp{f}_n: \mG_n \rightarrow
  \{0,1\}^*$ defined as follows. Given a graph $G \in \mG_n$ with $m$ edges, $\pp{f}_n(G)$ is
  comprised of the binary representation of $m$, followed by the index of $G$ among all the
  graphs in $\mG_{n,m}$ which have the same number of edges $m$. Since $m \leq
  \binom{n}{2} < n^2$, and the number of the graphs with $m$ edges is precisely
  $\binom{\binom{n}{2}}{m}$, we have 
  \begin{equation*}
    \len(\pp{f}_n(G)) \leq 2 + 2 \log_2 n + \log_2 \binom{\binom{n}{2}}{m} =: l_{n,m}.
  \end{equation*}
  Now, we define another compression map $\tf_n:  \mG_n \rightarrow
  \{0,1\}^*$  as follows. Assume that $G \in \mG_n$ is given. If $\len(f_n(G))
  \leq l_{n,m}$, define  $b \in \{0,1\}^*$ to be obtained
  by concatenating the binary representation of $\len(f_n(G))$ using $1 + \lfloor
  \log_2(l_{n,m}) \rfloor$ bits 
  followed by $f_n(G)$. Thereby
  \begin{equation}
    \label{eq:bits-b-bound}
    \len(b) \leq 1 + \log_2 l_{n,m} + \len(f_n(G)).
  \end{equation}
Then, if $\len(b) < \len(\pp{f}_n(G))$, we define
  $\tf_n(G)$ to be a single bit with value zero followed by $b$. Otherwise, if
  either 
  $\len(f_n(G)) > l_{n,m}$,  or $\len(f_n(G)) \leq l_{n,m}$ and $\len(b) \geq
  \len(\pp{f}_n(G))$, we
  define $\tf_n(G)$ to be a single bit with value one followed by $\pp{f}_n(G)$.
  Observe that $\tf_n$ defined above satisfies the prefix condition.
  Additionally, since
  both $f_n$ and $f'_n$ are lossless, $\tilde{f}_n$ is also lossless. Moreover, for all $G
  \in \mG_n$ with $m$ edges, we have
  \begin{equation*}
    \len(\tf_n(G)) \leq 1 + l_{n, m} = 3 + 2 \log_2 n + \log_2 \binom{\binom{n}{2}}{m},
  \end{equation*}
  or equivalently
  \begin{equation}
    \label{eq:nat-tf-upperbound}
    \nat(\tf_n(G)) \leq 3 \log 2 + 2 \log n + \log \binom{\binom{n}{2}}{m}.
  \end{equation}
  In addition to this, we claim that for all $G \in \mG_{n}$ having $m$ edges
  we have 
  \begin{equation*}
    \len(\tf_n(G)) \leq (1 + 1) + \log_2 l_{n,m} + \len(f_n(G)),
  \end{equation*}
  or equivalently
  \begin{equation}
    \label{eq:nat-tfn-G-nat-fnG-bound}
    \nat(\tf_n(G)) \leq (1 + 1)\log 2 + \log l_{n,m} + \nat(f_n(G)).
  \end{equation}
  To see this, observe that if $\len(f_n(G)) \leq l_{n,m}$ and $\len(b) <
  \len(\pp{f}_n(G))$, then using~\eqref{eq:bits-b-bound} 
  we have $\len(\tf_n(G)) = 1 + \len(b) \leq 1 + (1 + \log_2
  l_{n,m}) + \len(f_n(G))$. On the other hand, if $\len(f_n(G)) \leq l_{n,m}$ and
  $\len(b) \geq \len(\pp{f}_n(G))$, we have $\len(\tf_n(G)) = 1 +
  \len(\pp{f}_n(G)) \leq 1 + \len(b) \leq 1 + (1 + \log_2 l_{n,m}) +
  \len(f_n(G))$, where the last step again uses~\eqref{eq:bits-b-bound}.
  Finally, if $\len(f_n(G)) > l_{n,m}$, we have $\len(\tf_n(G)) = 1 +
  \len(\pp{f}_n(G)) \leq 1 + l_{n,m} \leq 1 + \len(f_n(G)) \leq 1 + (1 + \log_2 l_{n,m}) +
  \len(f_n(G))$. Hence, we have verified that the claimed bound 
  in~\eqref{eq:nat-tfn-G-nat-fnG-bound} holds in all the three cases.

  Now let $W$ and the sequence $\rho_n$ be as in the statement of
  Theorem~\ref{thm:lwc-graphon-convesee} and let $\Gn$ be a sequence  of
  $W$-random graphons with target density $\rho_n$. Let $\mn$ denote the
  number of edges in $\Gn$. Using~\eqref{eq:nat-tfn-G-nat-fnG-bound}, for all $t$,  we have
  \begin{equation}
    \label{eq:pr-limspu-nat-fn-nat-tfn}
\begin{aligned}
  &\pr{\limsup_{n \rightarrow \infty} \frac{\nat(f_n(\Gn)) - \barmn \log \frac{1}{\rho_n}}{\barmn} \leq t} \\
  &\qquad \leq \pr{\limsup_{n \rightarrow \infty} \frac{\nat(\tf_n(\Gn)) - 2 \log 2 - \log l_{n, \mn} - \barmn \log \frac{1}{\rho_n}}{\barmn} \leq t}.
  \end{aligned}
  \end{equation}
  Using a crude upper bound, we have
  \begin{equation*}
    \log l_{n,\mn} \leq \log \left( 2 + \log_2 n + \log_2 2^{\binom{n}{2}} \right) \leq \log(2 +  n + n^2) \leq \log(n+1)^2 = 2 \log (n+1).
  \end{equation*}
  On the other hand, we have $\barmn = \binom{n}{2} \rho_n$ and $n\rho_n
  \rightarrow \infty$ as $n \rightarrow \infty$. Consequently, with probability
  one we have
  \begin{equation*}
     \lim_{n \rightarrow \infty} \frac{2 \log  2 + \log l_{n,\mn}}{\barmn} = 0.
   \end{equation*}
   Comparing this
  with~\eqref{eq:pr-limspu-nat-fn-nat-tfn} above, we realize that in order to
  show~\eqref{eq:converse-graphon-part-statement}, it suffices to show that
  \begin{equation}
    \label{eq:converse-graphon-suffices-nat-tfn}
    \pr{\limsup_{n \rightarrow \infty} \frac{\nat(\tf_n(\Gn)) - \barmn \log \frac{1}{\rho_n}}{\barmn} \leq t}  < 1 \qquad \forall t < 1 - \Ent(W).
  \end{equation}
We fix   $t < 1 - \Ent(W)$ and
define the random variables
  \begin{equation*}
    L_n := \frac{\nat(\tf_n(\Gn)) - \barmn \log \frac{1}{\rho_n}}{\barmn},
  \end{equation*}
  and
  \begin{equation*}
    U_n = \frac{3 \log 2 + 2 \log n + \log \binom{\binom{n}{2}}{\mn} - \barmn \log \frac{1}{\rho_n}}{\barmn}.
  \end{equation*}
  Note that, from~\eqref{eq:nat-tf-upperbound}, with probability one we have
  \begin{equation*}
    L_n \leq U_n \qquad \forall n.
  \end{equation*}
  Thereby, employing Fatou's lemma, we get
  \begin{equation}
    \label{eq:Ln-Un-Fatou}
    \liminf_{n \rightarrow \infty} \ev{U_n - L_n} \geq \ev{\liminf_{n \rightarrow \infty} (U_n - L_n)}.
  \end{equation}
  Note that 
  \begin{align*}
    \liminf_{n \rightarrow \infty} \ev{U_n - L_n} &= \liminf_{n \rightarrow \infty} (\ev{U_n} + \ev{-L_n}) \\
    &\leq \liminf_{n \rightarrow \infty} \ev{-L_n} + \limsup_{n \rightarrow \infty} \ev{U_n} \\ 
    &= - \limsup_{n \rightarrow \infty} \ev{L_n} + \limsup_{n \rightarrow \infty} \ev{U_n}.
  \end{align*}
  Combining this with~\eqref{eq:Ln-Un-Fatou}, we get
  \begin{equation}
    \label{eq:lims-Un-lims-Ln-ev-limi-Un-Ln}
    \limsup_{n \rightarrow \infty} \ev{U_n} - \limsup_{n \rightarrow \infty} \ev{L_n} \geq \ev{\liminf_{n \rightarrow \infty} (U_n - L_n)}.
  \end{equation}
  Note that $\barmn = \binom{n}{2} \rho_n = \frac{n-1}{2} n \rho_n$ and $n
  \rho_n \rightarrow \infty$. Hence, $(3 \log 2 + 2 \log n) / \barmn \rightarrow
  0$ as $n \rightarrow \infty$. Thereby, from
  Part~\ref{item:prop-mn-asymp-log-n2-1-as} of
  Proposition~\ref{prop:graphon-mn-asymptotics} in Section~\ref{sec:graphon-analysis},
  we have
  \begin{equation}
    \label{eq:lim-Un-1-as}
    \lim_{n \rightarrow \infty} U_n = 1 \qquad \text{a.s.}.
  \end{equation}
  Also, using
  Part~\ref{item:prop-mn-asymp-log-n2-1-ev} of Proposition~\ref{prop:graphon-mn-asymptotics},
  we have
  \begin{equation}
    \label{eq:limsup-Ev-Un-1}
    \limsup_{n \rightarrow \infty} \ev{U_n} \leq 1.
  \end{equation}
  Using~\eqref{eq:lim-Un-1-as} and \eqref{eq:limsup-Ev-Un-1}
  in~\eqref{eq:lims-Un-lims-Ln-ev-limi-Un-Ln}, we get
  \begin{equation*}
    1 - \limsup_{n \rightarrow \infty} \ev{L_n} \geq 1 - \ev{\limsup_{n \rightarrow \infty} L_n},
  \end{equation*}
  or equivalently,
  \begin{equation}
    \label{eq:limsup-ev-Ln-ev-limsup-Ln}
    \limsup_{n \rightarrow \infty} \ev{L_n}  \leq \ev{\limsup_{n \rightarrow \infty} L_n}.
  \end{equation}
  Note that $\tf_n$ is lossless and satisfies the  prefix condition, which implies that $\ev{\nat(\tf_n(\Gn))}
  \geq H(\Gn)$. Consequently, using
  Proposition~\ref{prop:graphon-entropy-asymptotics}, we have
  \begin{equation*}
    \limsup_{ n \rightarrow \infty} \ev{L_n} \geq \limsup_{n \rightarrow \infty} \frac{H(\Gn) - \barmn \log \frac{1}{\rho_n}}{\barmn} = 1 - \Ent(W).
  \end{equation*}
  Combining this with~\eqref{eq:limsup-ev-Ln-ev-limsup-Ln}, we get
  \begin{equation*}
    1 - \Ent(W) \leq \ev{\limsup_{n \rightarrow \infty} L_n}.
  \end{equation*}
  This establishes~\eqref{eq:converse-graphon-suffices-nat-tfn} and completes
  the proof of Theorem~\ref{thm:lwc-graphon-convesee}.

\section*{Acknowledgements}

Research of the authors was supported by the NSF grants CNS–1527846,
CCF–1618145, CCF-1901004, CIF-2007965, the NSF Science
\& Technology Center grant CCF–0939370 (Science of
Information), and the William and Flora Hewlett Foundation
supported Center for Long Term Cybersecurity at Berkeley.

 \black
 
 
  \editfinish

\appendix

\section{Proof of Theorem~\ref{thm:graphon-entropy-props}}
\label{sec:app-graphon-ent-properties}

In this section we give the proof of Theorem~\ref{thm:graphon-entropy-props}.
Before that, we state and prove the following lemma.

\begin{lem}
  \label{lem:xlogx}
  Assume that $x, y \geq 0$ and $\max\{x, y\} \geq 1$. Then, we have
  \begin{equation*}
    | x\log x - y \log y| \leq |x-y|(1 + \max\{x,y\}).
  \end{equation*}
  Here $0 \log 0 := 0$. 
\end{lem}
\begin{proof}
The claim is obvious when $x = y$. 
Hence assume without loss of generality that $x > y$. 
By the convexity of $x \log x$ on $[0, \infty)$
we have 
\[
x\log x - y \log y \leq (x-y)(1 + \log x).
\]
Since $x \ge 1$ by assumption, we have $0 \le \log x \le x$. Substituting this on the right hand side in the preceding inequality completes the proof.
\black
\end{proof}

\begin{proof}[Proof of Theorem~\ref{thm:graphon-entropy-props}]
To see part~\ref{thm-ge-well}, 
note that $|x \log x|
    \leq \frac{1}{e} + x^2$ on $[0,\infty)$.
  Thereby, we have
    \begin{equation}
      \label{eq:W-L2-W-log-W-finite}
      \int |W(x,y) \log W(x,y)| d \pi(x) d \pi(y) \leq \frac{1}{e} + \int |W(x,y)|^2 d \pi(x) d \pi(y) < \infty.
    \end{equation}  
    Since $W$ is $L^2$, this establishes that $\Ent(W)$ is indeed well defined and
    $\Ent(W) < \infty$. 
\black
    To see part~\ref{thm-ge-Ent-pos}, using convexity of $x
    \mapsto x \log x$, we have $\ev{W \log W} \geq \ev{W} \log \ev{W}$, which
    means that $\Ent(W) \geq 0$. Part~\ref{thm-ge-scale} follows directly from
    the definition of $\Ent(W)$.
    Now, we give the proof of part~\ref{thm-ge-L2-conv}.
Since
\black
$\delta_2(W_n, W) \rightarrow 0$, we have $\ev{W_n} \rightarrow
\ev{W}$ and $\ev{W_n} \log \ev{W_n} \rightarrow \ev{W} \log \ev{W}$. Therefore, in order to show that $\Ent(W_n) \rightarrow \Ent(W)$
as $n \rightarrow \infty$,
\black
it
suffices to show that
\begin{equation}
  \label{eq:L2-convergence-xlogx-claim}
  \ev{W_n \log W_n} = \int W_n(x,y) \log W_n(x,y)d\pi_n(x) d\pi_n(y) \rightarrow \int W(x,y) \log W(x,y) d\pi(x) d\pi(y)= \ev{W \log W}.
\end{equation}
 Using the definition of the $\delta_2$ norm in~\eqref{eq:delta-2-def}, 
 we can find for  each $n$
 \black
a coupling $\nu_n$ of $\pi_n$ and $\pi$ such that
\begin{equation}
  \label{eq:nu-n-coupling-property}
  \int |W_n(x,y) - W(x',y')|^2 d \nu_n(x,x') d\nu_n(y,y') \leq \delta_2^2(W_n, W) + \frac{1}{n}.
\end{equation}
Note that we have
\begin{equation}
  \label{eq:Ent-Wn-Ent-W-diff-1}
  \begin{aligned}
    \ev{W_n \log W_n} - \ev{W \log W} 
    &= \int (W_n(x,y) \log W_n(x,y) - W(x',y') \log W(x',y')) d \nu_n(x,x') d \nu_n(y,y').
  \end{aligned}
\end{equation}
\black
To simplify the notation, define $\mu_n:= \nu_n \times \nu_n$ to be the product
measure on $(\Omega_n \times \Omega) \times (\Omega_n \times \Omega)$. Moreover,
define
\begin{equation}
  \label{eq:Bn-def}
  B_n := \{ (x,x', y,y') : W_n(x,y) \vee W(x',y') \geq 1\}.
\end{equation}
Then, using~\eqref{eq:Ent-Wn-Ent-W-diff-1}, we can write
\begin{equation}
  \label{eq:Ent-diff-B-Bc}
  \begin{aligned}
    |\ev{W_n \log W_n} - \ev{W \log W}| &\leq \int_{B_n} |W_n(x,y) \log W_n(x,y) - W(x',y') \log W(x',y')| d \mu_n(x,x',y,y')  \\
    &\quad +\int_{B_n^c} |W_n(x,y) \log W_n(x,y) - W(x',y') \log W(x',y')| d \mu_n(x,x',y,y').  
  \end{aligned}
\end{equation}
We bound each term separately. We start with the integral over $B_n$. Using
Lemma~\ref{lem:xlogx}, we have
\begin{align*}
  \int_{B_n} |W_n\log W_n - W \log W| d\mu_n &\leq \int_{B_n} |W_n - W|(1 + \max\{W_n, W\}) d \mu_n \\
                                             &\leq \int_{B_n} |W_n - W| d \mu_n + \int_{B_n} |W_n - W|(W + W_n) d \mu_n \\
                                             &\leq \int |W_n - W| d\mu_n + \int |W_n - W|(W + W_n) d \mu_n \\
                                             &\leq \left( \int |W_n - W|^2 d \mu_n \right)^{1/2} \left[  1 + \left( \int (W +W_n)^2 d \mu_n \right)^{1/2} \right] \\
                                             &\leq \left (\delta^2_2(W_n, W) + \frac{1}{n}\right )^{1/2} \left( 1 + \left(\int W^2 d\mu_n\right)^{1/2} + \left(\int W_n^2 d \mu_n\right)^{1/2} \right). 
\end{align*}
Note that, by assumption, we have $\delta_2(W_n , W) \rightarrow 0$. Also, $\left(\int
  W^2 d\mu_n\right)^{1/2} < \infty$ since $W$ is a $L^2$
graphon. Furthermore, $W_n$ is $L^2$ and $\delta_2(W_n, W) \rightarrow 0$, hence
$\left(\int W_n^2 d \mu_n\right)^{1/2}$ is a bounded sequence. Therefore, we
have
\begin{equation}
  \label{eq:int-Bn-zero}
  \lim_{n \rightarrow \infty} \int_{B_n} |W_n\log W_n - W \log W| d\mu_n = 0.
\end{equation}

Next, we focus on the second term in~\eqref{eq:Ent-diff-B-Bc}, i.e.\ the integral over $B_n^c$. Fix $\epsilon > 0$ and define
\begin{equation}
  \label{eq:Anepsilon-def}
  A_{n, \epsilon} := \{(x,x',y,y') : |W_n(x,y) - W(x',y')| > \epsilon\}.
\end{equation}
With this, we split the integral over $B_n^c$ as follows:
\begin{align*}
  \int_{B_n^c}  |W_n \log W_n - W \log W| d\mu_n &= \int_{B_n^c \cap A_{n, \epsilon}} |W_n \log W_n - W \log W| d \mu_n \\
  &\qquad + \int_{B_n^c \cap A_{n, \epsilon}^c} |W_n \log W_n - W \log W| d \mu_n.
\end{align*}
Recalling the definition of $B_n$ from~\eqref{eq:Bn-def}, for $(x, x',y, y') \in
B_n^c$, we have $W_n(x,y) < 1$ and $W(x',y') < 1$. Since $|x \log x| \leq
\frac{1}{e}$ for $x \in [0,1]$, we have
\begin{align*}
  \int_{B_n^c \cap A_{n, \epsilon}} |W_n \log W_n - W \log W|d \mu_n &\leq \frac{2}{e} \mu_n(A_{n, \epsilon}) \\
&\leq \frac{2}{e} \frac{\int |W_n - W|^2 d \mu_n}{\epsilon^2} \\
 &\leq \frac{2}{e} \frac{\delta_2(W_n, W)^2 + 1/n}{\epsilon^2}.
\end{align*}
Since $\delta_2(W_n, W) \rightarrow 0$, we have
\begin{equation}
  \label{eq:lim-Bnc-An-0}
  \lim_{n \rightarrow \infty}   \int_{B_n^c \cap A_{n, \epsilon}} |W_n \log W_n - W \log W|d \mu_n = 0.
\end{equation}
Moreover, since the function $x \log x$ is uniformly continuous for $x \in
[0,1]$, for $x, x' \in [0,1]$ such that $|x - x'| \leq \epsilon$, we have $|x
\log x - x' \log x'| \leq \delta(\epsilon)$, where $\delta(\epsilon) \rightarrow
0$ as $\epsilon \rightarrow 0$. Hence, we have 
\begin{equation*}
  \int_{B_n^c \cap A_{n, \epsilon}^c} |W_n \log W_n - W \log W| d \mu_n \leq \delta(\epsilon).
\end{equation*}
Therefore,
\begin{equation*}
  \limsup_{n \rightarrow 0}   \int_{B_n^c \cap A_{n, \epsilon}^c} |W_n \log W_n - W \log W|d \mu_n \leq \delta(\epsilon).
\end{equation*}
This together with~\eqref{eq:lim-Bnc-An-0} implies that
\begin{equation*}
  \limsup_{n \rightarrow \infty} \int_{B_n^c} |W_n \log W_n - W \log W| d \mu_n \leq \delta(\epsilon).
\end{equation*}
Since this holds for all $\epsilon > 0$ and $\delta(\epsilon) \rightarrow 0$ as
$\epsilon \rightarrow 0$, by sending $\epsilon$ to zero
we have
\black
\begin{equation*}
  \lim_{n \rightarrow 0}   \int_{B_n^c} |W_n \log W_n - W \log W|d \mu_n = 0.
\end{equation*}
This together with~\eqref{eq:int-Bn-zero} and~\eqref{eq:Ent-diff-B-Bc} implies
that $|\ev{W_n \log W_n} - \ev{W \log W} | \rightarrow 0$ which is
precisely~\eqref{eq:L2-convergence-xlogx-claim}. This completes the proof of part~\ref{thm-ge-L2-conv}.
\end{proof}

\section{Proof of Proposition~\ref{prop:graphon-entropy-asymptotics}}
\label{sec:app-graphon-random-graph-asymp-entropy}

\begin{proof}[Proof of Proposition~\ref{prop:graphon-entropy-asymptotics}]
  Note that since $\rho_n \rightarrow 0$
  we have $\rho_n < 1$
  for all $n$ large enough.
  \black
  Also, since $n \rho_n \rightarrow \infty$, for $n$ large enough
  we have $\rho_n > 1 / n > 0$. Therefore, throughout the proof, we may
  assume that $n$ is large enough so that $0 < \rho_n < 1$.

  We prove the result in two steps. First, we show that
  \begin{equation}
    \label{eq:graphon-asymp-ent-lowerboun-claim}
    \liminf_{n \rightarrow \infty} \frac{H(\Gn) - \barmn \log \frac{1}{\rho_n}}{\barmn} \geq  1 - \Ent(W).
  \end{equation}
  Recall that in order to generate $\Gn$ we start with an i.i.d.\ sequence
  $(X_i)_{i=1}^\infty$ from distribution $\pi$ 
  and 
  connect 
  \black
  two nodes $1 \leq i,j \leq n$, $i \neq j$, with probability
  $\rho_n W(X_i, X_j) \wedge 1$. Note that, conditioned on $X_{[1:n]}$, the
  placement of edges is performed independently for each pair of vertices.
  Therefore, with $H_b(x):= - x \log x  - (1-x) \log (1-x)$ denoting the binary
  entropy of $x \in [0,1]$
  to the natural base, and 
  \black
  identifying $0 \log 0 \equiv 0$ as usual, we may write
  \begin{equation}
    \label{eq:graphon-ent-asym-lower-bound-1}
    \begin{aligned}
      H(\Gn) &\geq H(\Gn | X_{[1:n]}) \\
      &= \sum_{1 \leq i < j \leq n} H(\one{i \sim_{\Gn} j} | X_{[1:n]}) \\
      &= \evwrt{X_{[1:n]}}{\sum_{1 \leq i < j \leq n} H_b(\rho_n W(X_i, X_j) \wedge 1)} \\
      &= \binom{n}{2} \ev{H_b(\rho_n W(X_1, X_2) \wedge 1)}\\
      &= \binom{n}{2} \ev{H_b(\rho_n W
    \wedge 1)},
    \end{aligned}
  \end{equation}
    where in the last line we view $W$ as a random variable on $\Omega \times \Omega$ with probability distribution $\pi \times \pi$.
    \black
We may write
  \begin{equation}
    \label{eq:Hb-rhon-W-split}
    \ev{H_b(\rho_n W \wedge 1)} = \ev{\one{W \leq \frac{1}{\rho_n}} \rho_n W \log \frac{1}{\rho_n W}} + \ev{\one{W \leq \frac{1}{\rho_n}} (1-\rho_n W) \log \frac{1}{1-\rho_n W}}.
  \end{equation}
  We continue by bounding each term separately.

  For the first term in~\eqref{eq:Hb-rhon-W-split}, we may write
  \begin{align*}
    \ev{\one{W \leq \frac{1}{\rho_n}} \rho_n W \log \frac{1}{\rho_n W}} &= \rho_n \log \frac{1}{\rho_n} \ev{\one{W \leq \frac{1}{\rho_n}} W} + \rho_n \ev{\one{W \leq \frac{1}{\rho_n}} W \log \frac{1}{W}}\\
&= \rho_n \log \frac{1}{\rho_n} \ev{W} - \rho_n \log \frac{1}{\rho_n} \ev{\one{W > \frac{1}{\rho_n}} W}  \\
    &\quad + \rho_n \ev{\one{W \leq \frac{1}{\rho_n}} W \log \frac{1}{W}}.
  \end{align*}
  Since 
  \black
  $W$ is normalized by assumption, we have $\ev{W} = 1$.
  Moreover, we have
  \begin{equation*}
    \one{W > \frac{1}{\rho_n}} W = \one{\rho_n W > 1} W \leq \rho_n W^2.
  \end{equation*}
  Therefore, we have
  \begin{equation*}
    \ev{\one{W \leq \frac{1}{\rho_n}} \rho_n W \log \frac{1}{\rho_n W}} \geq \rho_n \log \frac{1}{\rho_n} - \rho_n^2 \log \frac{1}{\rho_n} \ev{W^2} + \rho_n \ev{\one{W \leq \frac{1}{\rho_n}} W \log \frac{1}{W}}.
  \end{equation*}
  Multiplying both sides by $\binom{n}{2}$, then  dividing by $\barmn$, and
  recalling $\barmn = \binom{n}{2} \rho_n$, we
  realize that
  \begin{equation}
    \label{eq:ev-rho-n-W-log-rho-n-W-bound-1}
\begin{aligned}
    \liminf_{n \rightarrow \infty} \frac{\binom{n}{2}\ev{\one{W \leq 1 / \rho_n} \rho_n W \log \frac{1}{\rho_n W}} - \barmn \log \frac{1}{\rho_n}}{\barmn} &\geq \liminf_{n \rightarrow \infty} \ev{\one{W \leq \frac{1}{\rho_n}} W \log \frac{1}{W}} \\
    &\quad - \limsup_{n \rightarrow \infty} \rho_n \log \frac{1}{\rho_n} \ev{W^2}.
    \end{aligned}
  \end{equation}
  Since $W$ is a $L^2$ graphon
  we have
  \black
  $\ev{W^2} < \infty$. Moreover, since $\rho_n
  \rightarrow 0$
  we have
  \black
  $\rho_n \log 1 / \rho_n \rightarrow 0$. Thereby,
  \begin{equation}
    \label{eq:rho-n-log-rho-n-ev-W2}
    \lim_{n \rightarrow \infty} \rho_n \log \frac{1}{\rho_n} \ev{W^2} = 0.
  \end{equation}
  On the other hand, since $W$ is $L^2$, from~\eqref{eq:W-L2-W-log-W-finite} we
  know that $\ev{|W \log 1 / W|} < \infty$. Therefore, as $1 / \rho_n
  \rightarrow \infty$, the dominated convergence theorem implies that
  \begin{equation}
    \label{eq:W<1-rhon-W-log-W-to-ent}
    \lim_{n \rightarrow \infty} \ev{\one{W \leq \frac{1}{\rho_n}} W \log \frac{1}{W}} = \ev{W \log \frac{1}{W}} = -\Ent(W).
  \end{equation}
  Substituting~\eqref{eq:rho-n-log-rho-n-ev-W2} and \eqref{eq:W<1-rhon-W-log-W-to-ent}
  into~\eqref{eq:ev-rho-n-W-log-rho-n-W-bound-1}, we get
  \begin{equation}
    \label{eq:lower-bound-rhoW-log-rhoW-part-final}
    \liminf_{n \rightarrow \infty} \frac{\binom{n}{2}\ev{\one{W \leq 1 / \rho_n} \rho_n W \log \frac{1}{\rho_n W}} - \barmn \log \frac{1}{\rho_n}}{\barmn}  \geq -\Ent(W).
  \end{equation}

  Now we turn to the second term on the right hand side
  of~\eqref{eq:Hb-rhon-W-split}. Let $\tau_n := 1 / \sqrt{\rho_n} \leq 1/
  \rho_n$  and note that
  \begin{equation}
    \label{eq:1-rhon-taun-W}
    \ev{\one{W \leq \frac{1}{\rho_n}} (1-\rho_n W) \log \frac{1}{1 - \rho_n W}} \geq \ev{\one{W \leq \tau_n} (1 - \rho_n W) \log \frac{1}{1 - \rho_n W}}.
  \end{equation}
  Using 
  the
  \black
  Taylor remainder theorem, for $x \geq 0$, we can write
  \begin{equation*}
    (1 - x) \log \frac{1}{1-x}= x + x \eta(x),
  \end{equation*}
  where $\eta(x) \rightarrow 0$ as $x \rightarrow 0$. Thereby,
  \begin{equation*}
    \ev{\one{W \leq \tau_n} (1 - \rho_n W) \log \frac{1}{1 - \rho_n W}} = \rho_n \ev{\one{W \leq \tau_n} W} + \rho_n \ev{\one{W \leq \tau_n}  W \eta(\rho_n W)}.
  \end{equation*}
  Using this in~\eqref{eq:1-rhon-taun-W}, multiplying both sides by
  $\binom{n}{2}$, and then dividing by $\bar{m}_n$, we get
  \begin{equation}
    \label{eq:1-rho-nW-log-1-rhon-W-bound-1}
    \frac{\binom{n}{2}\ev{\one{W \leq \frac{1}{\rho_n}} (1-\rho_n W) \log \frac{1}{1 - \rho_n W}}}{\barmn}  \geq \ev{\one{W \leq \tau_n} W} + \ev{\one{W \leq \tau_n} W \eta(\rho_n W)}.
  \end{equation}
Note that $\ev{W} = 1$ and $\tau_n \rightarrow \infty$ as $n \rightarrow
  \infty$. Hence,
  \begin{equation}
    \label{eq:W-taun-W-1}
    \lim_{n \rightarrow \infty} \ev{\one{W \leq \tau_n} W} = 1.
  \end{equation}
  On the other hand, when $W \leq \tau_n$
  we have
  \black
  $\rho_n W \leq \rho_n \tau_n =
  \sqrt{\rho_n}$. Recall that $\eta(x) \rightarrow 0$ as $x \rightarrow 0$, and
  $\sqrt{\rho_n} \rightarrow 0$ as $n \rightarrow \infty$. Hence, we can
  conclude that there exists a sequence $\epsilon_n \rightarrow 0$ such that
  $|\eta(\rho_n W)| \leq \epsilon_n$ when $W \leq \tau_n$. Therefore,
  \begin{equation*}
    \ev{|\one{W \leq \tau_n} W \eta(\rho_n W)|} \leq \epsilon_n \ev{\one{W \leq \tau_n} W} \leq \epsilon_n \ev{W} = \epsilon_n.
  \end{equation*}
  Hence,
  \begin{equation}
    \label{eq:W-taun-phi-0}
    \lim_{n \rightarrow \infty} \ev{\one{W \leq \tau_n} W \eta(\rho_n W)} = 0.
  \end{equation}
  Substituting~\eqref{eq:W-taun-W-1} and~\eqref{eq:W-taun-phi-0} back
  into~\eqref{eq:1-rho-nW-log-1-rhon-W-bound-1}, we realize that
  \begin{equation}
    \label{eq:1-rhon-W-log-1-rho-n-W-final-bound}
    \liminf_{n \rightarrow \infty} \frac{\binom{n}{2}\ev{\one{W \leq \frac{1}{\rho_n}} (1-\rho_n W) \log \frac{1}{1 - \rho_n W}}}{\barmn}  \geq 1.
  \end{equation}
  Putting together~\eqref{eq:graphon-ent-asym-lower-bound-1}, \eqref{eq:Hb-rhon-W-split},
  \eqref{eq:lower-bound-rhoW-log-rhoW-part-final},
  and~\eqref{eq:1-rhon-W-log-1-rho-n-W-final-bound}, we arrive at~\eqref{eq:graphon-asymp-ent-lowerboun-claim}.

  Now, we show a matching upper bound
  for~\eqref{eq:graphon-asymp-ent-lowerboun-claim}, i.e.\ we show that
  \begin{equation}
    \label{eq:graphon-asym-upper-claim}
    \limsup_{n \rightarrow \infty} \frac{H(\Gn) - \barmn \log \frac{1}{\rho_n}}{\barmn} \leq 1 - \Ent(W).
  \end{equation}
  From Theorem~2.9 in~\cite{borgs2015consistent}, $W$ is equivalent to a
graphon over $[0,1]$ equipped with the uniform distribution. Therefore, without
loss of generality, we may assume that $W$ is a $L^2$ graphon over the space
$[0,1]$ equipped with uniform distribution, and $(X_i)_{i=1}^\infty$ is an
i.i.d.\ sequence of uniform $[0,1]$ random variables. Fix some $k \geq 1$ and
let $W_k$ 
be 
\black
a graphon over $[0,1]$ defined as follows. Let $Y_1, \dots, Y_k$
be a partition of $[0,1]$ into consecutive intervals of length $1/k$, and for $x
\in Y_i$ and $y \in Y_j$, define
\begin{equation}
  \label{eq:graphon-ent-asym-Wk-def}
  W_k(x,y) := \frac{1}{\left( \frac{1}{k} \right)^2} \int_{Y_i \times Y_j} W(u,v) du dv.
\end{equation}
In other words, $W_k$ is obtained from $W$ by taking the average in each cell
formed by the partition $(Y_i: 1 \leq i \leq k)$. Lemma~5.6 in~\cite{borgs2018L}
implies that since $W$ is $L^2$, we have
\begin{equation}
  \label{eq:graphon-ent-asymp-Wk-d2-0}
  \lim_{k \rightarrow \infty} \delta_2(W, W_k) = 0.
\end{equation}
Now, we fix $k \geq 1$ and find an upper bound for $H(\Gn)$. Define the random
variables $\tX_i, 1 \leq i \leq n$, as follows. For $1 \leq i \leq n$, 
if 
$X_i \in Y_j$
\black
we define $\tX_i := j$.
We may write
\begin{equation}
  \label{eq:graphon-ent-up-1}
  \begin{aligned}
    H(\Gn) &\leq H(\Gn, \tX_{[1:n]}) = H(\tX_{[1:n]}) + H(\Gn | \tX_{[1:n]}) \\
    &= n H(\tX_1) + H(\{\one{i \sim_{\Gn} j : 1 \leq i < j \leq n}\} | \tX_{[1:n]}) \\
    &\stackrel{(a)}{\leq} n \log k + \sum_{1 \leq i < j \leq n} H(\one{i \sim_{\Gn} j} | \tX_{[1:n]}) \\
    &= n \log k + \sum_{1 \leq i < j \leq n} H(\one{i \sim_{\Gn} j} | \tX_i, \tX_j) \\
    &\stackrel{(b)}{=} n \log k + \binom{n}{2} H(\one{1 \sim_{\Gn} 2} | \tX_1, \tX_2),
  \end{aligned}
\end{equation}
where $(a)$ uses the fact that $\tX_1$ is uniformly distributed over $\{1, \dots,
k\}$
and that a joint entropy is bounded above by the sum of the corresponding marginal entropies.
\black
Also, in $(b)$, we have used the symmetry
in $\Gn$. Note that conditioned on $\tX_1 = i$ and $\tX_2 = j$ for some $1 \leq
i, j \leq k$, $X_1$ and $X_2$ are independent and are distributed uniformly over
$Y_i$ and $Y_j$, respectively. Therefore,
\begin{align*}
  \pr{1 \sim_{\Gn} 2 | \tX_1 = i, \tX_2 = j} &= \frac{1}{\left( \frac{1}{k} \right)^2} \int_{Y_i \times Y_j} \rho_n W(u,v) \wedge 1 du dv \\
  &= \rho_n \frac{1}{\left( \frac{1}{k} \right)^2} \int_{Y_i \times Y_j} W(u, v) \wedge \frac{1}{\rho_n} du dv.
\end{align*}
For $1 \leq i,j \leq k$ and $n \geq 1$, define
\begin{equation}
  \label{eq:graphon-ent-asym-anij-def}
  \an_{i,j} := \frac{1}{\left( \frac{1}{k} \right)^2} \int_{Y_i \times Y_j} W(u,v) \wedge \frac{1}{\rho_n} du dv.
\end{equation}
Consequently, we have
\begin{equation}
  \label{eq:graphon-ent-asym-ent-1-2-tX12}
  \begin{aligned}
    H(\one{1 \sim_{\Gn} 2} | \tX_1, \tX_2) &= \sum_{i=1}^k \sum_{j=1}^k \pr{\tX_1 = i} \pr{\tX_2 = j} H_b\left(\pr{1 \sim_{\Gn}2 | \tX_1 = i, \tX_2 = j}\right) \\
    &= \left( \frac{1}{k} \right)^2 \sum_{i=1}^k \sum_{j=1}^k H_b(\rho_n \an_{i,j}) \\
    &= \left( \frac{1}{k} \right)^2 \sum_{i=1}^k \sum_{j=1}^k \rho_n \an_{i,j} \log \frac{1}{\rho_n \an_{i,j}} + (1 - \rho_n \an_{i,j}) \log \frac{1}{1 - \rho_n \an_{i,j}} \\
    &= \rho_n \log \frac{1}{\rho_n} \left(  \sum_{i=1}^k \sum_{j=1}^k \left( \frac{1}{k} \right)^2 \an_{i,j}\right) + \rho_n \left( \sum_{i=1}^k \sum_{j=1}^k \left( \frac{1}{k} \right)^2 \an_{i,j} \log \frac{1}{\an_{i,j}} \right) \\
    &\qquad + \sum_{i=1}^k \sum_{j=1}^k \left( \frac{1}{k} \right)^2 \left(1 - \rho_n \an_{i,j}\right) \log \frac{1}{1 - \rho_n \an_{i,j}}. 
  \end{aligned}
\end{equation}
We now simplify each of these three terms. 
For $1 \leq i \leq k$, let $y_i$ be
an arbitrary point in the interval 
$Y_i$.
With $\an_{i,j}$ as in~\eqref{eq:graphon-ent-asym-anij-def},
we have
\black
\begin{equation}
  \label{eq:anij-Wkyiyj-bound}
  \an_{i,j} =  \frac{1}{\left( \frac{1}{k} \right)^2} \int_{Y_i \times Y_j} W(u,v) \wedge \frac{1}{\rho_n} du dv\leq \frac{1}{\left( \frac{1}{k} \right)^2} \int_{Y_i \times Y_j} W(u,v) du dv = W_k(y_i, y_j).
\end{equation}
Thereby,
\begin{equation}
  \label{eq:sum-anij-bounded-by-1}
  \begin{aligned}
    \sum_{i=1}^k \sum_{j=1}^k \left( \frac{1}{k} \right)^2 \an_{i,j} &\leq     \sum_{i=1}^k \sum_{j=1}^k \left( \frac{1}{k} \right)^2 W_k(y_i, y_j) \\
    &=     \sum_{i=1}^k \sum_{j=1}^k \left( \frac{1}{k} \right)^2 \frac{1}{\left( \frac{1}{k} \right)^2} \int_{Y_i \times Y_j} W(u,v) du dv \\
    &= \int_0^1 \int_0^1 W(u,v) du dv \\
    &= 1,
  \end{aligned}
\end{equation}
where in the last step, we have used the assumption that $W$ is a normalized
graphon. Consequently, when $n$ is 
so large
\black
that $\rho_n < 1$, we have
\begin{equation}
  \label{eq:rhon-log-rhon-sum-anij-1}
  \rho_n \log \frac{1}{\rho_n} \left(     \sum_{i=1}^k \sum_{j=1}^k \left( \frac{1}{k} \right)^2 \an_{i,j} \right) \leq \rho_n \log \frac{1}{\rho_n}.
\end{equation}
From
\black
the definition of $\an_{i,j}$
in~\eqref{eq:graphon-ent-asym-anij-def}, since $1 / \rho_n \rightarrow \infty$,
for all $1 \leq i,j \leq k$, we have
\begin{equation*}
  \lim_{n \rightarrow \infty} \an_{i,j} = \frac{1}{\left( \frac{1}{k} \right)^2} \int_{Y_i \times Y_j} W(u,v) du dv = W_k(y_i, y_j).
\end{equation*}
Hence, since $\snorm{W_k}_1 = \snorm{W}_1 = 1$, we have
\begin{equation}
  \label{eq:lim-sum-anij-log-anij-Ent-Wk}
  \begin{aligned}
    \lim_{n \rightarrow \infty}     \sum_{i=1}^k \sum_{j=1}^k \left( \frac{1}{k} \right)^2 \an_{i,j} \log \frac{1}{\an_{i,j}} &=     \sum_{i=1}^k \sum_{j=1}^k \left( \frac{1}{k} \right)^2 W_k(y_i, y_j) \log \frac{1}{W_k(y_i, y_j)} \\
    &= \int_0^1 \int_0^1 W_k(u,v) \log \frac{1}{W_k(u,v)}du dv \\
    &= -\Ent(W_k).
  \end{aligned}
\end{equation}
On the other hand, using $(1-x) \log \frac{1}{1-x} = x + x \eta(x)$ where $\eta(x)
\rightarrow 0$ as $x \rightarrow 0$, as we discussed before in
proving~\eqref{eq:graphon-asymp-ent-lowerboun-claim}, we may write
\begin{align*}
  \sum_{i=1}^k \sum_{j=1}^k \left( \frac{1}{k} \right)^2 \left( 1 - \rho_n \an_{i,j} \right) \log \frac{1}{1 - \rho_n \an_{i,j}} = \rho_n     \sum_{i=1}^k \sum_{j=1}^k \left( \frac{1}{k} \right)^2 \an_{i,j}(1 + \eta(\rho_n \an_{i,j})).
\end{align*}
From~\eqref{eq:anij-Wkyiyj-bound}, $\an_{i,j} \leq W_k(y_i, y_j) \leq \max_{i,j}
W_k(y_i, y_j)$, and 
we have
\black
$\rho_n \rightarrow 0$ as $n \rightarrow \infty$. Thereby,
there exists a sequence $\epsilon_n \rightarrow 0$ such that $\eta(\rho_n
\an_{i,j}) \leq \epsilon_n$ for $1 \leq i , j \leq k$. Using this together
with~\eqref{eq:sum-anij-bounded-by-1} in the above, we get
\begin{equation}
  \label{eq:sum-1-rhon-anij-log-rhon-epsilon}
  \sum_{i=1}^k \sum_{j=1}^k \left( \frac{1}{k} \right)^2 \left( 1 - \rho_n \an_{i,j} \right) \log \frac{1}{1 - \rho_n \an_{i,j}} \leq \rho_n(1 + \epsilon_n).
\end{equation}
Substituting~\eqref{eq:rhon-log-rhon-sum-anij-1},
\eqref{eq:lim-sum-anij-log-anij-Ent-Wk},
and~\eqref{eq:sum-1-rhon-anij-log-rhon-epsilon} back
into~\eqref{eq:graphon-ent-asym-ent-1-2-tX12}, we realize that
\begin{equation*}
  \limsup_{n \rightarrow \infty} \frac{H(\one{1 \sim_{\Gn}2} | \tX_1, \tX_2) - \rho_n \log \frac{1}{\rho_n}}{\rho_n} \leq 1 - \Ent(W_k).
\end{equation*}
Multiplying the numerator and the denominator on the left hand side by
$\binom{n}{2}$ and recalling $\barmn = \binom{n}{2} \rho_n$, we get
\begin{equation*}
  \limsup_{n \rightarrow \infty} \frac{\binom{n}{2} H(\one{1 \sim_{\Gn} 2} | \tX_1, \tX_2) - \barmn \log \frac{1}{\rho_n}}{\barmn} \leq 1 - \Ent(W_k).
\end{equation*}
Using this in~\eqref{eq:graphon-ent-up-1}, we get
\begin{align*}
  \limsup_{n \rightarrow \infty} \frac{H(\Gn) - \barmn \log \frac{1}{\rho_n}}{\barmn} &\leq 1 + \Ent(W_k) + \limsup_{n \rightarrow \infty} \frac{n \log k}{\barmn} \\
  &= 1 - \Ent(W_k),
\end{align*}
where the last line uses the fact that since $\barmn = \binom{n}{2} \rho_n$ and
$n \rho_n \rightarrow \infty$, $n / \barmn \rightarrow 0$. Note that this bound
holds for all $k \geq 1$. Moreover, from~\eqref{eq:graphon-ent-asymp-Wk-d2-0},
$\delta_2(W_k, W) \rightarrow 0$ as $k \rightarrow \infty$. Therefore,
part~\ref{thm-ge-L2-conv} in Theorem~\ref{thm:graphon-entropy-props}
implies that $\Ent(W_k)
\rightarrow \Ent(W)$. Hence, we arrive at~\eqref{eq:graphon-asym-upper-claim} by
sending $k$ to infinity in the above bound. The proof is complete by
putting~\eqref{eq:graphon-asymp-ent-lowerboun-claim} and
\eqref{eq:graphon-asym-upper-claim} together.
\end{proof}

\section{Proof of Proposition~\ref{prop:no-good-splitting}}
\label{app:no-good-splitting-proof}

We assume  that  $((\Tn_1, \Tn_2): n \geq 1)$ is a sequence of good splitting
mechanisms, and we arrive at a contradiction.

For $n \geq 1$ and $1 \leq k \leq n$, let $\Gn_k$ be an \ER random graph on $n$
vertices 
where each edge is 
\black
independently present 
with
probability
$k/n$.
\black
We can assume that $(\Gn_k: n \geq 1, k \leq n)$
live independently on a joint probability space. From this point forward, all
of our probabilistic statements will refer to this joint probability space. 

We know that with $\mu_k \in \mP(\mT_*)$ being the law of the unimodular
Galton--Watson tree with Poisson degree distribution and average degree $k$
we have
\black
$U(\Gn_k) \Rightarrow \mu_k$ a.s.
for each fixed $k \ge 1$ as
$n \rightarrow \infty$. 
\black
Therefore, the assumption that $((\Tn_1,
\Tn_2): n \geq 1)$ is a sequence of good splitting mechanisms implies that with $\Gn_{k,1} := \Tn_1(\Gn_k)$,
we have $U(\Gn_{k,1}) \Rightarrow \mu_k$ a.s.. Let $\mn_{k,1}$
denote the number of edges in $\Gn_{k,1}$. Moreover, let $\deg : \mG_*
\rightarrow \reals_+$ be such that $\deg([G,o]) = \deg_G(o)$ and note that for
any $\alpha > 0$ the function $(\deg \wedge \alpha): \mG_* \rightarrow \reals_+$,
defined as $(\deg \wedge \alpha) ([G,o]) := \deg([G,o]) \wedge \alpha$,  is bounded
and continuous. Thereby, the fact that $U(\Gn_{k,1}) \Rightarrow \mu_k$ a.s.\
implies that
\begin{equation}
  \label{eq:lim-deg-alpha}
  \lim_{n \rightarrow \infty} \int (\deg \wedge \alpha) d U(\Gn_{k,1}) = \int (\deg \wedge \alpha) d \mu_k = \int \deg_G(o) \wedge \alpha d \mu_k ([G,o]) \qquad \text{a.s.}.
\end{equation}
On the other hand, we have
\begin{equation*}
  \frac{\mn_{k,1}}{n} = \frac{1}{2} \int \deg dU(\Gn_{k,1})  \geq  \frac{1}{2} \int (\deg \wedge \alpha) dU(\Gn_{k,1}).
\end{equation*}
This together with \eqref{eq:lim-deg-alpha} implies that
\begin{equation*}
  \liminf_{n \rightarrow \infty} \frac{\mn_{k,1}}{n} \geq \frac{1}{2} \int \deg_G(o) \wedge \alpha d \mu_k([G,o]) \qquad \text{a.s.}.
\end{equation*}
Since this holds for all $\alpha > 0$, sending $\alpha$ to infinity,
we realize
that
\begin{equation}
  \label{eq:er-n-k-mn-k-1-liminf-k-2-as}
  \liminf_{n \rightarrow \infty} \frac{\mn_{k,1}}{n} \geq \frac{k}{2} \qquad \text{a.s.}.
\end{equation}
This means that for all $\epsilon > 0$, we have
\begin{equation*}
  \lim_{n \rightarrow \infty} \pr{\frac{\mn_{k,1}}{n} < \frac{k}{2} - \epsilon} = 0.
\end{equation*}
In particular, we have 
\begin{equation}
  \label{eq:lim-pr-mn-k1-k-4-zero}
  \lim_{n \rightarrow \infty} \pr{\frac{\mn_{k,1}}{n} < \frac{k}{4}} = 0 \qquad \forall k \geq 1.
\end{equation}

We now
\black
define a sequence of integers $(n_k: k\geq 1)$ inductively as follows.
Let 
$n_0 := 0$
\black
and, for $k \geq 1$, assuming that $n_{k-1}$ is chosen, we
choose $n_k$ large enough such that the following three conditions are satisfied:
\begin{enumerate}
\item $n_k > n_{k-1}$;
\item
  \begin{equation}
    \label{eq:nk-cond-mn-nk-k4-k2}
    \pr{\frac{m^{(n_k)}_{k,1}}{n_k} < \frac{k}{4} } < \frac{1}{k^2};
  \end{equation}
\item
  \begin{equation}
    \label{eq:nk-cond-k-nk-1-k}
    \frac{k+1}{n_k} < \frac{1}{k+1}.
  \end{equation}
\end{enumerate}

Note that condition~\eqref{eq:nk-cond-mn-nk-k4-k2} 
can be
\black
satisfied, due to~\eqref{eq:lim-pr-mn-k1-k-4-zero}.
We next
\black
define the sequence $(\Gn: n \geq 1)$ of random graphs as follows. For
$n \geq 1$, let $k(n)$  be the unique integer  $k \geq 1$ such that $n_{k-1} < n \leq n_k$ and let $\Gn
= \Gn_{k(n)}$.  Note that since $n_{k(n) - 1} < n$,
using~\eqref{eq:nk-cond-k-nk-1-k} we have
\begin{equation}
  \label{eq:kn-n-less-1-kn}
  \frac{k(n)}{n} < \frac{k(n)}{n_{k(n) - 1}} < \frac{1}{k(n)}.
\end{equation}
In particular this means that $k(n) < n$ and
so
\black
the sequence $\Gn$ is well defined. 

Observe that this sequence $\Gn$ can be represented in terms of a sequence of
$W$--random graphs for the graphon $W$ defined on the probability space $[0,1]$
equipped with the uniform distribution such that $W(x,y) = 1$ for all $x, y \in
[0,1]$. To see this, 
let $\rho_n
:= k(n) / n$
for $n \ge 1$.
\black
The
\black
distribution of the sequence $\Gn = \Gn_{k(n)}$ is 
then
\black
identical to the
distribution of the sequence of $W$--random graphs 
with target densities
$\rho_n$.
\black
Further, 
\black
due to~\eqref{eq:kn-n-less-1-kn}, we have
\begin{equation}
  \label{eq:no-good-split-rhon-goes-zero}
  \lim_{n \rightarrow \infty} \rho_n = 0.
\end{equation}
and
\begin{equation}
  \label{eq:no-good-split-lim-n-rhon-infty}
  \lim_{n \rightarrow \infty} n \rho_n = \infty,
\end{equation}
because 
$k(n) \rightarrow
\infty$ as $n \rightarrow \infty$.
\black
As a result, the assumption that $((\Tn_1, \Tn_2): n \geq 1)$ is a sequence of
good splitting mechanisms ensures that
\begin{equation}
  \label{eq:no-good-split-ensure-mn-mbarn-small}
  \lim_{n \rightarrow \infty} \frac{\mn_{k(n),1}}{\barmn} = 0 \qquad \text{a.s.},
\end{equation}
where $\barmn := {n \choose 2} \rho_n$.
But
\black
note that by definition we have $k(n_k) = k$ and thereby $\rho_{n_k} = k(n_k)
/ n_k = k / n_k$. Hence
\begin{equation*}
\barm_{n_k} = \binom{n_k}{2} \rho_{n_k} = \frac{n_k(n_k - 1)}{2} \frac{k}{n_k} = \frac{(n_k - 1)k}{2}.
\end{equation*}
Therefore, using~\eqref{eq:nk-cond-mn-nk-k4-k2}, we have
\begin{align*}
  \pr{\frac{m^{(n_k)}_{k,1}}{\barm_{n_k}} < \frac{1}{2}} &= \pr{\frac{2 m^{(n_k)}_{k,1}}{k(n_k-1)} < \frac{1}{2}} \\
&\leq \pr{\frac{2 m^{(n_k)}_{k,1}}{kn_k} < \frac{1}{2}} \\
&= \pr{\frac{m^{(n_k)}_{k,1}}{n_k} < \frac{k}{4}} \\
&< \frac{1}{k^2}.
\end{align*}
\black
Consequently, using the Borel--Cantelli lemma, we have
\begin{equation*}
  \pr{\frac{m^{(n_k)}_{k,1}}{\barm_{n_k}} < \frac{1}{2} \text{ for infinitely many } k}  = 0.
\end{equation*}
\black
Therefore,
\begin{equation*}
  \liminf_{k \rightarrow \infty} \frac{m^{(n_k)}_{k,1}}{\barm_{n_k}} \geq \frac{1}{2} > 0 \qquad \text{a.s.}.
\end{equation*}
\black
Recall that $m^{(n_k)}_{k,1}$ is the number of edges in
$T^{(n_k)}_1(G^{(n_k)}_k) = T^{(n_k)}_1(G^{(n_k)})$. Since $n_k \rightarrow \infty$ as $k \rightarrow
\infty$, this in particular means that
\begin{equation*}
  \liminf_{n \rightarrow \infty} \frac{\mn_{k(n), 1}}{\barmn}  \geq \frac{1}{2} > 0 \qquad \text{a.s.}.
\end{equation*}
\black
But this is in contradiction
with~\eqref{eq:no-good-split-ensure-mn-mbarn-small}. Therefore no sequence of
good splitting mechanisms exists and the proof is complete.

\section{Proof of Proposition~\ref{prop:graphon-mn-asymptotics}}
\label{app:proof-prop-graphon-mn-asymptotics}

Throughout this section, 
we assume that 
\black
$W$ is a normalized $L^2$ graphon and $\Gn \sim \mG(n; \rho_n W)$ is a sequence
  of $W$--random graphs with target density $\rho_n$ such that 
  $\rho_n \rightarrow 0$ and
  $n \rho_n
  \rightarrow \infty$.
  \black
  Also,  $\mn$ denotes the
  number of edges in $\Gn$ and $\barmn :=
  \binom{n}{2} \rho_n$. For better organization, we prove
  Proposition~\ref{prop:graphon-mn-asymptotics} in separate lemmas. 

\begin{lem}
\label{lem:graphon-mn-asymp-mn-mbarn-1}    
We have $\mn / \barmn \rightarrow 1$ a.s..
\end{lem}

\begin{proof}
  From Theorem~\ref{thm:ganguli-W-random-graphon-convergence}, we know that
  $\rho(\Gn) / \rho_n \rightarrow 1$ a.s., where $\rho(\Gn) = 2 \mn / n^2$.
  Comparing this with definition  $\mn = \binom{n}{2} \rho_n$, we realize that
  $\mn / \barmn \rightarrow 1$ a.s..
\end{proof}

\begin{lem}
  \label{lem:mn-concentrates-around-barmn}
We have 
  \begin{equation*}
    \lim_{n \rightarrow \infty} \frac{\mn - \barmn}{\barmn} \log \frac{1}{\rho_n} =  0 \qquad \text{a.s.}
  \end{equation*}
\end{lem}
  
\begin{proof}[Proof of Lemma~\ref{lem:mn-concentrates-around-barmn}]
  We pick
  \black
  $(X_i)_{i=1}^\infty$ i.i.d.\ from $\Omega$ and generate $\Gn$ based on
  $X_{[1:n]}$.
  From Theorem~2.9 in~\cite{borgs2015consistent}, $W$ is equivalent to a
graphon over $[0,1]$ equipped with the uniform distribution. Therefore, without
loss of generality, we may assume that $W$ is a $L^2$ graphon over $[0,1]$, and $(X_i)_{i=1}^\infty$ is an i.i.d.\
sequence of random variables distributed uniformly over $[0,1]$.

  Define the random variable
  $\barMn$ as
  \begin{equation*}
    \barMn = \barMn (X_{[1:n]}) := \sum_{1 \leq i < j \leq n} (\rho_n W(X_i,X_j)) \wedge 1.
  \end{equation*}
  Note that we have
  \begin{equation*}
    \barMn = \ev{\mn|X_{[1:n]}}.
  \end{equation*}
  With this definition, we prove the lemma in two steps, namely
  \begin{equation}
    \label{eq:limsup-barMn-barmn-claim}
    \lim_{n \rightarrow \infty} \frac{\barMn - \barmn}{\barmn} \log \frac{1}{\rho_n} = 0 \qquad \text{a.s.},
  \end{equation}
  and
  \begin{equation}
    \label{eq:limsup-mn-barMn-claim}
    \lim_{n \rightarrow \infty} \frac{\mn - \barMn}{\barmn} \log \frac{1}{\rho_n} = 0 \qquad \text{a.s.},
  \end{equation}
  which together complete the proof.

  We first prove~\eqref{eq:limsup-barMn-barmn-claim}. With $\tau_n :=
  \rho_n^{-1/8}$, define
\begin{subequations}
  \begin{align}
    Y_n  = Y_n(X_{[1:n]}) &:= \sum_{1 \leq i < j \leq n} \one{W(X_i,X_j) \leq \tau_n} \rho_n W(X_i,X_j), \label{eq:Yn-def}\\
    Z_n = Z_n(X_{[1:n]})&:=  \sum_{1 \leq i < j \leq n} \one{W(X_i,X_j) > \tau_n} \rho_n W(X_i,X_j). \label{eq:Zn-def}
  \end{align}
  \end{subequations}
 Note that
  \begin{equation}
    \label{eq:barMn-Yn+Zn}
    \barMn \leq \sum_{1 \leq i < j \leq n} \rho_n W(X_i, X_j) = Y_n + Z_n.
  \end{equation}
On the other hand, 
if 
\black
for some $1 \leq i < j \leq n$, we have
$W(X_i,X_j) \leq \tau_n$, then $\rho_n W(X_i, X_j) \leq \rho_n \tau_n =
\rho_n^{7/8}$. But $\rho_n \rightarrow 0$ as $n \rightarrow \infty$. Hence, for
$n$ large enough, we have $\rho_n^{7/8} < 1$. This means that for $n$ large
enough we have
\begin{equation}
  \label{eq:n-large-Yn<Mbarn}
  Y_n = \sum_{1 \leq i < j \leq n} \one{W(X_i, X_j) \leq \tau_n} ((\rho_n W(X_i, X_j))\wedge 1) \leq \barMn.
\end{equation}
Putting this together with~\eqref{eq:barMn-Yn+Zn}, we realize that for $n$ large
enough we have
\begin{equation}
  \label{eq:barMn-Yn-Yn-Zn-sandwich}
  Y_n \leq \barMn \leq Y_n + Z_n.
\end{equation}
Now, we claim that
\begin{equation}
  \label{eq:lim-Yn-mn-log-rhon-0-as-claim}
  \lim_{n \rightarrow \infty} \frac{Y_n - \barmn}{\barmn} \log \frac{1}{\rho_n} = 0 \qquad \text{a.s.},
\end{equation}
and
\begin{equation}
  \label{eq:lim-Zn-log-rhon-0-as-claim}
  \lim_{n \rightarrow \infty} \frac{Z_n}{\barmn} \log \frac{1}{\rho_n} = 0 \qquad \text{a.s.}.
\end{equation}
Note that~\eqref{eq:limsup-barMn-barmn-claim} follows
from~\eqref{eq:barMn-Yn-Yn-Zn-sandwich},
\eqref{eq:lim-Yn-mn-log-rhon-0-as-claim}, and \eqref{eq:lim-Zn-log-rhon-0-as-claim}.

We start with showing~\eqref{eq:lim-Yn-mn-log-rhon-0-as-claim}.
Observe that for $1 \leq i \leq n$, $x_1, \dots, x_n \in [0,1]$, and $x'_i \in
[0,1]$, since $W$ is a symmetric function,  we have
\begin{align*}
  |Y_n(x_1, \dots, x_i, \dots, x_n) - Y_n(x_1, \dots, x'_i, \dots, x_n)| &= \left | \sum_{\stackrel{1 \leq j \leq n}{j \neq i}} \one{W(x_i, x_j) \leq \tau_n} \rho_n W(x_i, x_j) - \one{W(x'_i, x_j) \leq \tau_n} \rho_n W(x'_i, x_j) \right| \\
                                                                         &\leq \rho_n \sum_{\stackrel{1 \leq j \leq n}{j \neq i}} \Big| \one{W(x_i, x_j) \leq \tau_n} W(x_i, x_j) - \one{W(x'_i, x_j) \leq \tau_n} W(x'_i, x_j) \Big | \\
  &\leq n \rho_n \tau_n.
\end{align*}
Therefore, using the bounded difference inequality 
(see, for
  instance, \cite[Theorem 6.2]{boucheron2013concentration})
we have
\begin{align*}
  \pr{|Y_n - \ev{Y_n}| > n^{13/8} \rho_n^{7/8}} &\leq 2\exp\left( - 2 \frac{n^{13/4} \rho_n^{7/4}}{n(n\rho_n \tau_n)^2} \right) \\
                                              &= 2\exp(-2n^{1/4}).
\end{align*}
Since 
$\sum_{n \ge 1} \exp(-2n^{1/4}) < \infty$, the Borel-Cantelli lemma implies that 
\black
with
probability one, for $n$ large enough (where the threshold of $n$ can be
random), we have
\begin{equation}
  \label{eq:Yn-EYn-bound}
  |Y_n -\ev{Y_n}| \leq  n^{13/8} \rho_n^{7/8}.
\end{equation}
This means that with probability one we have
\begin{align*}
  \limsup_{n \rightarrow \infty} \frac{|Y_n - \ev{Y_n}|}{\barmn} \log \frac{1}{\rho_n} &\leq \lim_{n \rightarrow \infty} \frac{n^{13/8} \rho_n^{7/8}}{\binom{n}{2} \rho_n} \log \frac{1}{\rho_n} \\
&= \lim_{n \rightarrow \infty}\frac{2n}{(n-1)} n^{-3/8} \rho_n^{-1/8} \log \frac{1}{\rho_n} \\
&= \lim_{n \rightarrow \infty} \frac{2n}{n-1} (n\rho_n)^{-3/8} \rho_n^{1/4} \log \frac{1}{\rho_n} \\
&= 0,
\end{align*}
where the last equality follows from the facts that as $n \rightarrow \infty$, we
have  $\rho_n \rightarrow 0$ and $n\rho_n \rightarrow \infty$. Consequently, we have 
\begin{equation}
  \label{eq:lim-Yn-EYn-log-rhon-zero}
  \lim_{n \rightarrow \infty} \frac{Y_n - \ev{Y_n}}{\barmn} \log \frac{1}{\rho_n} = 0 \qquad \text{a.s.}.
\end{equation}

We turn to studying $\ev{Y_n}$.
\black
Recalling the definition of $Y_n$ from~\eqref{eq:Yn-def}, 
and writing 
\black
$W$ for $W(X, \pp{X})$ with $X$
and $\pp{X}$ being i.i.d.\ on $\Omega$ with distribution $\pi$, we may write
\begin{align*}
  \ev{Y_n} &= \binom{n}{2} \ev{W \one{W \leq \rho_n^{-1/8}}} \rho_n \\
  &= \barmn \ev{W \one{W \leq \rho_n^{-1/8}}}.
\end{align*}
Therefore,
\begin{equation}
  \label{eq:EYn-mbarn-log-rhon-simplifty}
\begin{aligned}
  \frac{\ev{Y_n} - \barmn}{\barmn} \log \frac{1}{\rho_n} &= \left( \ev{W \one{W \leq \rho_n^{-1/8}}} - 1 \right) \log \frac{1}{\rho_n} \\
  &= - \ev{W \one{W > \rho_n^{-1/8}}} \log \frac{1}{\rho_n},
\end{aligned}
\end{equation}
where the second line employs the fact that $W$ is a normalized graphon and
hence $\ev{W} = 1$. 
Using the Cauchy-Schwartz inequality, we may write
\begin{align*}
  \ev{W \one{W > \rho_n^{-1/8}}} &\leq \sqrt{\ev{W^2}} \sqrt{\ev{\one{W > \rho_n^{-1/8}}^2}} \\
                                 &= \sqrt{\ev{W^2}} \sqrt{\ev{\one{W > \rho_n^{-1/8}}}} \\
                                 &= \snorm{W}_2 \sqrt{\pr{W > \rho_n^{-1/8}}} \\
                                 &\leq \snorm{W}_2 \sqrt{\ev{W} / \rho_n^{-1/8}} \\
                                 &= \snorm{W}_2 \rho_n^{1/16},
\end{align*}
where the last step uses the fact that $\ev{W} = 1$. Note that by assumption $W$
is an $L^2$ graphon and 
so
\black
$\snorm{W}_2 < \infty$. Also, by assumption, $\rho_n
\rightarrow 0$ as $n \rightarrow \infty$. This together
with~\eqref{eq:EYn-mbarn-log-rhon-simplifty} implies that
\begin{equation}
  \label{eq:EYn-mbarn-log-rhon-0}
  \lim_{n \rightarrow \infty} \frac{\ev{Y_n} - \barmn}{\barmn} \log \frac{1}{\rho_n} = 0.
\end{equation}
Putting~\eqref{eq:lim-Yn-EYn-log-rhon-zero} and~\eqref{eq:EYn-mbarn-log-rhon-0}
together, we arrive at~\eqref{eq:lim-Yn-mn-log-rhon-0-as-claim}.


We next
\black
focus on showing~\eqref{eq:lim-Zn-log-rhon-0-as-claim}. Note that we
have 
\begin{align*}
  Z_n &= \sum_{1 \leq i < j \leq n} \one{\tau_n < W(X_i, X_j)} \rho_n W(X_i, X_j) \\
      &= \sum_{1 \leq i < j \leq n} \one{1 < \tau_n^{-1} W(X_i, X_j)} \rho_n W(X_i, X_j) \\
      &\leq \sum_{1 \leq i < j \leq n} \tau_n^{-1} W(X_i, X_j) \rho_n W(X_i,X_j) \\
      &= \rho_n \tau_n^{-1} \sum_{1 \leq i < j \leq n} W^2(X_i, X_j).
\end{align*}
Since $\rho_n \rightarrow 0$ as $n \rightarrow \infty$, we have $\rho_n < 1$ when $n$ is large enough.
Consequently, for $n$ large enough, we have
\black
\begin{equation}
  \label{eq:Zn-bound-1}
  \begin{aligned}
    \frac{1}{\barmn} Z_n \log \frac{1}{\rho_n} &\leq \frac{\rho_n \tau_n^{-1}}{\binom{n}{2} \rho_n} \left( \sum_{1 \leq i < j \leq n} W^2(X_i, X_j) \right) \log \frac{1}{\rho_n} \\
    &= \left[ \frac{1}{\binom{n}{2}} \sum_{1 \leq i < j \leq n} W^2(X_i, X_j) \right] \rho_n^{1/8} \log \frac{1}{\rho_n}.
  \end{aligned}
\end{equation}
Note that $\rho_n^{1/8} \log 1 / \rho_n \rightarrow 0$ as $\rho_n \rightarrow
0$. Additionally, since $W$ is a $L^2$ graphon and $X_i$ are i.i.d.\ and
uniformly distributed over $[0,1]$, the strong law of large numbers for U--statistics
(see \cite{hoeffding1961strong}) implies that with probability one,
\begin{equation*}
  \lim_{n \rightarrow \infty} \frac{1}{\binom{n}{2}} \sum_{1 \leq i < j \leq n} W^2(X_i, X_j) = \int_0^1 \int_0^1 W^2(x,y) dx dy = \snorm{W}_2^2 < \infty.
\end{equation*}
Substituting into~\eqref{eq:Zn-bound-1}, we arrive
at~\eqref{eq:lim-Zn-log-rhon-0-as-claim}. As we discussed earlier,
\eqref{eq:barMn-Yn-Yn-Zn-sandwich} together
with~\eqref{eq:lim-Yn-mn-log-rhon-0-as-claim}
and~\eqref{eq:lim-Zn-log-rhon-0-as-claim} imply \eqref{eq:limsup-barMn-barmn-claim}.

  Next, we focus on showing~\eqref{eq:limsup-mn-barMn-claim}. 
Recall that 
  $\ev{\mn|X_{[1:n]}} = \barMn$. Also,  conditioned on $X_{[1:n]}$, the edges in $\Gn$
  are placed independently from each other.
  Define
  \begin{equation*}
    \delta_n :=
    \begin{cases}
      n^{-1/8} \rho_n^{1/8} \sqrt{\frac{\barmn}{\barMn}} & \text{ if } \frac{\barMn}{\barmn} > n^{-1/4} \rho_n^{1/4}, \\
      n^{-1/4} \rho_n^{1/4} \frac{\barmn}{\barMn} & \text{ if } \frac{\barMn}{\barmn} \leq n^{-1/4} \rho_n^{1/4} \text{ and } \barMn > 0, \\
      0 &  \text{ if } \barMn = 0.
    \end{cases}
  \end{equation*}
  With this, using the Chernoff bound, we have
  \begin{equation*}
    \pr{\mn > (1+\delta_n) \barMn | X_{[1:n]}} \leq \exp \left( - \frac{\delta_n^2 \barMn}{2 + \delta_n} \right).
  \end{equation*}
  Now, we 
  consider the three cases in the definition of $\delta_n$ in turn:
  \black

  \underline{Case 1:} If $\barMn / \barmn > n^{-1/4} \rho_n^{1/4}$, we have
  \begin{equation*}
    \delta_n = \frac{n^{-1/8} \rho_n^{1/8}}{\sqrt{\barMn / \barmn}} \leq 1.
  \end{equation*}
  Therefore, $2 + \delta_n \leq 3$ and
  \begin{equation*}
    \frac{\delta_n^2 \barMn}{2 + \delta_n} \geq \frac{1}{3} \delta_n^2 \barMn = \frac{1}{3} n^{-1/4} \rho_n^{1/4} \barmn.
  \end{equation*}

  \underline{Case 2:} If $\barMn / \barmn \leq n^{-1/4} \rho_n^{1/4}$ and
  $\barMn > 0$, we have
  \begin{equation*}
    \delta_n = \frac{n^{-1/4} \rho_n^{1/4}}{\barMn / \barmn} \geq 1.
  \end{equation*}
Thereby, using the inequality $x^2 / (2+x) \geq x/3$ which holds for $x \geq 1$,
we have
\begin{equation*}
  \frac{\delta_n^2 \barMn}{2 + \delta_n} \geq \frac{1}{3} \delta_n \barMn = \frac{1}{3} n^{-1/4} \rho_n^{1/4} \barmn.
\end{equation*}

\underline{Case 3:} If $\barMn = 0$, recalling the definition of $\barMn$, we
realize that $W(X_i, X_j) = 0$ for all $1 \leq i, j \leq n$. Thereby, there is
no edge in $\Gn$ and $\mn = 0$. In this case, automatically we have $\pr{\mn
  > (1+\delta_n) \barMn | X_{[1:n]}} = 0$.

Combining the above three cases, we realize that in general we have
\begin{equation*}
  \pr{\mn > (1+\delta_n) \barMn | X_{[1:n]}} \leq \exp \left( - \frac{1}{3} n^{-1/4} \rho_n^{1/4} \barmn \right).
\end{equation*}
Since the bound does not depend on $X_{[1:n]}$, we have
\begin{equation*}
  \pr{\mn > (1+\delta_n) \barMn} \leq \exp\left( - \frac{1}{3} n^{-1/4} \rho_n^{1/4} \barmn \right).
\end{equation*}
Recalling the definition of $\barmn$, we have
\begin{equation*}
  n^{-1/4} \rho_n^{1/4} \barmn = \frac{1}{2} \frac{n-1}{n} n^{-1/4} \rho_n^{1/4} n^2 \rho_n = \frac{1}{2} \frac{n-1}{n} n^{1/2} (n\rho_n)^{5/4}.
\end{equation*}
Since $n \rho_n \rightarrow \infty$ as $n \rightarrow \infty$, we have
\begin{equation*}
  \sum_{n} \exp\left( - \frac{1}{3} n^{-1/4} \rho_n^{1/4} \barmn \right) < \infty.
\end{equation*}
Hence, 
the
\black
Borel-Cantelli lemma implies that with probability one, for $n$ large
enough (where the threshold itself can be random), we have $\mn \leq
(1+\delta_n) \barMn$. Consequently, 
since $\rho_n < 1$ when $n$ is large enough,
we have with probability one that
\black
\begin{equation}
  \label{eq:mn-barMn-delta-n-limsup-bound}
  \limsup \frac{\mn - \barMn}{\barmn} \log \frac{1}{\rho_n} \leq \limsup \frac{\barMn}{\barmn} \delta_n \log \frac{1}{\rho_n}.
\end{equation}
If $\barMn / \barmn > n^{-1/4} \rho_n^{1/4}$, we have
\begin{equation*}
  \frac{\barMn}{\barmn} \delta_n \log \frac{1}{\rho_n} = \sqrt{\frac{\barMn}{\barmn}} n^{-1/8} \rho_n^{1/8} \log \frac{1}{\rho_n}.
\end{equation*}
On the other hand, if $\barMn / \barmn \leq n^{-1/4} \rho_n^{1/4}$ and $\barMn
> 0$, then
\begin{equation*}
  \frac{\barMn}{\barmn} \delta_n \log \frac{1}{\rho_n}  = n^{-1/4} \rho_n^{1/4} \log \frac{1}{\rho_n}.
\end{equation*}
Furthermore, if $\barMn = 0$, then $  \frac{\barMn}{\barmn} \delta_n \log
\frac{1}{\rho_n}  = 0$. Combining these cases and comparing
with~\eqref{eq:mn-barMn-delta-n-limsup-bound}, we realize that with probability one,
\begin{equation}
  \label{eq:mn-barMn-limsup-bound-2}
  \limsup \frac{\mn - \barMn}{\barmn} \log \frac{1}{\rho_n} \leq \limsup \left( 1 + \sqrt{\frac{\barMn}{\barmn}} \right) n^{-1/8} \rho_n^{1/8} \log \frac{1}{\rho_n},
\end{equation}
because 
$n^{-1/4} \rho_n^{1/4} \le n^{-1/8} \rho_n^{1/8}$
when $n$ is large enough.
\black
Observe that we have
\begin{equation}
  \label{eq:barMn-wedge-drop-bound}
  \barMn \leq \sum_{1 \leq i < j \leq n} \rho_n W(X_i,X_j).
\end{equation}
Recall that $(X_i)_{i=1}^\infty$ is sequence of i.i.d.\ random variables
uniformly distributed over $[0,1]$.
Hence, using the strong law of large numbers for U-statistics
\cite{hoeffding1961strong}, we have
\begin{equation*}
  \lim_{n \rightarrow \infty} \frac{1}{\binom{n}{2}} \sum_{1 \leq i < j \leq n} W(X_i,X_j) = \int_0^1 \int_0^1 W(x,y) dx dy = 1.
\end{equation*}
This together with~\eqref{eq:barMn-wedge-drop-bound} and $\barmn = n(n-1) \rho_n /2$
leads to
\black
\begin{equation}
  \label{eq:limsup-Mn-mn-1-as}
  \limsup \frac{\barMn}{\barmn} \leq 1 \qquad \text{a.s.}
\end{equation}
Substituting this into~\eqref{eq:mn-barMn-limsup-bound-2} and noting that
$\rho_n^{1/8} \log 1/ \rho_n \rightarrow 0$
as $n \rightarrow \infty$,
\black
we realize that
\begin{equation}
  \label{eq:limsup-mn-barMn-log-rhon}
  \limsup_{n \rightarrow \infty} \frac{\mn - \barMn}{\barmn} \log \frac{1}{\rho_n} \leq 0 \qquad \text{a.s.}.
\end{equation}

In order to obtain a matching lower bound, let $\tilde{\delta}_n := n^{-1/8}
\rho_n^{1/8} \sqrt{\barmn / \barMn}$ and note that another usage of the
Chernoff bound implies that conditioned on $X_{[1:n]}$, if $\tilde{\delta}_n
< 1$ and $\barMn > 0$, we have
\begin{align*}
  \pr{\mn < (1 - \tilde{\delta}_n) \barMn | X_{[1:n]}} &\leq \exp \left(  - \frac{\tilde{\delta}_n^2 \barMn}{2} \right) \\
                                                       &= \exp\left( -\frac{1}{2} n^{-1/4} \rho_n^{1/4} \frac{\barmn}{\barMn} \barMn \right)\\
                                                       &= \exp\left( -\frac{1}{2} n^{-1/4}\rho_n^{1/4} \barmn \right).
\end{align*}
Note that if either $\tilde{\delta}_n \geq 1$ or $\barMn = 0$, the left
hand side becomes  zero and this bound automatically holds. Furthermore, since 
this
\black
upper bound does not depend on $X_{[1:n]}$, we conclude that
\begin{equation*}
  \pr{\mn < (1-\tilde{\delta}_n) \barMn} \leq \exp \left( -\frac{1}{2} n^{-1/4}\rho_n^{1/4} \barmn \right).
\end{equation*}
Thereby, another usage of the Borel-Cantelli lemma implies that, with probability
one, we have $\mn \geq (1-\tilde{\delta}_n) \barMn$ for $n$ large enough.
Therefore, 
since $\rho_n < 1$ when $n$ is large enough,
\black
we realize that with probability one we have
\begin{align*}
  \liminf_{n \rightarrow \infty} \frac{\mn - \barMn}{\barmn} \log \frac{1}{\rho_n} &\geq \liminf_{n \rightarrow \infty} - \frac{\tilde{\delta}_n \barMn}{\barmn} \log \frac{1}{\rho_n} \\
                                                                                   &= \liminf_{n \rightarrow \infty} - n^{-1/8} \rho_n^{1/8} \sqrt{\frac{\barMn}{\barmn}} \log \frac{1}{\rho_n} \\
                                                                                   &= - \limsup_{n \rightarrow \infty} n^{-1/8} \rho_n^{1/8} \sqrt{\frac{\barMn}{\barmn}} \log \frac{1}{\rho_n}.
\end{align*}
Recall from~\eqref{eq:limsup-Mn-mn-1-as} that $\limsup \barMn / \barmn \leq 1$
a.s.. On the other hand, $\rho_n \rightarrow 0$ and hence $\rho_n^{1/8} \log
\frac{1}{\rho_n} \rightarrow 0$. Consequently, we have
\begin{equation}
  \label{eq:liminf-mn-barMn-log-rho-n-zero-as}
  \liminf_{n \rightarrow \infty} \frac{\mn - \barMn}{\barmn} \log \frac{1}{\rho_n} \geq 0 \qquad \text{a.s.}.
\end{equation}
This together with~\eqref{eq:limsup-mn-barMn-log-rhon}
implies~\eqref{eq:limsup-mn-barMn-claim}, which together
with~\eqref{eq:limsup-barMn-barmn-claim} completes the proof.
\end{proof}

\begin{lem}
  \label{lem:graphon-log-n2-mn-mbn-1}
  We have
  \begin{equation*}
    \lim_{n \rightarrow \infty} \frac{\log \binom{\binom{n}{2}}{\mn} - \barmn \log \frac{1}{\rho_n}}{\barmn} = 1 \qquad \text{a.s.}.
  \end{equation*}
\end{lem}

\begin{proof}
  Note that
  \begin{equation}
    \label{eq:log-n2-mn-cn}
    \begin{aligned}
      \log \binom{\binom{n}{2}}{\mn} &= \log \frac{\binom{n}{2} \left( \binom{n}{2} - 1 \right)\dots \left( \binom{n}{2} - \mn + 1 \right)}{\mn!} \\
      &= \mn \log \left( \binom{n}{2} - c_n \right) - \log \mn!,
    \end{aligned}
  \end{equation}
  where $c_n \in \reals$ and $0 \leq c_n \leq \mn$.
  Using Stirling's approximation, we have $\log \mn! = \mn \log \mn - \mn +
  O(\log \mn)$. Observe that since $\mn \leq \binom{n}{2}$, we have $\log
  \mn = O(\log n)$. But $\barmn= \binom{n}{2} \rho_n = \frac{n-1}{2} n \rho_n =
  \omega(n)$ since $n \rho_n \rightarrow \infty$. Thereby $\log \mn =
  o(\barmn)$ and $\log \mn! = \mn \log \mn - \mn + o(\barmn)$. Using these in~\eqref{eq:log-n2-mn-cn} above, we get
  \begin{align*}
    \log \binom{\binom{n}{2}}{\mn} &= \mn \log \left( \binom{n}{2} - c_n \right) - \mn \log \mn + \mn + o(\barmn) \\
                                   &= \mn \log \binom{n}{2} + \mn \log \left( 1 - \frac{c_n}{\binom{n}{2}} \right) - \mn \log \mn + \mn + o(\barmn) \\
                                   &= \mn \log \frac{\binom{n}{2} \rho_n}{\mn \rho_n} + \mn + \mn \log \left( 1 - \frac{c_n}{\binom{n}{2}} \right) + o(\barmn) \\
                                   &= \mn \log \frac{\barmn}{\mn} + \mn \log \frac{1}{\rho_n} + \mn + \mn \log \left( 1 - \frac{c_n}{\binom{n}{2}} \right) + o(\barmn).
  \end{align*}
Consequently,
\begin{equation}
  \label{eq:n2-mn-mbn-cn-o1}
  \begin{aligned}
    \frac{\log \binom{\binom{n}{2}}{\mn} - \barmn \log \frac{1}{\rho_n}}{\barmn} &= \frac{\mn - \barmn}{\barmn} \log \frac{1}{\rho_n} + \frac{\mn}{\barmn} \log \frac{\barmn}{\mn} + \frac{\mn}{\barmn} \\
    &\qquad + \frac{\mn}{\barmn} \log \left( 1 - \frac{c_n}{\binom{n}{2}} \right) + o(1).
  \end{aligned}
\end{equation}
From Lemma~\ref{lem:mn-concentrates-around-barmn}, we know that
\begin{equation}
  \label{eq:mn-nmb-logrhon-0-restate-converse}
  \lim_{n \rightarrow \infty} \frac{\mn - \barmn}{\barmn} \log \frac{1}{\rho_n} = 0 \qquad \text{a.s.}.
\end{equation}
Also, from Lemma~\ref{lem:graphon-mn-asymp-mn-mbarn-1}, we have
\begin{equation}
  \label{eq:conv-mn-mbn-1}
  \lim_{n \rightarrow \infty} \frac{\mn}{\barmn} = 1 \qquad \text{a.s.},
\end{equation}
and
so
\black
\begin{equation}
  \label{eq:conv-mn-mbn-log-1}
  \lim_{n \rightarrow \infty} \frac{\mn}{\barmn} \log \frac{\barmn}{\mn} + \frac{\mn}{\barmn} = 1 \qquad \text{a.s.}.
\end{equation}
Moreover, recall that $0 \leq c_n \leq \mn$ and $\mn / \barmn \rightarrow 1$
a.s.. But $\barmn / \binom{n}{2} = \rho_n \rightarrow 0$. Hence, $c_n /
\binom{n}{2} \rightarrow 0$ a.s.. Combining this with~\eqref{eq:conv-mn-mbn-1},
we get
\begin{equation}
  \label{eq:conv-lim-mn-mbn-cn-0}
  \lim_{n \rightarrow \infty} \frac{\mn}{\barmn} \log \left( 1 - \frac{c_n}{\binom{n}{2}} \right) = 0 \qquad \text{a.s.}.
\end{equation}
Substituting \eqref{eq:mn-nmb-logrhon-0-restate-converse},
\eqref{eq:conv-mn-mbn-log-1}, and \eqref{eq:conv-lim-mn-mbn-cn-0} back
into~\eqref{eq:n2-mn-mbn-cn-o1} completes the proof.
\end{proof}

\begin{lem}
  \label{lem:limsup-ev-n2-mn-mbn-rhon-1}
  We have
    \begin{equation*}
    \limsup_{n \rightarrow \infty} \ev{\frac{\log \binom{\binom{n}{2}}{\mn} - \barmn \log \frac{1}{\rho_n}}{\barmn}} \leq 1.
  \end{equation*}
\end{lem}

\begin{proof}
  Using the inequality $\log \binom{r}{s} \leq s \log \frac{re}{s}$, we have
  \begin{equation}
    \label{eq:ev-n2-mn-upper-bound}
    \begin{aligned}
      \ev{\log \binom{\binom{n}{2}}{\mn}} &\leq \ev{\mn + \mn \log \frac{\binom{n}{2}}{\mn}} \\
      &= \ev{\mn} + \ev{\mn \log \frac{\binom{n}{2} \rho_n}{\mn \rho_n}} \\
      &= \ev{\mn} + \ev{\mn \log \frac{\barmn}{\mn}}  + \ev{\mn \log \frac{1}{\rho_n}} \\
      &= \ev{\mn} + \barmn \ev{\frac{\mn}{\barmn} \log \frac{\barmn}{\mn}}  + \ev{\mn \log \frac{1}{\rho_n}} \\
      &\leq \ev{\mn} + \barmn \ev{\frac{\mn}{\barmn}} \log \frac{1}{\ev{\frac{\mn}{\barmn}}}  + \ev{\mn} \log \frac{1}{\rho_n},
    \end{aligned}
  \end{equation}
  where the last inequality employs the concavity of the map $x \mapsto x \log 1
  / x$. Now, recalling the construction procedure of the $W$-random graph $\Gn$
  with target density $\rho_n$ from Section~\ref{sec:prelim-graphon}, we have
  \begin{equation*}
    \ev{\mn|X_{[1:n]}} = \sum_{1 \leq i < j \leq n} 1 \wedge \rho_n W(X_i, X_j).
  \end{equation*}
  \black
  Thereby,
  \begin{equation}
    \label{eq:ev-mn-mbarn-W-1-rhon}
    \begin{aligned}
      \ev{\mn} &= \binom{n}{2} \ev{1 \wedge \rho_n W(X,X')}\\
      &= \barmn \ev{W(X,X') \wedge \frac{1}{\rho_n}},
    \end{aligned}
  \end{equation}
  where $X$ and $X'$ are independent and have the same distribution as $X_i, 1
  \leq i \leq n$. This means that
  \begin{equation}
    \label{eq:ev-mn-less-mbarn}
    \ev{\mn} \leq \barmn \ev{W(X,\pp{X})} = \barmn,
  \end{equation}
  where the last equality uses the fact that $W$ is a normalized graphon. On the
  other hand, sending $n$ to infinity in~\eqref{eq:ev-mn-mbarn-W-1-rhon} and
  using the fact that $\rho_n \rightarrow 0$ as $n \rightarrow \infty$, we get
  \begin{equation}
    \label{eq:lim-ev-mn-mbarn-1}
    \lim_{n \rightarrow \infty} \ev{\frac{\mn}{\barmn}} = 1.
  \end{equation}
  Now, using~\eqref{eq:ev-mn-less-mbarn} together with~\eqref{eq:ev-n2-mn-upper-bound}, for $n$ large enough so that $\rho_n <
  1$, we have
  \begin{equation*}
    \ev{\frac{\log \binom{\binom{n}{2}}{\mn} - \barmn \log \frac{1}{\rho_n}}{\barmn}} \leq 1 + \ev{\frac{\mn}{\barmn}} \log \frac{1}{\ev{\frac{\mn}{\barmn}}} .
  \end{equation*}
  Finally, using~\eqref{eq:lim-ev-mn-mbarn-1} and sending $n$ to infinity
  completes the proof.
\end{proof}

\section{Proofs for Lemmas in Section~\ref{sec:graphon-analysis}}
\label{app:graphon-analysis-lem-proofs}

The proof of Lemma~\ref{lem:phi-properties}
is straightforward.
\black

\mathstart

\begin{proof}[Proof of Lemma~\ref{lem:hWn-converges-to-W}]
  From Theorem~\ref{thm:ganguli-W-random-graphon-convergence},
  we have $\rho(\Gn) / \rho_n
  \rightarrow \infty$ a.s., where  $\rho(\Gn) = 2\mn / n^2$. Also, recall that  $\rho_n =
  2\barmn / (n(n-1))$. Therefore, we have 
  \begin{equation}
    \label{eq:graphon-analysis-mn-barmn-1}
    \lim_{n \rightarrow \infty} \frac{\mn}{\barmn} = 1 \qquad \text{a.s.}.
  \end{equation}
  On the other hand, by assumption, we have $\mn_{\Delta_n} / \barmn \rightarrow 0$ a.s.. But
  $\mn = \mn_{\Delta_n} + \mn_*$. Comparing this
  with~\eqref{eq:graphon-analysis-mn-barmn-1} above, we realize that
  \begin{equation*}
    \lim_{n \rightarrow \infty} \frac{\mn_*}{\barmn} = 1 \qquad \text{a.s.}.
  \end{equation*}

  Since $\mn_* / \barmn
  \rightarrow 1$ a.s., with probability one, for $n$ large enough
  (for $n > n_0$ where $n_0$ can be random), we
  have $\frac{1}{e} \leq \frac{\mn_*}{\barmn} \leq e$. Using this, we realize that
  with probability one, for $n$ large enough, we have 
  \begin{align*}
    \left \lfloor  \log \frac{\mn_*}{n} \right \rfloor &= \left \lfloor  \log \frac{\barmn}{n} + \log \frac{\mn_*}{\barmn} \right \rfloor \\
                                          &\in \left \{ \left \lfloor  \log \frac{\barmn}{n} + x \right \rfloor: x \in [-1,1] \right \} \\
    &\subseteq \left \{ \left \lfloor \log \frac{\barmn}{n} \right \rfloor - 1, \left \lfloor \log \frac{\barmn}{n} \right \rfloor, \left \lfloor \log \frac{\barmn}{n} \right \rfloor + 1 \right \}.
  \end{align*}
  This means that with probability one, for $n$ large enough, we have $\alpha_n
  \in \{\alpha^{(i)}_n: 1 \leq i \leq 3\}$, where 
  \begin{equation*}
    \alpha^{(1)}_n := \frac{1}{e} \exp \left( \left\lfloor \log \frac{\barmn}{n} \right \rfloor \right)
    \qquad
    \alpha^{(2)}_n := \exp \left( \left\lfloor \log \frac{\barmn}{n} \right \rfloor \right)
    \qquad
    \alpha^{(3)}_n := e. \exp \left( \left\lfloor \log \frac{\barmn}{n} \right \rfloor \right),            
  \end{equation*}
  and $\alpha_n$ is defined 
  in~\eqref{eq:alphan-def}.
  \black
Consequently, if we define $\beta^{(i)}_n := \phi(\alpha^{(i)}_n)$ for $1 \leq i
\leq 3$, 
then
\black
with probability one we have $\beta_n \in \{ \beta^{(i)}_n: 1 \leq i
\leq 3 \}$ for $n$ large enough.

Note that for each $1 \leq i \leq 3$, $(\beta^{(i)}_n)_{n=1}^\infty$ is a
deterministic sequence. Furthermore, we claim that for each $1 \leq i \leq 3$,
we have the following
\begin{equation}
  \label{eq:betain-good}
  \lim_{n \rightarrow \infty} \beta^{(i)}_n = \infty \qquad \text{and} \qquad (\beta^{(i)}_n)^2 \log \beta^{(i)}_n = o( n \rho_n).
\end{equation}
To see this, note that $\barmn / n = (n-1) \rho_n / 2 \rightarrow \infty$  as $n
\rightarrow \infty$. Thereby, $\exp(\lfloor \log \barmn/ n \rfloor) \rightarrow
\infty$. This means that  $\alpha^{(i)}_n \rightarrow \infty$ and hence from
Lemma~\ref{lem:phi-properties}, $\beta^{(i)}_n =\phi(\alpha^{(i)}_n)\rightarrow
\infty$. On the other hand, since $\alpha^{(i)}_n \rightarrow \infty$, another
usage of Lemma~\ref{lem:phi-properties} implies that 
\begin{equation}
  \label{eq:lim-betain-alphain}
  \lim_{n \rightarrow \infty} \frac{(\beta^{(i)}_n)^2 \log \beta^{(i)}_n}{\alpha^{(i)}_n} = \lim_{n \rightarrow \infty} \frac{\phi^2(\alpha^{(i)}_n) \log \phi(\alpha^{(i)}_n)}{\alpha^{(i)}_n} = 0.
\end{equation}
But we have 
\begin{equation*}
  \alpha^{(i)}_n \geq \frac{1}{e} \exp(\lfloor \barmn / n \rfloor) \geq \frac{1}{e^2} \frac{\barmn}{n} = \frac{(n-1) \rho_n}{2e^2}.
\end{equation*}
This together with~\eqref{eq:lim-betain-alphain} implies that $(\beta^{(i)}_n)^2
\log \beta^{(i)}_n = o(n \rho_n)$. Consequently, we have
verified~\eqref{eq:betain-good}. 
As we discussed earlier, with probability one, for $n$ large, we have $\beta_n
\in \{\beta^{(i)}_n: 1 \leq i \leq 3\}$. This together
with~\eqref{eq:betain-good} implies that with probability one, $\beta_n
\rightarrow \infty$ and $\beta_n^2 \log \beta_n = o(n \rho_n)$.
On the other hand, if $\hWn_i$ is defined similar to
$\hWn$ based on solving the estimation
problem~\eqref{eq:least-sq-alg-better} with $\beta_n$ replaced by
$\beta^{(i)}_n$, from Theorem~3.1 in \cite{borgs2015consistent}, with
probability one, we have
 \begin{equation}
   \label{eq:delta2-hWn-i-0}
   \delta_2\left( \frac{1}{\rho_n} \hWn_i, W \right) \rightarrow 0.
 \end{equation}
but we have previously shown that with probability one, we have $\beta_n \in
\{\beta^{(i)}_n: 1 \leq i \leq 3\}$ for $n$ large enough. This means that with
probability one, for $n$ large enough, we have $\hWn \in \{\hWn_i: 1 \leq i \leq
3\}$. This together with~\eqref{eq:delta2-hWn-i-0} implies
that with probability one, $\delta_2 (\hWn / \rho_n, W) \rightarrow 0$ and
completes the proof.
\end{proof}

\mathfinish


\mathstart

\begin{proof}[Proof of Lemma~\ref{lem:delta-hWns-hWn-goes-to-zero}]
Recalling the definition of $\delta_2$ and employing the identity coupling, we
get
\begin{equation}
  \label{eq:delta-2-hWn-hWn*-bound-1}
  \delta_2(\hWn, \hWn_*) \leq \left( \sum_{i=1}^{\beta'_n} \sum_{j=1}^{\beta'_n} \frac{n_i}{n} \frac{n_j}{n} (\lambda_{i,j} - \lambda^*_{i,j})^2 \right)^{1/2}.
\end{equation}
On the other hand, using the facts that for $1 \leq i,j \leq \beta'_n$, we have
$\lambda_{i,j}\geq \lambda^*_{i,j}$ and $n_i \geq n / \beta_n$, we have
\begin{align*}
  \left( \sum_{i=1}^{\beta'_n} \sum_{j=1}^{\beta'_n} \frac{n_i n_j}{n^2} (\lambda_{i,j} - \lambda^*_{i,j})\right)^2 &\geq \sum_{i=1}^{\beta'_n} \sum_{j=1}^{\beta'_n} \left( \frac{n_i}{n} \frac{n_j}{n} \right)^2 (\lambda_{i,j} - \lambda^*_{i,j})^2  \\
  &\geq \frac{1}{\beta_n^2} \sum_{i=1}^{\beta'_n} \sum_{j=1}^{\beta'_n} \frac{n_i}{n} \frac{n_j}{n} (\lambda_{i,j} - \lambda^*_{i,j})^2.
\end{align*}
Comparing this with~\eqref{eq:delta-2-hWn-hWn*-bound-1}, we get
\begin{equation}
  \label{eq:delta-2-hWn-hWn*-bound-2}
  \delta_2(\hWn, \hWn_*) \leq \beta_n \sum_{i=1}^{\beta'_n} \sum_{j=1}^{\beta'_n} \frac{n_i n_j}{n^2} (\lambda_{i,j} - \lambda^*_{i,j}).
\end{equation}
Notice that
\begin{align*}
  \sum_{i=1}^{\beta'_n} \sum_{j=1}^{\beta'_n} \frac{n_in_j}{n^2} \lambda_{i,j} &= \sum_{i=1}^{\beta'_n} \frac{n_i^2}{n^2} \frac{2 m_{i,i}}{n_i^2} + 2 \sum_{1 \leq i < j \leq \beta'_n} \frac{n_in_j}{n^2} \frac{m_{i,j}}{n_i n_j} \\
                                                                         &= \frac{2}{n^2} \left( \sum_{i=1}^{\beta'_n} m_{i,i} + \sum_{1 \leq i < j \leq \beta'_n} m_{i,j} \right) \\
  &= \frac{2\mn}{n^2}.
\end{align*}
Similarly,
\begin{equation*}
  \sum_{i=1}^{\beta'_n} \sum_{j=1}^{\beta'_n} \frac{n_in_j}{n^2} \lambda^*_{i,j} = \frac{2\mn_*}{n^2}.
\end{equation*}
Substituting these into~\eqref{eq:delta-2-hWn-hWn*-bound-2}, we get
\begin{equation*}
  \delta_2(\hWn, \hWn_*) \leq \frac{2\beta_n (\mn - \mn_*)}{n^2} = \frac{2 \beta_n \mn_{\Delta_n}}{n^2} \stackrel{(*)}{\leq} \frac{2\beta_n n \Delta_n / 2}{n^2} = \frac{\beta_n \Delta_n}{n},
\end{equation*}
where in $(*)$, we have used the fact that in $\Gn_{\Delta_n}$, all the degrees
are bounded by $\Delta_n$. Consequently, we have
\begin{equation}
\label{eq:delta-2-hWn-hWn*-bound-3}
  \delta_2 \left( \frac{1}{\rho_n} \hWn, \frac{1}{\rho_n} \hWn_* \right) \leq \frac{\beta_n \Delta_n}{n \rho_n} = \frac{\beta_n}{\sqrt{n\rho_n}} \frac{\Delta_n}{\sqrt{n \rho_n}}.
\end{equation}
From Lemma~\ref{lem:hWn-converges-to-W}, with probability one we have $\beta_n^2
\log \beta_n = o(n\rho_n)$ and $\beta_n \rightarrow \infty$. Thereby, with
probability one, $\beta_n / \sqrt{n \rho_n} \rightarrow 0$. Furthermore,
by assumption, we have $\Delta_n / \sqrt{n \rho_n} \rightarrow 0$ a.s..
Substituting these into~\eqref{eq:delta-2-hWn-hWn*-bound-3} completes the proof.
\end{proof}

\mathfinish

\newcommand{\etalchar}[1]{$^{#1}$}

\end{document}